\documentclass[11pt]{article} 
\usepackage[hmargin=3cm,vmargin=3cm]{geometry}
\usepackage{pdfpages} 
\usepackage{epsfig}
\usepackage[english]{babel}
\usepackage[cmex10]{amsmath}
\usepackage[section]{algorithm}
\usepackage[section]{algorithm}
\usepackage{textcomp}
\usepackage{algpseudocode}
\usepackage{amssymb}
\usepackage{amsthm}
\usepackage{amsmath}
\usepackage{longtable}  
\usepackage{amsmath}
\usepackage{mathtools}
\usepackage{lmodern}
\usepackage{relsize}
\usepackage{wasysym} 
\usepackage{booktabs}  

\usepackage{enumitem}

\newtheorem{theorem}{Theorem}[section]
\newtheorem{lemma}[theorem]{Lemma}

\newtheorem{corollary}[theorem]{Corollary}
\newtheorem{definition}[theorem]{Definition}

\newtheorem{remark}[theorem]{Remark}
\newtheorem{claim}[theorem]{Claim}

\def\calT{{\mathcal{T}}}

\newcommand{\BO}{\mathcal{O}}
\newcommand{\MMHphi}{(1-\eps)}
\newcommand{\Ge}{{G_{1-\eps}}}  
\newcommand{\Ga}{{G_{1-2\eps}}}  
\newcommand{\Ger}{{G_{1-\eps}}}  
\newcommand{\Gar}{{G_{1-2\eps}}}  
\newcommand{\Ree}{{\R_{1-\eps}}}  
\newcommand{\Ra}{{\R_{1-2\eps}}}  

\newcommand{\Ee}{{E_{1-\eps}}}  
\newcommand{\R}{R} 		
\newcommand{\Ratio}{\Lambda}
\newcommand{\poly}{\text{poly }}

\newcommand*\polylogf[1]{\text{polylog}\left(#1\right)}
\newcommand*\eps{\varepsilon}
\newcommand*\fack{f_{ack}}
\newcommand*\fprog{f_{prog}}

\newcommand*\fCONS{f_{CONS}}
\newcommand*\fapprog{f_{approg}}

\newcommand*\epsprog{\eps_{prog}}
\newcommand*\epstask{\eps_{task}}
\newcommand*\epsMMB{\eps_{MMB}}
\newcommand*\epsSMB{\eps_{SMB}}
\newcommand*\epsCONS{\eps_{CONS}}
\newcommand*\epsapprog{\eps_{approg}}
\newcommand*\epsack{\eps_{ack}}

\newcommand{\SWn}{\log^*(\Ratio/\epsapprog)}
\newcommand{\SW}{c\SWn}

\def\fb{\textsc{FallBack}}
\def\lp{\textsc{LowPower}}
\def\succ{\textsc{Success}}

\def\Pro{{\mathbb{P}}}

\newcommand{\hPhi}{c\cdot 4^\Phi\cdot\SWn}
\makeatletter
\newcommand\footnoteref[1]{\protected@xdef\@thefnmark{\ref{#1}}\@footnotemark}
\makeatother

\makeatletter
\newcommand*{\StartNewContent}{%
    \let\OrigLabel\label%
    \let\OrigRef\ref%
    \renewcommand{\label}[1]{\OrigLabel{FULL:##1}}%
    \renewcommand{\ref}[1]{\OrigRef{FULL:##1}}%
    \renewcommand{\label@in@display}[1]{%
        \ifx\df@label\@empty\else
            \@amsmath@err{Multiple \string\label's:
                label '\df@label' will be lost}\@eha
        \fi
        \gdef\df@label{FULL:##1}%
    }%
}
\makeatother

\newif\iffull
\newif\ifshort

\begin{document}

\title{A Local Broadcast Layer for the SINR Network Model
}

\author{
 \makebox[.2\linewidth]{Magn\'us M. Halld\'orsson\thanks{Supported by Icelandic Research Fund grants 120032011 and 152679-051.}}\\
{\small mmh@ru.is}\\
Reykjavik University
\and
 \makebox[.2\linewidth]{Stephan Holzer\thanks{Supported by the following grants: AFOSR Contract Number FA9550-13-1-0042, NSF Award 0939370-CCF, NSF Award CCF-1217506, NSF Award CCF-AF-0937274.}} \\
{\small holzer@csail.mit.edu}\\
MIT
\and
\makebox[.2\linewidth]{Nancy Lynch\footnotemark[2]}\\
{\small lynch@csail.mit.edu}\\
MIT
}

\date{}

\begin{titlepage}
\maketitle
\thispagestyle{empty}

\begin{abstract}
We present the first algorithm that implements an \emph{abstract MAC 
(absMAC)} layer in the \emph{Signal-to-Interference-plus-Noise-Ratio 
(SINR)} wireless network model.
We first prove that efficient SINR implementations are not possible 
for the standard absMAC specification.
We modify that specification to an "approximate" version that better 
suits the SINR model. We give an efficient algorithm to implement the modified 
specification, and use it to derive efficient algorithms for 
higher-level problems of global broadcast and consensus.

In particular, we show that the absMAC \emph{progress} property has no 
efficient implementation in terms of the SINR strong connectivity 
graph $\Ge$, which contains edges between nodes of distance at most 
$(1-\eps)$ times the transmission range, where $\eps>0$ is a small constant that can be chosen by the user.
This progress property bounds the time until a node is guaranteed to 
receive \emph{some} message when at least one of its neighbors is 
transmitting.
To overcome this limitation, we introduce the slightly weaker notion 
of \emph{approximate progress} into the absMAC specification.
We provide a fast implementation of the modified specification, based 
on decomposing the algorithm of~\cite{daum2013broadcast} into local 
and global parts.
We analyze our algorithm in terms of local parameters such as node 
degrees, rather than global parameters such as the overall number of 
nodes.
A key contribution is our demonstration that such a local analysis is 
possible even in the presence of global interference.

Our absMAC algorithm leads to several new, efficient algorithms for 
solving higher-level problems in the SINR model.
Namely, by combining our algorithm with high-level algorithms 
from~\cite{DBLP:journals/adhoc/KhabbazianKKL14}, we obtain an improved (compared to~~\cite{daum2013broadcast}) 
algorithm for global single-message broadcast in the SINR model, and 
the first efficient algorithm for multi-message broadcast in that model.
We also derive the first efficient algorithm for network-wide 
consensus, using a result of~\cite{DBLP:conf/podc/Newport14}.
This work demonstrates that one can develop efficient algorithms for 
solving high-level problems in the SINR model, using graph-based 
algorithms over a local broadcast abstraction layer that hides the 
technicalities of the SINR platform such as global interference.
Our algorithms do not require bounds on the network size, nor the 
ability to measure signal strength, nor carrier sensing, nor 
synchronous wakeup.
\end{abstract}

\end{titlepage}

\fulltrue 

\section{Introduction}

Two active areas in Distributed Computing Theory are the  
attempts to understand wireless network algorithms in the  
\emph{Signal-to-Interference-plus-Noise-Ratio (SINR) model} and  
\emph{abstract Medium Access Control layers (absMAC)}.
\begin{itemize}
\item
The SINR model captures wireless networks in a more precise way than  
traditional graph-based models, taking into account the fact that
signal strength decays according to geometric rules and interference  
and does not simply stop at a certain border.
\item
Abstract MAC layers (a.k.a Local Broadcast Layers), express guarantees  
for local broadcast while hiding the complexities of managing message  
contention. These guarantees include message delivery latency bounds:  
an \emph{acknowledgment bound} on the time for a sender's message to  
be received by all neighbors, and a \emph{progress bound} on the time  
for a receiver to receive some message when at least one neighbor is  
sending.
\end{itemize}
In this paper we combine the strengths of both models by abstracting and modularizing broadcast with respect to global interference and decay via the SINR formula. This marks the start of a systematic study that simplifies the development of algorithms for the SINR model. At the same time we provide an example that modularizing and abstracting broadcast using MAC layers is beneficial and does not necessarily result in worse time-bounds than those of the broadcast algorithm being decomposed.

Traditionally, SINR platforms are quite complicated (compared to  
graph-based platforms), and consequently are very difficult
to use directly for designing and analyzing algorithms for higher-level problems\footnote{We refer by \textit{higher-level problems} to e.g. network-wide broadcast, consensus, or computing fast relaying-routes, max-flow and other problems whose solution requires a good understanding of lower-level problems. Here, we refer by \textit{lower-level problems} to e.g. achieving connectivity, minimizing schedules and capacity maximization, which are better understood by now.}.
We show how absMACs can help to mask their complexity  
and make algorithms easier to design.
This demonstrates the potential power of absMACs with  
respect to algorithm design for the SINR model.
During this process we point out and overcome inherent difficulties  
that at first glance seem to  
separate the MAC layers from the SINR model and other physical models.
These difficulties arise because absMACs are graph-based interference  
models, while physical models capture (global) interference by specific  
signal-propagation formulas.
Overcoming this mismatch is a key difficulty addressed in this work.

We tackle this mismatch by introducing the concept of  
\emph{approximate progress} into the absMAC specification and analysis.
The definition of approximate progress enables us to obtain a good  
implementation of an absMAC, which enables anyone to immediately transform  
generic algorithms designed for an absMAC into algorithms  
for the SINR model. The main observation that inspired the definition of approximate progress is a proof, that no SINR absMAC implementation is able to guarantee fast progress in an SINR-induced graph $G$, while fast progress can be guaranteed with respect to an approximation $\tilde{G}$ of $G$. Roughly speaking, as SINR-induced strong connectivity graphs are defined based on discs representing transmission ranges, we choose $\tilde{G}:=\Ga$ to approximate $G:=\Ge$ by making the disc a tiny bit smaller than in $G$.

This abstraction makes it easier to design algorithms for higher-level problems in the SINR model and has further benefits. One of the most intriguing properties of abstract MAC layers is their separation of global from local computation. This is beneficial in two ways.
On the one hand this separation allows us to expose useful SINR techniques in the simple setting of local broadcast. On the other hand this separation provides the basic structure to perform an analysis based on local parameters, such as the number of nodes in transmission/communication range and the distance-ratios between them, which is beneficial as pointed out in 
\iffull
Section~\ref{sec:demo}.
\fi
 Due to this, and the plug-and-play nature of the absMAC theory, we obtain a faster algorithms for global single-message broadcast than~\cite{daum2013broadcast} and fast algorithms for global multi-message broadcast and consensus in the SINR model. To achieve these results, we simply plug our absMAC implementation and bounds into the results of~\cite{DBLP:journals/adhoc/KhabbazianKKL14} and~\cite{DBLP:conf/podc/Newport14}.

\paragraph{Future Benefits of Abstract MAC Layers in the SINR Model.}
Many higher-level problems such as global broadcast, routing and reaching  
consensus are not yet well understood in the  
SINR model
and recently gained more attention~\cite{daum2013broadcast,even2012multi,DBLP:conf/tamc/HobbsWHYL12,jurdzinski2013distributed2,jurdzinski2013distributed,DBLP:conf/sirocco/YuHWTL12,DBLP:conf/infocom/YuHWYL13}.
Many of these problems in the algorithmic SINR can be attacked in a structured way by using and implementing absMACs that hide all complications arising from the SINR model and global interference. Using MAC layers, graph-based algorithms can be analyzed in the SINR model even without knowledge of the SINR model and might still lead to almost optimal algorithms as we demonstrate here.

\section{Contributions and Related Work}\label{sec:con}

We devote large parts of this article to prove theorems on implementing an absMAC in the SINR model and how to modify the absMAC specification to get better results. Based on these theorems we derive results on higher-level problems in the SINR model. 
\iffull
Table~\ref{tab:results} summarizes our algorithmic contributions. 
In the following $\Ge$ and $\Ga$ denote two versions of strong connectivity SINR-induced graphs. By $\Delta_\Ge$ and $\Delta_\Ga$ we denote their degree and by $D_\Ge$ and $D_\Ga$ their diameter. The network size is denoted by $n$ and the ratio of the minimum distance to the smallest distance between nodes connected by an edge in $\Ge$ is denoted by $\Ratio$. Parameter $\alpha$ denotes the path-loss exponent of the SINR model. We state more detailed definitions in Section~\ref{sec:model}.
\fi
\iffull\paragraph{Efficient implementation of acknowledgments. } \fi
Theorem~\ref{thm:ack} transfers Algorithm~1 of~\cite{halldorsson2012towards} and its analysis to implement fast acknowledgments of the absMAC and modifies it to use local parameters. 
\iffull
The resulting algorithm performs acknowledgments with probability at least $1-\epsack$ in time $\BO\left(\Delta_{\Ger} \log \left(\frac{\Ratio}{\epsack}\right) \ \ +\ \   \log(\Ratio)\log\left(\frac{\Ratio}{\epsack}\right)\right)$. 
\fi
\iffull
Remark~\ref{rem:ackopt}  
\fi
 provides a close lower bound.

\iffull\paragraph{Proof of impossibility of efficient progress. }\fi
Theorem~\ref{thm:fprogLB} shows that one cannot expect an efficient implementation of progress using the standard definition of absMAC. 
\iffull
In particular one cannot implement an absMAC in the SINR model that achieves progress in time $\Delta_\Ge$ or less. This is not much better than our bound on acknowledgments and therefore inefficient. This lower bound is even true when an optimal schedule for transmissions in the network is computed by a central entity that has full knowledge of all node positions and can choose arbitrary transmission powers for each node. In contrast, all algorithms presented here are fully distributed, use uniform transmission power and do not know the positions of nodes.
\fi

\iffull\paragraph{The notion of approximate progress.}\fi 
Achieving progress faster than acknowledgment is key to several algorithms designed for absMACs. Motivated by the above lower bound, we relax the notion of progress in the specification of an absMAC to \emph{approximate progress}. Definition~\ref{def:appr} introduces approximate progress with respect to an approximation (or some subgraph) of the graph in which local broadcast is performed. Although this new notion of approximate progress is weaker than the usual (single-graph) notion of progress, bounds on approximate progress turn out to be strong enough to yield, e.g., good bounds for global broadcast as long as $G$ is, e.g., connected---see Theorem~\ref{thm:combKuhn}. The introduction of approximate progress is the main conceptual contribution of this article.

\iffull
\paragraph{Efficient implementation of approximate progress. }
We propose an algorithm that implements approximate progress in time $\BO\left(\left(\log^\alpha(\Ratio) + \log^*\left(\frac{1}{\epsapprog}\right)\right)\log(\Ratio)\log\left(\frac{1}{\epsapprog}\right)\right)$ with probability at least $1-\epsapprog$, see Theorem~\ref{thm:approg}. This algorithm is a modification of the global single-message broadcast algorithm of~\cite{daum2013broadcast} to guarantee approximate progress in a local multi-message environment. This also makes this algorithm suitable for a localized analysis, which enables us to bound on approximate progress depending only on local parameters and the desired success probability. A key issue is that transmissions made below the MAC layer to implement its broadcast service might be highly unsuccessful due to being performed randomly and being prone to interference. Although the absMAC implementation is guaranteed to perform approximate progress with arbitrarily high probability guarantee $1-\epsapprog$ (specified by by the user), it is crucial to use very low probability guarantees below the MAC layer. Fast approximate progress for large values of $\epsapprog$ can only be achieved when this is reflected in probability guarantees below the MAC layer (e.g. by avoiding network wide union bounds, as these require w.h.p. guarantees). This is an important step towards the improved bounds on global broadcast stated below. To argue that despite constant success probability of transmissions during our constructions we can still achieve the desired probability guarantee $1-\epsapprog$ for correctness of approximate progress, we 1) argue that global interference from nodes that erroneously participate in the protocol due to previously unsuccessful transmissions does not affect local broadcast much, and 2) bound the local effects of previously unsuccessful transmissions by studying the probability of correct execution of the algorithm in a receiver's neighborhood. This analysis is arguably the main technical contribution of this article.
\fi

\iffull\paragraph{Global consensus, single-message and multi-message broadcast in the SINR model.}\fi
We immediately derive an algorithm for global consensus (CONS) in Corollary~\ref{cor:cons}  by combining our acknowledgment-bound with a result of~\cite{DBLP:conf/podc/Newport14}. 
\iffull 
CONS can be achieved with probability at least $1-\epsCONS$ in time $\BO\left(D_{\Ge}(\Delta_{\Ge}+\log(\Ratio))\log\left(\frac{n\Ratio}{\epsCONS}\right)\right)$.
\fi
 Section~\ref{sec:comb2} combines our absMAC implementation with results of~\cite{DBLP:journals/adhoc/KhabbazianKKL14} in a straightforward way to derive algorithms for global single-message broadcast (SMB) and global multi-message broadcast (MMB). 
\iffull
Global SMB can be performed in time $\BO\left(\left(D_\Ga\ \ +\ \ \log\left( \frac{n}{\epsSMB}\right)\right)\log^{\alpha+1}(\Ratio)\right)$ with probability at least $1-\epsMMB$.
Global MMB can be performed  with probability at least $1-\epsMMB$ in time $\BO\left(D_\Ga \log^{\alpha+1}(\Ratio) \ \ + \ \ 
k\left(\Delta_{\Ger}+\polylogf{\frac{nk\Ratio}{\epsMMB}}\right)\log \left(\frac{nk}{\epsMMB}\right)\right)$.
\fi
\\
{\begin{table}[ht]
\small
\center{
\iffull 
\renewcommand{\arraystretch}{1.8}
\fi
\begin{tabular}{lcc}
\bottomrule
Task/Bound			& Lower bound & Upper bound presented here \\
\toprule
$\fack$ 				& $\Delta_\Ge$ $^{(+)}$ & $\BO\left(\Delta_{\Ger}\cdot \log \left(\frac{\Ratio}{\epsack}\right) \ \ +\ \   \log(\Ratio)\log\left(\frac{\Ratio}{\epsack}\right)\right)$\\
$\fprog$				& $\Delta_\Ge^{(*)}$ & $\BO\left(\Delta_{\Ger}\cdot \log \left(\frac{\Ratio}{\epsack}\right) \ \ +\ \   \log(\Ratio)\log\left(\frac{\Ratio}{\epsack}\right)\right)$\\
$\fapprog$			&  -- & $\BO\left(\left(\log^\alpha(\Ratio)+ \log^*\left(\frac{1}{\epsapprog}\right)\right)\log(\Ratio)\log\left(\frac{1}{\epsapprog}\right)\right)$ 
\\
global SMB			& $\Omega\Big(D_\Ge\log \left(\frac{n}{D_\Ge}\right)$ & $\BO\left(\left(D_\Ga\ \ +\ \ \log\left( \frac{n}{\epsSMB}\right)\right)\log^{\alpha+1}(\Ratio)\right)^{\text{(\textdagger)}}$ \\
& \ \ \ \ \ \ \ \ \ \ \ \ $ \ \ \ \ +\ \ \log^2 (n)\Big)^{(\ddagger)}$
\\
global MMB			& $\Omega\Big(D_\Ge\log \left(\frac{n}{D_\Ge}\right)$ & $\BO$\footnotesize$\Big(D_\Ga \log^{\alpha+1}(\Ratio) \ \  + \ \ k\Delta_{\Ger}\log \left(\frac{nk}{\epsMMB}\right) $\ \ \ \ 
\\
& $ \ \ \  +\ \  k\log(n) + \log^2 (n)\Big)^{(\ddagger)}$  &  \footnotesize \ \ \ \ \ \  \ \ \  \ \ \ \ \ \ \ \ \ \  \ \ \ \ \ \ \  \ \ \  \ \ \ \ \ \ \ $\ \ + \ \ \polylogf{\frac{nk\Ratio}{\epsMMB}}\Big)
^{\text{(\textdagger)}}$ 
\\
global CONS& -- & $\BO\left(D_{\Ge}(\Delta_{\Ge}+\log(\Ratio))\log\left(\frac{n\Ratio}{\epsCONS}\right)\right)^{\text{(\textdagger)}}$
 \\
\midrule
\end{tabular}
\caption{\small Summary of algorithmic results, see Section~\ref{sec:model} for details on notation. The table compares our new upper bounds to (known and new) lower bounds. Known lower bounds are graph-based and transfer to our setting, as we use weaker assumptions. To compare graph-based lower bounds with our upper bounds, one might choose $\Ratio=n$ to account for possible high degree and choose $\epsSMB=\epsMMB=n^{-c}$ to achieve w.h.p. correctness. (*) Lower bound proven in this paper using absMAC assumptions of~\cite{DBLP:journals/adhoc/KhabbazianKKL14}. (\textdagger) Lower bounds require runtimes of global broadcast to depend on $n$ even though we perform a local analysis. $(\ddagger)$ Combinations of lower bounds of~\cite{ABLP91,DBLP:journals/corr/abs-1302-0264, KM98}
 for graph based models. 
(+) Trivial lower bound 
\iffull
(Remark~\ref{rem:ackopt}).
\fi
\label{tab:results}
}
}
\end{table}
}
\iffull
\begin{remark}
All assumptions that we make in the SINR model and in absMACs are listed in Section~\ref{sec:modelassumpt} and are mainly adapted from~\cite{daum2013broadcast} and~\cite{DBLP:conf/podc/JurdzinskiKRS14}. Our SINR-related assumptions are rather weak. We do neither require ability to measure signal strength, nor carrier sensing, nor synchronous wakeup nor knowledge of positions. We do assume, e.g., (arbitrary) bounds on the minimal physical distance between nodes and on the background noise (from which $\Ratio$ can be derived), as well as conditional wakeup.
\end{remark}
\fi

\subsection{Comparison of Algorithmic Results with Previous Work}\label{sec:SMBhighlights}
\iffull\paragraph{Global single-message broadcast.}\fi
 Table~\ref{tab:2} compares the runtime of our algorithm for global $SMB$ with previous work. Currently~\cite{daum2013broadcast} and~\cite{DBLP:conf/podc/JurdzinskiKRS14} provide the best implementations of global SMB in the SINR model (see the runtimes in Table~\ref{tab:2}). The result of~\cite{daum2013broadcast} is as good or better than~\cite{DBLP:conf/podc/JurdzinskiKRS14} in case $\log^{\alpha+1}(\Ratio)\leq \log (n)$ and vice versa. To make it possible to compare our result to theirs, we need to choose $\epsSMB=1/n^c$ such that global SMB is correct w.h.p.. Furthermore, we execute our algorithm with $\eps':=\eps/2$ instead of $\eps$, while algorithms in previous work are executed without changing $\eps$. This ensures that our bounds are stated in terms of the same parameter $D_\Ge$ rather than the possibly larger parameter $D_\Ga$. At the same time the choice of $\eps'$ affects the runtime only by a constant factor. This results in a runtime of our algorithm of $\BO\left((D_\Ge + \log (n))\log^{\alpha+1}(\Ratio)\right)$ in the strong connectivity graph $\Ge$. This improves over the algorithm presented in~\cite{daum2013broadcast} in the full range of all parameters, and improves in case of $\log^{\alpha+1}(\Ratio)  \leq \min(D_\Ge\log(n) , \log^2 (n))$ over the algorithm of~\cite{DBLP:conf/podc/JurdzinskiKRS14}. Note that compared to~\cite{DBLP:conf/podc/JurdzinskiKRS14} we (and~\cite{daum2013broadcast}) assume knowledge of a bound on $\Ratio$. The key-ingredient of this improvement is our localized analysis in combination with~\cite{DBLP:journals/adhoc/KhabbazianKKL14}.

\begin{table}[ht]
\center{
\iffull\renewcommand{\arraystretch}{1.4}\fi
\begin{tabular}{ccc}
\bottomrule
Article			& Runtime bound for global SMB& We improve this runtime in case of\\
\toprule
this & $\BO\left(\left(D_\Ge + \log\left(n\right)\right)\log^{\alpha+1}(\Ratio)\right)
$ & 
\\
\cite{daum2013broadcast} & 
$\BO\left(D_\Ge\log^{\alpha+1}(\Ratio)\log (n)\right)$ &  
all parameters and ranges
\\
\cite{DBLP:conf/podc/JurdzinskiKRS14} & $\BO\left(D_\Ge\log^2(n)\right)$ & 
$\log^{\alpha+1}(\Ratio) \  \leq  \ 
\min(D_\Ge\log(n)\ ,\ \log^2 (n))$
\\
\midrule
\end{tabular}
}
\caption{
Comparison of the runtime of our global SMB protocol with previous results.}\label{tab:2}
\end{table}

\iffull\paragraph{Global multi-message broadcast.}\fi
The algorithm for global MMB derived from~\cite{halldorsson2012towards} runs in $\BO((D_\Ge+k)(\Delta_\Ge\cdot\log n + \log^2n))$ time. Roughly speaking, our algorithm replaces the dependency on the potentially large multiplicative term $D_\Ge\Delta_\Ge$ by $D_\Ge$ up to polylog factors. 
Section~\ref{sec:rel} summarizes global MMB in related models.

\iffull
\paragraph{Global consensus.}
We are not aware of any previous work in the model we consider.
\fi
\iffull
\subsection{A Demonstration how Algorithms Benefit from Abstract MAC Layers}\label{sec:demo}
When abstract MAC layers were introduced to decompose global broadcast into local and global parts, the original goal was to understanding broadcast better and to achieve a general framework that can be used to state, implement and analyze new algorithms faster and simpler with respect to different models. A downside was that decomposing broadcast by adding a MAC layer might slow down performance. We demonstrate that the absMAC not only help to decompose the SINR-algorithm for global single-message broadcast of~\cite{daum2013broadcast} into a local and global layer, but can be used to improve performance in an organized way when the algorithms of the two layers are modified and put back together. The key insight is, that the MAC layer provides the basic structure for a localized analysis by decomposing broadcast into a local and a global part. We show that a local analysis is indeed possible despite global interference and SINR constraints. To achieve best results, we make our analysis dependent on 1) local parameters such as the degree of a node, and 2) the desired probabilities of success of local broadcast. Combined with the algorithm~\cite{DBLP:journals/adhoc/KhabbazianKKL14} for global single-message and multi-message broadcast (that assumes an absMAC implementation such as ours), this immediately implies improved algorithms as highlighted in Section~\ref{sec:SMBhighlights}.
\fi 


\iffull
\section{Related Work}\label{sec:rel}
\paragraph{Graph Based Wireless Networks.} This model was introduced by Chlamtac and
Kutten~\cite{chlamtac:1985}, who studied deterministic centralized broadcast. 
\textbf{Global SMB:} For the case where the topology is not known, Bar-Yehuda, Goldreich, and Itai (BGI)~\cite{yehuda} provided a simple, efficient and fully distributed method called Decay for local broadcast. Using this method they perform global SMB in $\BO(D\log n+\log^2n)$ rounds w.h.p.. Later Czumaj and Rytter~\cite{CR03} and Kowalski and Pelc~\cite{KP-PODC03} simultaneously and independently presented an algorithm that performs global SMB in time $\BO(D\log(n/D)+\log^2n)$, w.h.p.. While this sequence of upper bounds was published, a lower bound of $\Omega (\log^2 n)$ was established for constant diameter networks by Alon et al.~\cite{ABLP91} and a lower bound of $\Omega(D\log (n/D))$ was established by Kushilevitz and Mansour~\cite{KM98}. Therefore the upper bounds are tight. In case the topology is unknown but collision detection is available, Ghaffari et al. \cite{GHKpodc2013}, present show how to perform global SMB w.h.p. in time $\BO(D+\log^6n)$. In case the topology is known, a sequence of articles presented increasingly tighter upper bounds~\cite{CW91,GM03,EK05,GPX05,KP07}, where an algorithm for for global SMB in optimal time $\BO(D+\log^2 n)$ was presented by Kowalski and Pelc~\cite{KP07}. 
\textbf{Global MMB:} When collision detection is available, the sequence of work~\cite{DBLP:journals/siamcomp/Bar-YehudaII93, GHKpodc2013, DBLP:conf/podc/KhabbazianK11} led to an $\BO(D+k\log n+\log^2 n)$ round algorithm that performs global broadcast of $k$ messages w.h.p. assuming knowledge of the topology, which is due to Ghaffari et al.~\cite{GHKpodc2013}. When this assumption is removed, the runtime of~\cite{GHKpodc2013} increases slightly to $\BO(D+k\log n+\log^6n)$. Earlier, Ghaffari et al.~\cite{DBLP:journals/corr/abs-1302-0264} showed a lower bound of $\Omega(k\log n)$ for global MMB. 
\textbf{Global Consensus:} Peleg~\cite{DBLP:conf/icdcit/Peleg07} provided a good survey on consensus in wireless networks. Of particular interest is the work of Cholker et al.~\cite{DBLP:conf/podc/ChocklerDGNN05} and~\cite{DBLP:dblp_journals/tpds/AlekeishE12}.
Many of these lower bounds can be transferred to the SINR-model using SINR-induced graphs. We can use upper bounds to benchmark our algorithms.
\paragraph{Abstract MAC layer.} The abstract MAC layer model was recently proposed by Kuhn et al.~\cite{DBLP:journals/dc/KuhnLN11}. This model provides an alternative approach to the various graph-based models mentioned above with the goal of abstracting away low level issues with model uncertainty. The probabilistic abstract MAC layer was defined by Khabbazian et al.~\cite{DBLP:journals/adhoc/KhabbazianKKL14}. \textbf{Implementations of absMACs:} Basic implementations of a probabilistic absMAC were provided by Khabbazian et. al~\cite{DBLP:journals/adhoc/KhabbazianKKL14} using Decay, and by~\cite{DBLP:conf/dialm/KhabbazianKLMP11} using Analog Network Coding. \textbf{Applications of absMACs:} The first to study an advanced problem using the absMAC of~\cite{DBLP:journals/dc/KuhnLN11} were Cornejo et al.~\cite{cornejo2009neighbor,cornejo2014reliable}, who investigated neighbor discovery in a mobile ad hoc network environment. Global SMB and MMB broadcast were studied by~\cite{DBLP:journals/adhoc/KhabbazianKKL14} in probabilistic environments and by Ghaffari et al.~\cite{unreliableTechreport} in the presence of unreliable links. Newport~\cite{DBLP:conf/podc/Newport14} showed how to achieve fast consensus using absMAC implementations. Our paper makes applies the results of~\cite{DBLP:journals/adhoc/KhabbazianKKL14,DBLP:conf/podc/Newport14}.

\paragraph{SINR model.} Moscibroda and Wattenhofer~\cite{MoWa06} were the first to study worst-case analysis in the SINR model. They pointed the algorithmic and distributed computing community to this model that was studied by engineers for decades.\textbf{ Local broadcast:} Short time after this, Goussevskaia et al.~\cite{DBLP:conf/dialm/GoussevskaiaMW08} presented two randomized distributed protocols for local broadcast assuming uniform transmission power and asynchronous wakeup.
This was improved simultaneously and independently by Yu et al.~\cite{DBLP:conf/dcoss/YuHWL12} and Halldorsson and Mitra~\cite{halldorsson2012towards} by obtaining similar bounds while using weaker model assumptions that are similar to those assumptions that we use. Both stated an algorithm for local broadcast in $\BO(N_x\cdot\log(n) + \log^2(n))$, where $N_x$ is the contention in the transmission range of node $x$. In this paper we transform the latter result to be part of an implementation of a probabilistic absMAC that yields fast acknowledgments. We modify the analysis of~\cite{halldorsson2012towards} to use purely local parameters.
\textbf{Global MMB:} The above algorithms for local broadcast immediately imply algorithms with runtime $\BO((D_\Ge+k)(\Delta_\Ge\cdot\log(n+k) + \log^2(n+k)))$ for global MMB of $k$ messages w.h.p.. Scheideler et. al~\cite{DBLP:conf/mobihoc/ScheidelerRS08} consider a model with synchronous wakeup, uniform power and physical carrier sensing (that allows to differentiate signal strength corresponding to two thresholds). In this model they provide a randomized distributed algorithm that computes a constant density dominating set w.h.p. in $\BO(\log n)$ rounds. Such a sparsified set can be used to speed up global MMB by replacing the dependency on $\Delta_\Ge$ in the formula above by $\log n$. Yu et al.~\cite{DBLP:conf/sirocco/YuHWTL12,DBLP:conf/infocom/YuHWYL13} obtain almost optimal bounds using arbitrary power control. For a large range of the parameters their runtimes are better than the runtime of the algorithm that we provide. However, we point out that arbitrary power control is known to be almost arbitrarily more powerful for some problems than the uniform power restriction that we use~\cite{DBLP:conf/soda/Kesselheim11, MoWa06} such that we do not use this result as a benchmark. Power control was also used in~\cite{DBLP:conf/podc/BodlaenderHM13} to achieve connectivity and aggregation, which in turn can be used for broadcast as well.
\textbf{Global SMB:}  This problem recently caught increased attention and was studied in a sequence of papers~\cite{daum2013broadcast,jurdzinski2013distributed2, DBLP:conf/podc/JurdzinskiKRS14,jurdzinski2013distributed} using strong connectivity graphs $\Ge$. Jurdzinski, Kowalski et al. \cite{jurdzinski2013distributed2,jurdzinski2013distributed} considered a setting where nodes know their own positions. In~\cite{jurdzinski2013distributed2} they were able to present a distributed protocol that completes global broadcasts in the near-optimal time $\BO(D+\log (1/\delta))$ with probability at least $1-\delta$. In~\cite{jurdzinski2013distributed} they perform broadcasts within $\BO(D\log^2n)$ rounds.
Daum et al.~\cite{daum2013broadcast} propose a model that avoids the rather strong assumption that node's locations are known and does not use carrier sensing. However, they assume polynomial bounds on $n$ and $\Ratio$. Thanks to a completely new approach they show how to still perform global broadcast in $\Ge$ within $\BO(D \log^{\alpha+1}(\Ratio) \log(n))$ rounds w.h.p. using this weaker model. Their algorithm is based on a new definition of probabilistic SINR induced graphs combined with an iterative sparsification technique via MIS computation. We transfer and modify this algorithm to implement approximate progress in a probabilistic absMAC and provide a significantly extended analysis. Shortly after that, Jurdzinski et al.~\cite{DBLP:conf/podc/JurdzinskiKRS14} came up with a $\BO(D\log(n) + \log^2 n)$ algorithm that w.h.p. performs global broadcast independent of knowing $\Ratio$. However, to achieve this runtime they assume all nodes are awake and start the protocol at the same time. When assuming conditional wakeup, as~\cite{daum2013broadcast} and we do, their algorithm still requires only $\BO(D\log(n) + \log^2 n)$ rounds. Table~\ref{tab:2} compares these results to ours.
\textbf{Further work:} During the last years significant progress was made on lower-level problems that might provide useful tools for absMAC design, such as connectivity~\cite{PODC12}, minimizing schedules~\cite{HM11a}, and capacity maximization~\cite{halldorsson2013power,DBLP:conf/soda/Kesselheim11}. 
\fi

\section{Model and Definitions}\label{sec:model}

\iffull
We begin by defining basic notation for graphs, which we use 
throughout the paper.  Although the SINR model is not graph-based, we 
derive graphs from SINR models using reception zones. Abstract MAC 
layers are defined explicitly in terms of graphs. We continue by by describing the computational devices we use and recalling definitions of the SINR model, abstract MAC layers and global broadcast problems.
\fi

\iffull\subsection{Graphs and their Properties}\fi 
Let $G=(V,E)$ be a graph over $n$ nodes $V$ and edges $E$. We denote by $d_G(v,w)$ the hop-distance between $w$ and $v$ (the number of edges on a shortest $(u,v)$-path), and by $D_G:=\max_{u,v\in V} d_G(u,v)$ the \textit{diameter} of graph $G$.
All neighbors of $v$ in $G$ are called \textit{$G$-neighbors} of $v$. We 
\iffull
denote the direct neighborhood of $v$ in $G$ by $N_G(v)$. This includes $v$ itself. 
More formally we define
\fi
$N_G(v):=\{u | (v,u)\in E\}$ and extend this to $N_{G,r}(v):=\{u | d_G(v,u)\leq r\}$ for the $r$-neighborhood, $r\in\mathbb{N}$. For any set $W\subseteq V$ we generalize this to $N_{G,r}(W):=\bigcup_{w\in W} N_{G,r}(w)$. 
\iffull
The degree $\delta_G(v)$ of a node is the number of its (direct) neighbors in $G$, formally $\delta_G(v):=|N_G(v)|\setminus\{v\}$. We denote the maximum node degree of $G$ by $\Delta_G:=\max_{v\in V}\delta_G(v)$.
\fi
Let $S\subseteq V$ be a subset of $G$'s vertices, then $G|_S=(S,E|_S)$ denotes the subgraph of $G$ induced by nodes $S$, where $E|_S :=\{(u,v)\in E | u,v\in S\}$. A set $S\subseteq S'\subseteq V$ is called a \textit{maximal independent set (MIS)} of $S'$ in $G$ if 1) any two nodes $u,v\in S$ are independent, that is $(u,v)\notin E$, and 2) any node $v\in S'$ is covered by some neighbor in $S$, that is $N_G(v)\cap S\neq \emptyset$. 
\iffull
\begin{definition}[Growth bounded graphs]\label{def:growth}
A graph $G=(V,E)$ is \emph{(polynomial) growth-bounded} if there is a polynomial bounding function $f(r)$ such that for each node $v\in V$, the number of nodes in the neighborhood $N_{G,r}(v)$ that are in any independent set of $G$ is at most $f(r)$ for all $r\geq 0$. 
\end{definition}
\fi
\begin{lemma} Let $G$ be polynomially growth-bounded by function $f$, then it holds that $|N_{G,r}(v)|\leq \Delta f(r)$ for all $v\in V$ and $r\in\mathbb{N}$.
\end{lemma}
\begin{proof}
The proof 
\iffull
is deferred to Appendix~\ref{app:growth}.
\fi
\end{proof}

\iffull\subsection{The SINR Model}\label{sec:SINRmod}
The following describes the foundations of the \emph{physical model} (or \emph{SINR model}) of interference. We start by introducing a second distance function. 
\fi
Nodes are located in a plane and we write $d(v,w)$ for the Euclidean distance between points $v,w$ (often corresponding to node's positions). It is clear from the context when $d$ refers to hop-distance or Euclidean distance.
\iffull

\fi
When a node $v$ (of a wireless network) sends a message, it transmits with (uniform) power $P> 0$. A transmission of $v$ is received successfully at a node $u$, if and only if
\iffull
\begin{equation}
SINR_u(v):=\frac{P/d(v,u)^\alpha}{\sum_{w \in S \setminus  \{u,v\}}
P/d(w,u)^\alpha + N} \ge \beta ,\label{eq:SINR}
\end{equation}
\fi
where $N$ is a universal constant denoting the ambient noise. The parameter $\beta>1$ denotes the minimum
SINR (signal-to-interference-noise-ratio) required for a message to be successfully received,
$\alpha$ is the so-called path-loss constant. Typically it is assumed that 
 $\alpha\in(2,6]$, see~\cite{DBLP:conf/dialm/GoussevskaiaMW08}. 
Here, $S$ is the subset of nodes in $V$ that are sending.
\iffull
(All other nodes send with power $P'=0$). 
Independent of whether a distance function for the nodes is known, we assume in the analysis that the minimum distance between two nodes is $1$\footnote{Otherwise the SINR-formula implies that the power $P/d(v,u)^\alpha$ received by a node $u$ closer to a sender $v$ is higher than the power that $v$ uses to send.} (a.k.a near-field effect). This assumption can be justified by scaling length when assuming that two nodes cannot be at the same position (e.g. as each antenna's size is strictly larger than $0$).  
In this article we restrict attention to \emph{uniform power} assignments. All nodes $v\in S$ send with the same power $P_v=P$ for some constant $P$. 
\fi
By $\R:=(P/\beta N)^{1/\alpha}$ we denote the transmission range, i.e.~the maximum distance at which two nodes can communicate assuming no other nodes are sending at the same time. For $a\in \mathbb{R}^+$, we define $\R_a:=a\cdot \R$. 
If $d(v,u)\leq \R_a$ 
and $a<1$, we say $u$ and $v$ are connected by a \emph{$a$-strong} link. Like previous literature
we consider a link to be strong if it is $(1-\eps)$-strong  for constant $\eps>0$. 
\iffull
Intuitively, this means that the link uses at least slightly more power than the absolute minimum needed to overcome the ambient noise caused, e.g., by a few nodes sending far away. 
\fi
If $\R_a<d(u,v)\leq\R_{1}$, we say $u$ and $v$ are connected by an \emph{$a$-weak} link. A $(1-\eps)$-weak link is just called weak link.
\iffull
Strong connectivity is a reasonable and often used assumption~\cite{DBLP:conf/infocom/AndrewsD09,daum2013broadcast,Dinitz2010,DBLP:conf/dialm/GoussevskaiaMW08,KV10}.
\fi
\iffull
\subsection{SINR Induced Graphs}\label{par:SINRgraphs}
Like e.g. in~\cite{daum2013broadcast}, we consider \emph{strong connectivity broadcast} in this article while using uniform power. Therefore we 
\fi
consider the strong connectivity graph $\Ge=(V,\Ee)$, where $(u,v)\in \Ee$, if $u,v\in V$  are connected by a strong link. Given a graph $G$, we denote by $\Ratio_G$ the ratio between the maximum and minimum Euclidean length of an edge in $E$. In case that $G$ is $\Ge$, we simply write $\Ratio$ instead of $\Ratio_\Ge$. 
\iffull
\begin{remark}\label{rem:diff}
Note that~\cite{daum2013broadcast} uses a different but equivalent definition, while the above one is more common. In~\cite{daum2013broadcast} two nodes are connected in a graph if they are at distance at most $\R/(1+\rho)$, where $\rho$ takes the place of our $\eps$. When restating their lemmas, we simply use our notation without further comments as one can chose e.g. $\rho:=\eps/(1-\eps)$ or $\eps:=\rho/(1+\rho)$.
\end{remark}
\fi

\iffull\subsection{Abstract MAC Layers}\fi

\iffull
While there are several abstract MAC layer models~\cite{unreliableTechreport,DBLP:journals/adhoc/KhabbazianKKL14, DBLP:journals/dc/KuhnLN11}, the probabilistic version defined in~\cite{DBLP:journals/adhoc/KhabbazianKKL14} is most suitable for our purposes. Like any abstract MAC layer, the probabilistic MAC layer is defined for a graph $G=(V,E)$ and provides an acknowledged local broadcast primitive for communication in $G$. In our setting we are interested in strong connectivity broadcast with respect to the SINR formula, such that we use $\Ge$ as the communication graph (defined in Section~\ref{sec:SINRmod}). 

We use the definitions of Ghaffari et al.~\cite{unreliableTechreport} adapted to the probabilistic setting of~\cite{DBLP:journals/adhoc/KhabbazianKKL14}. To initiate such a broadcast, the MAC layer provides an interface to higher layers via input $bcast(m)_i$ for any node $i\in V$ and message $m\in M$.
\fi
To simplify the definition of this primitive,
assume w.l.o.g. that all local broadcast messages are unique. When a node $u\in V$ broadcasts a message $m$,
the model delivers the message to all neighbors in $E$. 
It then returns an acknowledgment of $m$ to $u$ indicating the broadcast is complete, denoted by $ack(m)_u$. In between it returns a $rcv(m)_v$ event for each node $v$ that received message $m$.
This model provides two timing bounds
\iffull
, defined with respect to two positive functions, $\fack$ and $\fprog$
which are fixed for each execution.
\fi
. The first is the acknowledgment bound, which guarantees that each broadcast will complete and be acknowledged within $\fack$ time. The second is the progress bound, which guarantees
\iffull
 the following slightly more complex condition
\fi
: fix some $(u,v)\in E$ and interval of length $\fprog$ throughout which $u$ is broadcasting a message $m$; during this interval $v$ must receive some message (though not necessarily $m$, but a message that some location is currently working on, not just some ancient message from the distant past).
The progress bound, in other words, bounds the time for a node to receive some message when at least one of
its neighbors is broadcasting. In both theory and practice $\fprog$ is typically much smaller than $\fack$~\cite{DBLP:journals/adhoc/KhabbazianKKL14}. Further motivation and power of these delay bounds is demonstrated e.g. in~\cite{unreliableTechreport,DBLP:journals/adhoc/KhabbazianKKL14, DBLP:journals/dc/KuhnLN11}.
\iffull

\fi
We emphasize that in abstract MAC layer models
the order of receive events is determined
non-deterministically by an arbitrary message scheduler. The timing of these events is also determined nondeterministically
by the scheduler, constrained only by the above time bounds.

\iffull\paragraph{The Standard Abstract MAC Layer. }\fi
Nodes are modeled as event-driven automata.
While~\cite{unreliableTechreport} assumes that an environment abstraction fires a wake-up event at each node at the beginning of each execution, we assume conditional wake-up to be consistent with the model of~\cite{daum2013broadcast}, see Definition~\ref{def:condwake}.
This is a weaker wake-up assumption with respect to upper bounds when compared to synchronous wake-up~\cite{unreliableTechreport}. This strengthens our algorithmic results. In contrast to this our lower bounds assume synchronized wake-up, which is in turn the weaker assumption with respect to lower bounds.
The environment is also responsible for any events specific to the problem being solved. In multi-message
broadcast, for example, the environment provides the broadcast messages to nodes at the beginning.

\begin{definition}[Conditional (a.k.a non-spontaneous) wake-up of~\cite{daum2013broadcast} adapted to absMACs]\label{def:condwake}
\iffull
Only after a node is woken up it can participate in computations below the MAC layer (i.e.~in the network layer). Communication needs to be scheduled from scratch. This corresponds to the conditional (non-spontaneous) wake up model.
\fi
\end{definition}

\iffull\paragraph{The Enhanced Abstract MAC Layer. }\fi 
The enhanced abstract MAC layer model differs from the standard model
in two ways. First, it allows nodes access to time (formally, they can set timers that trigger events when they
expire), and assumes nodes know $\fack$ and $\fprog$. 
Second, the model also provides nodes an abort interface that allows them to abort a broadcast in progress. 

\iffull\paragraph{The Probabilistic Abstract MAC Layer. }\fi 
We use parameters $\epsprog$ and $\epsack$ to indicate the error probabilities for satisfying the delay bounds $\fprog$ and $\fack$. Roughly speaking this means that the MAC layer guarantees that progress is made with probability $1-\epsprog$ within $\fprog$ time. With probability $1-\epsack$ the MAC layer correctly outputs an acknowledgment within $\fack$ time steps. More details can be found in Section 4.2 of~\cite{DBLP:journals/adhoc/KhabbazianKKL14}.  

\iffull\paragraph{Reliable Communication. }
Note that like in~\cite{DBLP:journals/adhoc/KhabbazianKKL14} all our communication graphs $G:=\Ge$ are static and undirected. In contrast to this,~\cite{unreliableTechreport, DBLP:journals/dc/KuhnLN11} defines not only a graph $G$ for guaranteed communication, but also a graph $G'$ for possible (but unreliable)  communication.
\fi

\iffull\subsection{Problems }\fi 
We derive algorithms in the SINR-model that perform the tasks listed below correctly with probability $1-\epstask$. When choosing $\epstask\leq n^{-c}$ we say that an algorithm performs a task with high probability (w.h.p.). Here, $c>0$ is an arbitrary constant provided to the algorithm as an input-parameter. We use the notation w.h.p. only to compare our results with previous work.

\iffull\paragraph{The Multi-Message Broadcast Problem (MMB)~\cite{DBLP:journals/adhoc/KhabbazianKKL14}.}\fi This problem inputs $k \geq 1$ messages into the network at the beginning of an execution, perhaps providing multiple messages to the same node. We assume $k$ is not known in advance. The problem is solved once every message $m$, starting at some node $u$, reaches every node in $G$. Note that we assume $G$ is connected to be consistent with previous work in the SINR model, while in~\cite{unreliableTechreport} this is not assumed. We treat messages as black boxes that cannot be combined.

\iffull\paragraph{The Single-Message Broadcast Problem (SMB)~\cite{DBLP:journals/adhoc/KhabbazianKKL14}.}\fi
The SMB problem is the special case of MMB with $k=1$. The single node at which the message is input is denoted by $i_0$.

\iffull\paragraph{The Consensus Problem (CONS), version considered in~\cite{DBLP:conf/podc/Newport14}.}\fi
In this problem each node begins an execution with an initial value from $\{0,1\}$. Every node has the ability to perform a single irrevocable $decide$ action for a value in $\{0,1\}$. To solve consensus, an algorithm must guarantee the following three properties: 1) $agreement$: no two nodes decide different values; 2) $validity$: if a node decides value $v$, then some node had $v$ as its initial value; and 3) $termination$: every non-faulty process eventually decides.

\iffull\subsection{General Model Assumptions}\label{sec:modelassumpt}\fi
\iffull
In principle each node has unlimited computational power. However, our algorithms perform only very simple and efficient computations. Finally, we assume that each node has private access to an unlimited perfect random source. This assumption can be weakened. 
\fi
As in~\cite{daum2013broadcast} wake up of nodes is conditional, see Definition~\ref{def:condwake}.
\iffull
From SINR-based work~\cite{daum2013broadcast} that we use we take the following assumptions:
\fi
Nodes are located in
the Euclidean plane\footnote{Our results can be generalized to any growth-bounded metric space when revising the assumption on $\alpha$.} and 
locations are unknown. Nodes send with uniform power, where the fixed power level $P$ is not known to the nodes. We use the common assumption that $\alpha>2$, see~\cite{DBLP:conf/dialm/GoussevskaiaMW08}. No collision detection mechanism is provided.
\iffull
Even when the same message is sent by (at least) two nodes and arrives with constructive interference we assume that the received signal cannot be distinguished from the case where no node is sending at all. 
\fi
As previous work we assume $\Ge$ is connected.
MAC-layer based work~\cite{DBLP:journals/adhoc/KhabbazianKKL14} requires us to assume that nodes can detect if a received message originates from a neighbor in a graph $G$--in our setting this is $\Ge$--(only one graph $G$ is used in~\cite{DBLP:journals/adhoc/KhabbazianKKL14}, while messages from any sender in the network might arrive but do not cause rcv-events).   

\iffull
\begin{remark}[Concerning SINR assumptions]\label{rem:SINR}
Although not explicitly stated in~\cite{daum2013broadcast}, footnote 5 in the full version~\cite{daumfull} of~\cite{daum2013broadcast} indicates that they use the assumption $\alpha>2$ as well. We also assume that $\alpha,\beta$ and $N$ are known. Note that~\cite{daum2013broadcast} allows $\alpha$ to be unknown, but fix in a known range $[\alpha_{\min},\alpha_{\max}]$. Furthermore they assume upper and lower bounds for $\beta$ and $N$ given by $\beta_{\min},N_{\min}$ and $\beta_{\max},N_{\max}$. For simplicity we do not make these assumptions, but claim our results can be stated in terms of these bounds as well.
Compared to~\cite{daum2013broadcast} nodes that execute local broadcast do not need to know a polynomial bound on the network size $n$. In~\cite{daum2013broadcast} this knowledge is only needed to achieve w.h.p. successful transmissions at each step. In our setting the desired probability of success is provided by the user of the absMAC. 
\end{remark}

\begin{remark}[Concerning absMAC assumptions]\label{rem:exact}
We want to remark that the assumption that nodes can detect if a received message originated in the $\Ge$-neighborhood is not used by any of the algorithms presented in this paper. In particular this assumption is not needed by previous algorithms on top of the MAC layer we use. This is due to the assumption that $\Ge$ is connected. However, being able to detect if a message was sent from a $\Ge$-neighbor might be required by future algorithms using our absMAC implementation that need broadcast to be implemented on exactly $\Ge$. Examples of this include $\Ge$-specific problems not studied in this article (such as e.g. computing shortest paths in $\Ge$). 
Note that a node $x$ executing our absMAC implementation might also successfully transmit messages to nodes that are not in $N_\Ge(x)$ but still in transmission range. 
For the reasons explained above we state our algorithms without the assumption that nodes can detect in which range a received message originated. If future algorithms using our absMAC implementation require exact broadcast, nodes executing our absMAC implementation could simply disregard messages they receive from nodes that are not their $\Ge$-neighbors (using the then necessary assumption that nodes can detect in which range a received message originated required to achieve exact local broadcast). 
If needed, there are several ways to implement this assumption. E.g. assuming that the SINR of the received message as well as the total received signal strength CCA can be measured. Using these assumptions there might also be faster SINR-implementations of an absMAC than provided in this article.
\end{remark}
\fi

\iffull
\subsection{Overview of Frequently used Notation}
For the convenience of the reader the following table summarizes notation used frequently (or globally) in this article. Definitions of notation not listed here are stated nearby where it is used.
{\small
\begin{longtable}{p{2.6cm} p{\textwidth-3.0cm}} 
\toprule
Notation & Explanation/Reference \\
\midrule
\textbf{Problems:}\\
SMB 																						& Global Single-Message Broadcast 
\\
MMB 																						& Global Multi-Message Broadcast
\\
CONS																						& Global Consensus
\\
\scriptsize $\epsSMB,\epsMMB,\epsCONS$					& Bounds on the probability that SMB,MMB and CONS are performed incorrectly 
\\
\midrule
\textbf{SINR related:}\\
$\alpha$																				& $\alpha\in(2,6]$, path-loss exponent in the SINR model
\\
$P,P_v$																					& Constant sending power of a node
\\
$d(u,v)$																				& Euclidean distance between nodes $u$ and $v$
\\
$I_S(v)$																				& Interference caused at point $v$ by nodes in set $S$ sending with power $P$ 
\\
$\eps$																					& Parameter that helps to define strong reachability, see $\Ree$
\\
$\R$																					& $\R:=(P/\beta N)^{1/\alpha}$, transmission range if no other node sends (weak reachability) 
\\
$\Ree, \Ra$																			& $\R_a:=a\cdot \R$, transmission distance tolerating interference from a sparse set of nodes. Tolerated sparsity of the set depends on $a$ (strong reachability)  
\\
$\Ratio$																				& Ratio of $\Ree$ to the shortest distance between any two nodes
\\
\midrule
\multicolumn{2}{l}{\textbf{Graph related:}}
\\
$G_1$																						& Nodes at distance at most $\R_1$ are connected (weak connectivity graph). In our algorithms communication in this graph might be unreliable 
\\
$\Ger, \Gar$																	& Nodes at distance at most $\Ree$/or $\Ra$ are connected in  $\Ger$/or $\Gar$ (strong connectivity graphs). We implement reliable local broadcast in $\Ger$ and analyze fast approximate progress with respect to $\Gar$
\\
$G|_S$																					& The subgraph $(S,E|_S)$ of $G=(V,E)$ induced by nodes in $S$ and $E|_S:=E\cap (S\times S)$
\\
$N_G(v)$, $N_G(W)$ 															& Neighborhoods of node $v$/set $W$ in graph $G$
\\
$N_{G,r}(v), N_{G,r}(W)$												& $r$-neighborhood of node $v$/set $W$ in graph $G$
\\
$\Delta_G$																			& Maximal degree of any node in graph $G$ 
\\
$D_G$																						& Diameter of Graph $G$
\\
$f$ 																						& Polynomial increasing function bounding the growth of graphs in this article
\\
\midrule
\multicolumn{2}{l}{\textbf{MAC layer related:}}
\\
$bcast, rcv, ack$																& Broadcast/receive/acknowledgment events in the MAC layer
\\
$\fack,\fprog,\fapprog$													& Bounds on the time needed for acknowledgment/progress/approximate progress
\\
$\epsack,\epsprog,\epsapprog$										& Bounds on the probability that acknowledgment/progress/approximate progress are not performed in time $\fack,\fprog,\fapprog$
\\	
$G'$																						& Graph with unreliable communication~\cite{unreliableTechreport}. Our setting considers $G':=G_1$
\\
$G$																							& Graph with reliable communication~\cite{unreliableTechreport}.  Our setting considers $G:=\Ger$
\\
$\tilde{G}$																				& Notation we introduce to denote subgraphs/approximations of $G$ in which we measure approximate progress. Our setting considers $\tilde{G}:=\Gar$
\\
\midrule
\multicolumn{2}{l}{\textbf{Algorithm~\ref{alg:AdHoc2} related:}}
\\
Phase $\phi$ 																		& Phases $\phi=1,\dots,\Phi$ are executed within an epoch
\\
Epoch				 																		& Each epoch
performs approximate progress with respect to $\Gar$
\\
$\Phi$																					& $\Phi=\Theta(\log\Ratio)$, number of phases $\phi$ executed in an epoch
\\
$Q$																							& $Q=\log^{\alpha} \Ratio$, parameter used to adjust transmission probabilities in Line~\ref{line:1}
\\
$c$																							& $c\in \mathbb{N}$ s.t. $c\log^* N$ bounds the MIS algorithm~\cite{DBLP:conf/podc/SchneiderW08}'s runtime on IDs $\in[1,N]$
\\
$p$																							& $p\in(0,1/2]$, transmission-probability of a node (whenever not specified otherwise)
\\
$\mu$																						& $\mu\in(0,p)$, reliability-probability of an edge (whenever not specified otherwise)
\\
$T$																							&$T=\Theta\left(\log\left(f(h_1)/\epsapprog\right)/(\gamma^2\mu)\right)$, transmission-repetitions (when not changed)
\\
$m, m'$				 																	& $m$ is the bcast-message to be broadcast, $m'$ is a bcast-message that is received
\\
$\gamma$																				& $\gamma\in(0,1)$, parameter used to define the approximation $\tilde{H}_p^\mu[S]$ of $H_p^\mu[S]$
\\
$h_\phi, h_\phi'$																& $h_\phi:=h'_\phi:=1$, $h'_\phi:=3h_{\phi+1}$ and $h_\phi:=h_{\phi}'+\SW+1$ for $1\leq \phi < \phi$
\\
$H_p^\mu[S]$																		& Graph defined in~\cite{daum2013broadcast}: $\mu$-reliable edges, nodes $\in S$ send with prob. $p$, see Section~\ref{sec:H-graphs}
\\
$\tilde{H}_p^\mu[S]$														& Graph computed in~\cite{daum2013broadcast} to $(1-\gamma)$-approximate $H_p^\mu[S]$ w.h.p., see Section~\ref{sec:H-graphs} 
\\
$\tilde{\tilde{H}}_p^\mu[S]$										& Graph computed in Section~\ref{sec:H-graphs}. Likely to $(1-\gamma)$-approximate $H_p^\mu[S]$ locally 
\\
$S_1$																						& Set of nodes with an ongoing broadcast at a given time
\\
$S_\phi$																				& Independent Set in $\tilde{\tilde{H}}_p^\mu[S_\phi]$ that is $(\phi,i)$-locally maximal with some probability
\\
$(\phi,i)$-locally															& $(\phi,i)$-locally maximal independent sets are defined in Definition~\ref{def:sucphase}
\\
$u_\phi$ 																				& Closest node $\in S_\phi$ to location $i$ in the proofs, see Lemma~\ref{lem:daum4.5} 
\\
$U_{\phi,i}$																		& $U_{\phi,i}:=N_{\Ger}(i)\cap S_\phi$, phase-$\phi$-senders at distance $\leq \R_{1-\eps}$, see Definition~\ref{def:6.5}
\\
$S_{\phi,i}$																& $S_{\phi,i}:=N_{\tilde{\tilde{H}}_p^\mu[S_\phi],h_\phi}(U_{\phi,i})$, phase-$\phi$-senders relevant to location $i$, see Definition~\ref{def:6.5}
\\
$S'_{\phi,i}$																		& $S'_{\phi,i}:=N_{\tilde{\tilde{H}}_p^\mu[S_\phi],h'_\phi}(U_{\phi,i})\subseteq S_{\phi,i}$, see Definition~\ref{def:6.5}
\\
$W$																							& Nodes that computed ``wrong'' neighbors in some $\tilde{\tilde{H}}_p^\mu[S_\phi]$, see Definition~\ref{def:W}\\
\\
\end{longtable}
}
\fi

\iffull
\section{Efficient Acknowledgments with an Application to Consensus}\label{sec:ackcons}
\fi

\begin{theorem}\label{thm:ack}
In the SINR model using the assumptions of Section~\ref{sec:modelassumpt}, acknowledgments of an absMAC can be implemented w.r.t. graph $\Ge$ with probability guarantee $1-\epsack$ in time $\fack=\BO\left(\Delta_{\Ge}\log \left(\frac{\Ratio}{\epsack}\right) +  \log(\Ratio)\log\left(\frac{\Ratio}{\epsack}\right)\right)$.
\end{theorem}

\iffull
\begin{remark}\label{rem:notes}
This section only focuses on implementation of broadcast. Reactions to inputs from the MAC layer, such as $bcast(m)_i$, are handled in Section~\ref{sec:kuhn2}. Section~\ref{sec:kuhn2} that presents the final implementation of our absMAC.
\end{remark}
\begin{proof}
The bound on $\fack$ can be derived by modifying Theorem 3 in~\cite{halldorsson2012towards} to local parameters. We do this in Appendix~\ref{app:MMHproof}. The bound follows when Theorem~\ref{appthm:ack-app} is applied with parameter $\tilde{N_x}:=4\Ratio^2$, which upper bounds the number of nodes in transmission range $\R_1$ and thus the local contention. We derive our claim as the actual contention $N_x$ is upper bounded by $\Delta_{\Ge}$. Note that the network only knows a polynomial bound on $\Ratio$, not on $N_x$ nor $\Delta_{\Ge}$, which in turn are estimated by Algorithm~\ref{alg1}. Furthermore one simply needs to modify Algorithm~\ref{alg1} to stop after $\fack$ rounds, as then the probability-guarantee is reached. 
Note that this behavior does not guarantee that no messages from nodes that are not $\Ge$-neighbors are received. See Remark~\ref{rem:exact} how exact local broadcast can be implemented.

\end{proof}

\begin{remark}\label{rem:ackopt}
As a node can only receive one message at a time, the degree $\Delta_\Ge$ of the network corresponds to the maximal contention at some node and is therefore a lower bound for $\fack$. Therefore the result on $\fack$ in this section is close to optimal. 
\end{remark}

\subsection{Application to Network-Wide Consensus in the SINR Model}\label{sec:cons}
Before implementing the absMAC specification in a formal way in Section~\ref{sec:kuhn2}, we now derive an algorithm for global consensus based on the bound for acknowledgments, see Theorem~\ref{thm:ack}. This serves as a first example to demonstrate the power of the absMAC theory when applied to the SINR world. This example is possible due to a result of~\cite{DBLP:conf/podc/Newport14} for achieving consensus using an absMAC using the fact that they analyze this problem in terms of $\fack$, while $\fprog$ does not appear in their runtime. Although the MAC layer used in~\cite{DBLP:conf/podc/Newport14} is deterministic, we can obtain a randomized algorithm that works correct with probability $1-\epsCONS$ by choosing $\epsprog$ and $\epsack$ to be at most $\frac{\epsCONS}{t(n)n^2}$, where $t(n)$ is the runtime of the algorithm using the MAC layer. We obtain the following Theorem. 

\begin{theorem}[Theorem 4.2 of~\cite{DBLP:conf/podc/Newport14} transferred to our setting]
The wPAXOS algorithm [of~\cite{DBLP:conf/podc/Newport14}] solves network-wide consensus in $\BO(D_{\Ge}\cdot \fack)$ time in the (probabilistic) absMAC
model (with $\epsack=\epsprog=\frac{1}{n^{4}\epsCONS}$) in any connected network topology $\Ge$ w.h.p., where nodes have unique ids and knowledge of network size. 
\end{theorem}
Plugging in the bounds on $\fack$ of Theorem \ref{thm:ack} we obtain:

\begin{corollary}[Theorem 4.2. of~\cite{DBLP:conf/podc/Newport14} transferred to our setting]\label{cor:cons}
Network-wide consensus can be solved with probability $1-\epsCONS$ in time 
$$\fCONS=\BO\left(D_{\Ge}(\Delta_{\Ge}+\log(\Ratio))\log\left(\frac{n\Ratio}{\epsCONS}\right)\right).$$
\end{corollary}
\section{Impossibility of Fast Progress using the SINR-Model}\label{sec:fprogLB}
\fi
Many algorithms that are implemented in an absMAC benefit from the fact that typically $\fprog$ is much smaller than $\fack$. Often it is the case that $\fprog = \BO(\polylogf{\fack})$.   
We show that for any implementation of the absMAC~\cite{DBLP:journals/adhoc/KhabbazianKKL14} for $\Ge$ in the SINR model such a difference of the runtime is impossible. 
\iffull
One can even not expect a bound on $\fprog$ that is much better than $\fack$.
\fi
As the bound on $\fack$ in Theorem~\ref{thm:ack} is close to our lower bound on $\fprog$, we conclude that this algorithm is an almost optimal implementation of absMAC in the SINR-model with respect to both $\fack$ and $\fprog$.

\begin{theorem}\label{thm:fprogLB}
For worst-case locations of points there is no implementation of the absMAC in the SINR model that provides local broadcast in $\Ge$ and achieves fast progress. In particular it holds that $\fprog\geq \Delta_\Ge$. This is true even for an optimal schedule computed by an (even central) entity that has unbounded computational power, has full knowledge as well as control of the network and can choose an arbitrary power assignment. 
\end{theorem}
\iffull
\begin{proof}
We first recall a slightly more formal definition of $\fprog$ given in~\cite{
DBLP:journals/adhoc/KhabbazianKKL14}, Section 4.1. (we can choose $G=G'$ in their notation, as we do not consider unreliable links): \emph{a $rcv(m)_j$ event can only be caused by a $bcast(m)_i$ event when the proximity condition $(i,j) \in E$ is satisfied. The progress bound guarantees, that a $rcv$ event occurs at $j$ within time $\fprog$ when some neighbor of $j$ is broadcasting some message.} However, we make use of the assumption made in~\cite{DBLP:journals/adhoc/KhabbazianKKL14} that nodes perform only $rcv$ events for messages they received from $\Ge$-neighbors and discard messages from other nodes in transmission range. (See Remark~\ref{rem:exact} for a discussion on this and why this assumption is later not needed for our upper bounds with respect to global broadcast implementations.)

The reader might like to consult Figure~\ref{fig:example1} while following our construction. For simplicity we consider the Euclidean setting and the implied distances. Consider $\Delta$ nodes $V:=\{v_1,\dots,v_\Delta\}$ placed equidistant on a line with distance $1$ between neighboring nodes. Let $\Ree:=10\Delta$, i.e.~parameters $N, P$ and $\beta$ are chosen such that the transmission range is $10\Delta$, and consider a second line parallel to the first line at distance $\Ree$ to the first line. Now assume $\Delta$ nodes $U:=\{u_1,\dots,u_\Delta\}$ are placed equidistant on this second line with distance $1$ between neighboring nodes. Each node in $V$ and $U$ has degree $\Delta_\Ge=\Delta$ in $\Ge$. Each node in $V$ has exactly one edge in $\Ge$ to a node in $U$ and vice versa. W.l.o.g. we assume $v_i$ is connected to $u_i$. Now assume that $bcast(m_v)_v$ events occurred for each $v\in V$ and a message $m_v\in M$. Due to SINR constraints, a transmission through edge $(v_i,u_i)$ is only successful when no other node in $V\cup U \setminus \{v_i\}$ is sending at the same time. Although $u_i$ receives a message from $u_i$ in case of a successful transmission, no node in $U\setminus \{u_i\}$ receives a message. As there are $\Delta_\Ge$ pairs $(v_i,u_i)$, there is a node $u_j$ that does not receive a message from a $\Ge$-neighbor during the first $\Delta_\Ge-1$ time slots. Because $u_j$ has a neighbor that is broadcasting, we conclude that $\fprog\geq \Delta_\Ge$ for any algorithm. 
\end{proof}
\fi
\begin{figure}[ht]
	\begin{center}
		\includegraphics[width=.9\textwidth]{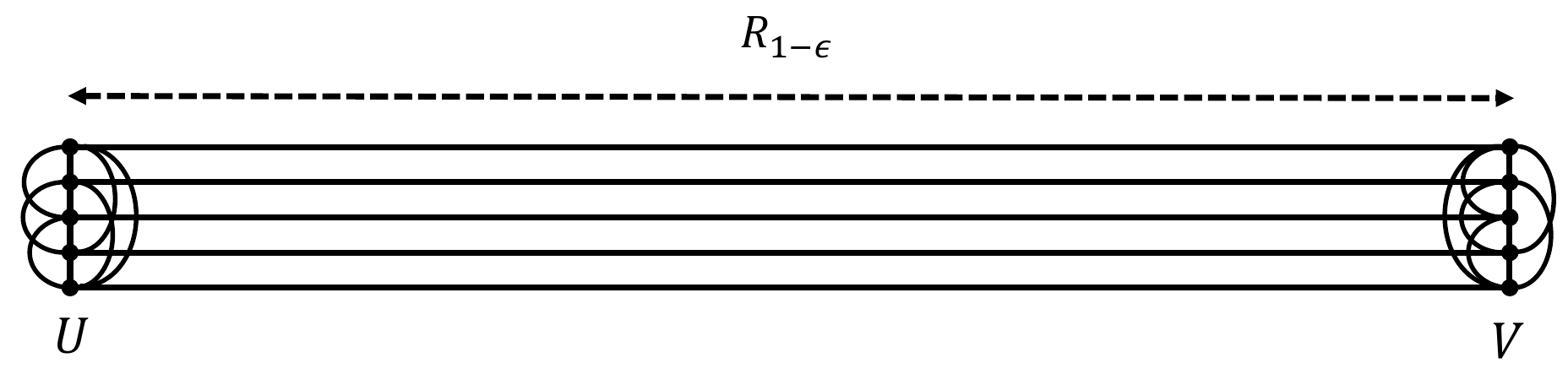}
	\end{center}
	\caption{Graph $\Ge$ based on the construction used in the proof of Theorem~\ref{thm:fprogLB}. Here we choose $\Delta=5$.}
	\label{fig:example1}
\end{figure}
\iffull
\begin{remark}[Comparison with previous lower bounds in the SINR model] Somewhat similar arguments were made earlier in~\cite{daum2013broadcast}. 
The same lower bound on $\fprog$ might also be derived by modifying the construction in the proof of Theorem 9 in~\cite{daum2013broadcast} to our setting when transferring it to the notion of $\fprog$. However, their lower bound is on the runtime for global single-message broadcast in (the weak connectivity) graph $G_{1}$ and therefore would need to be adapted. 
\end{remark}

\begin{remark}[Comparison with previous lower bounds in absMACs]
Note that in~\cite{DBLP:conf/wdag/GhaffariHLN12}, Theorems 7.1 and 7.2 also provide a lower bound of $\Omega(\Delta\log n)$ for $\fprog$ in the dual-graph model with unreliable links. The construction of their lower bound has a similar flavor to ours: Graph $G$ consists of $\Delta$ edges between nodes in $U$ and $V$ like in our example. Graph $G'$ is the complete bipartite graph over $U$ and $V$. Whenever one single node in $U$ is sending, an adversary prevents communication through unreliable edges (such that progress is only made at one node). When more than one node in $U$ is sending, the adversary allows communication through unreliable edges in a way that causes the lower bound. The difference to our model is, that in our setting the edges in $G'$ correspond to interference when a node incident to an edge is sending. This interference is fixed whenever a node is sending and cannot be switched on/off by an adversary. At the same time messages received via edges in $G'$ that are not in $G$ are discarded by the probabilistic MAC layer described in~\cite{DBLP:journals/adhoc/KhabbazianKKL14} such that no progress can be made using these edges. 
\end{remark}
\fi


\section{Approximate Progress}\label{sec:fapprog}

\iffull
Due to the lower bound on progress of Section~\ref{sec:fprogLB} we cannot expect $\fprog$ to be much better than $\fack$ in any SINR implementation. However, like in other wireless models it should take much less time until some message is received by a node $v$ (when several neighbors of $v$ are sending) compared to the time it takes until all neighbors of a sending node $u$ receive $u$'s message. The problem might be that the absMAC specification tries to measure progress in this physical model using a definition of progress that tried to capture the whole complication of the SINR model by a single graph. 
Motivated by this we try to capture a sense of progress by using two graphs and
\fi
modify the absMAC specification. An easy way would be to relax the progress bound and output a rcv-event not only for messages sent by $\Ge$-neighbors, but for all message received (i.e.~sent by any $G_1$ neighbor). This is problematic when considering randomized algorithms. In particular when computing e.g. overlay networks. It might happen that only $G_1\setminus \Ge$-neighbors of a node $v$ are chosen for the overlay due to the random event of low interference. This could of course be avoided by directly implementing the absMAC with respect to $G_1$ rather than $\Ge$, which in turn results in a $\Omega(n)$ lower bound for $\fprog$ and $\fack$ (e.g. when all nodes are located at distance at least $\R_1$ such that messages can only be received when exactly one node is sending). Later these overlay nodes might not be able to serve $v$.  
To avoid such a setting, we introduce an \emph{approximate} progress bound into the absMAC specification, where we use a graph $G$ and an \emph{approximation} (or any subgraph) $\tilde{G}$ of $G$ in which progress is measured. 

In the next sections we show that this generalization of progress has three desirable properties, it
\iffull
\begin{enumerate}
\item 
\fi
captures SINR behavior in the sense that we present an absMAC implementation in the SINR model that provides fast (approximate) progress, and
\iffull
\item 
\fi
replaces (with minor assumptions and effects) the progress bound in the runtime-analysis of e.g. global single-message and multi-message broadcast in the MAC layer~\cite{DBLP:journals/adhoc/KhabbazianKKL14}, and 
\iffull
\item 
\fi
does not affect the correctness of these algorithms.
\iffull
\end{enumerate}
\fi
Therefore we consider this notion of approximate progress to be a good modification of the specification of abstract MAC layers with respect to the SINR model.
\begin{definition}[Approximate progress]\label{def:appr}
Let there be (reliable\footnote{The notation of approximate progress might later be extended to unreliable broadcast~\cite{unreliableTechreport}.}) broadcast implemented with respect to a graph $G$ and let $\tilde{G}:=(V,\tilde{E})$ be a subgraph\footnote{Graph $\tilde{G}$ can be any subgraph of $G$ but will typically be an approximation of $G$, which results in the name approximate progress. Later we consider graph $\tilde{G}:=\Gar$, which approximates $G:=\Ger$ with respect to the SINR formula and Euclidean distances in the sense that it contains all, but the longest edges of $G$.} of $G$. Consider a node $i$ and assume that a $\tilde{G}$-neighbor of $i$ is broadcasting a message. The approximate progress bound guarantees that a $rcv$ event with a message originating in a $G$-neighbor occurs at node $i$ within time $\fapprog$ with probability $1-\epsapprog$. We say that approximate progress is implemented with respect to graphs $G$ and (its approximation) $\tilde{G}$.

We formalize this using the notation of~\cite{DBLP:journals/adhoc/KhabbazianKKL14}: Let $\beta$ be a closed execution that ends at time $t$. Let $I$ be the set of $\tilde{G}$-neighbors of $j$ that have active $bcasts$ at the end of $\beta$, where $bcast(m_i)_i$ is the $bcast$ at $i$. Suppose that $I$ is nonempty. Let $I'$ be the set of $G$-neighbors of $j$ that 
have active $bcasts$ at the end of $\beta$. Suppose that no $rcv(m_i)_j$ event occurs in $\beta$, for any $i\in I'$. Define the following sets $A$ and $B$ of time-unbounded executions that extend $\beta$.
\iffull
\begin{itemize}
\item $A$, the executions in which no $abort(m_i)_i$ occurs for any $i\in I$.
\item $B$, the executions in which, by time $t + \fapprog$, at least one of the 
following occurs:
\begin{enumerate}
\item An $ack(m_i)_i$ for every $i\in I$,
\item A $rcv(m_i)_j$ for some $i\in I'$, or
\item A $rcv_j$ for some message whose $bcast$ occurs after $\beta$.
\end{enumerate}
\end{itemize}
If $\Pro_\beta[A]>0$, then $\Pro_\beta[B|A] \geq 1 - \epsapprog$.
\fi
\end{definition}

This notation is useful, as there are settings where it is not crucial that progress is made with respect to exactly $G$. Already progress in subgraph $\tilde{G}$ might yield good overall bounds for solving a problem on $G$ especially when e.g. (depending on the problem at hand) $D_{\tilde{G}}\approx D_G$ or $\tilde{G}^2$. As we show in Theorem~\ref{thm:replace}, in the global SMB and MMB algorithms of~\cite{DBLP:journals/adhoc/KhabbazianKKL14} local broadcast does not need to be precise such that under some conditions progress can be replaced by approximate progress. 
\iffull
In the global broadcast algorithms of~\cite{DBLP:journals/adhoc/KhabbazianKKL14}, once a message is received by a node $i$, node $i$ broadcasts the message if it did not broadcast it before. The result of global broadcast is independent of whether a message was received due to transmission from a $\tilde{G}$-neighbor or a $G$-neighbor. However, one still needs to consider $\fack$ with respect to $G$. At the same time this obervation allows us to express large parts of runtimes of global broadcast algorithms~\cite{DBLP:journals/adhoc/KhabbazianKKL14} in terms of $D_{\tilde{G}}$ and $\fapprog$ instead of $D_{G}$ and $\fprog$. 
In graph based models, one could choose e.g. $G:=\tilde{G}^2$, the graph that is derived from $\tilde{G}$ when all paths of length at most $2$ in $\tilde{G}$ are replaced by edges. 
\fi
In the SINR model one might choose, e.g., $G:=\Ger\supseteq \Gar=:\tilde{G}$, as we do. This choice captures that any $\Ger$-neighbor is almost a $\Gar$-neighbor. In addition its signal has a similar strength when it arrives at the receiver and in reality might even be the same, as signal strengths can vary slightly. 
\iffull
The factor $2$ in $1-2\eps$ can be chosen arbitrarily small, but must be greater than 1.
\fi
\iffull
\begin{remark}\label{rem:fapprogrem} Like~\cite{unreliableTechreport} we study a dual-graph model. However, in our setting all communication is reliable. They consider a setting where reliable communication is provided in $G$ and communication which is unreliable in a non-deterministic way in $G'\supset G$. We extend this by a third graph $\tilde{G}\subset G$ such that unreliability can be studied in addition to approximate progress in future work. Note that approximate progress in $\tilde{G}$ inherits reliable communication from $G$. It is important to note that lower bounds on MMB of~\cite{unreliableTechreport} in their \emph{gray-zone model} with unreliable graphs do not apply in our setting. Their lower bounds are invalid, as we consider reliable broadcast in $G\supset\tilde{G}$ and assume that $G$ is connected. 
\end{remark}
\fi

\iffull
\section{Decay Fails to Yield Fast Approximate Progress}\label{sec:mod-decay}

Inspired by a proof of~\cite{daum2013broadcast}, we transfer this result into our setting. We show that using a (standard) \textsc{Decay} method, one cannot achieve fast approximate progress in the SINR model. In Section~\ref{sec:highlevel} we present an implementation based on an algorithm of~\cite{daum2013broadcast} that uses a different strategy than \textsc{Decay} achieves fast approximate progress and is analyzed in a more precise way.  
\begin{theorem}\label{lem:LB}
When using the \textsc{Decay} method of~\cite{yehuda} to implement local broadcast of a MAC layer in the SINR model, it holds that $\fapprog=\Omega(\Delta_\Ge\log(1/\epsapprog))$.
\end{theorem}
\begin{proof}In the (standard) \textsc{Decay} method of~\cite{yehuda} for graph-based models with collision detection, each node starts with sending probability $1$ and halves its transmission probability in each time slot until a sending probability is reached where no collision occurred for the first time. Then it keeps transmitting with this probability. This method can be applied in the SINR model as well (note that adding the assumption of collision detection yields a stronger lower bound than using our model assumptions). 

Consider two $\R_{1/4}$-balls whose centers are located at distance $\R_2$. Let ball $B_1$ contain $2$ nodes and let ball $B_2$ contain $\Delta_\Ge\leq n/2$ nodes. In the corresponding graph $\Ge$ the nodes located in different balls are not directly connected. We assume that the remaining $n/2-2$ nodes are arranged such that the nodes in the two balls are connected by a path of length $n/2-1$ in $\Ge$. Let's assume each of the nodes in $B_1$ and $B_2$ wants to broadcast a message and we perform a (standard) \textsc{Decay} mechanism. Once the probabilities reach a level where the nodes in $B_1$ are likely to transmit, the interference from nodes in $B_2$ is very strong. To be more formal, in round $i$, the probability that exactly one node in $B_1$ is sending is less than $2^{-(\log \Delta_\Ge-i)}$. The probability that no node in $B_2$ is sending is $(1-1/2^{\log (\Delta_\Ge)-i})^{\Delta_\Ge}\leq e^{-(\Delta_\Ge)/2^{\log (\Delta_\Ge)-i}}$. Thus the success probability of a node in $B_1$ is at most $e^{-(\Delta_\Ge)/2^{\log (\Delta_\Ge)-i}}/2^{-{\log (\Delta_\Ge)+i}}$. From this we conclude that the probability that a successful transmission takes place in $B_1$ within $\log (\Delta_\Ge)$ rounds with $i=1,\dots,\log (\Delta_\Ge)$ is less than
\begin{eqnarray*}
&&
\sum_{i=1}^{\log (\Delta_\Ge)} e^{-(\log (\Delta_\Ge))/2^{\log (\Delta_\Ge)-i}}/2^{-\log (\Delta_\Ge)+i}
\\
&=&
\sum_{i=\log (\Delta_\Ge) -\log\log (\Delta_\Ge)}^{\log (\Delta_\Ge)} e^{-\log (\Delta_\Ge)/2^i}/2^{-\log (\Delta_\Ge)+i}
\\
&&
 \ \ + \ \ 
\sum_{i=1}^{\log (\Delta_\Ge) - \log\log (\Delta_\Ge)} e^{-\log (\Delta_\Ge)/2^i}/2^{-\log (\Delta_\Ge)+i}
\end{eqnarray*}
We bound this by
\begin{eqnarray*}
&\leq&
e^{-\log (\Delta_\Ge)/2^{\log \log (\Delta_\Ge)}} \sum_{i=\log (\Delta_\Ge)-\log\log (\Delta_\Ge)}^{\log (\Delta_\Ge)} 1/2^{-\log (\Delta_\Ge)+i} 
\\
&& \ \ + \ \ 
 e^{-\log (\Delta_\Ge)/2^{\log (\Delta_\Ge)-\log\log (\Delta_\Ge)}}\sum_{i=1}^{\log (\Delta_\Ge)-\log\log (\Delta_\Ge)} 1/2^{\log (\Delta_\Ge)-i} 
\\
&\leq&  e^{-1/2}/2^{-\log (\Delta_\Ge)-\log\log (\Delta_\Ge) + 1} + 2 e^{-\log (\Delta_\Ge) +1}
\\
&\leq&
c \frac{\log (\Delta_\Ge)}{\Delta_\Ge} \text{ for some constant } c
\end{eqnarray*}

\noindent Therefore, $\left(c \frac{\log (\Delta_\Ge)}{\Delta_\Ge}\right)^{-1}\ln(1/\epsprog)$ repetitions of $\log (\Delta_\Ge)$ rounds $i=1,\dots,\log (\Delta_\Ge)$ are necessary such that the nodes in $B_1$ make progress with probability $\epsapprog$. We conclude that $\fapprog=\Omega(\Delta_\Ge\log (1/\epsprog)$. 
\end{proof}

Note that the authors of~\cite{daum2013broadcast} presented a lower bound of $\Omega(n)$ for SMB in $\Ge$ (Theorem 8 of~\cite{daum2013broadcast}), when using the \textsc{Decay} method of~\cite{yehuda}. This $\Omega(n)$ lower bound is of interest, as the construction of~\cite{daum2013broadcast} allows for an algorithm that needs only $\BO(1)$ rounds for SMB. Looking more closely at their lower bound, this can be interpreted as $\fapprog=\Omega(n)$ when $\epsapprog=n^{-c}$. We strengthen this lower bound for $\fapprog$ to $\fapprog=\Omega(n\log(1/\epsapprog))$. In the proof of this lower bound we use that the SINR model takes global interference into account (in contrast to graph based models). Also note that the proof of Theorem 8 of~\cite{daum2013broadcast} uses a network with maximal degree $\BO(n)$, and it can be easily generalized to yield $\fprog=\Omega(\Delta_\Ge)$ for arbitrary maximal degrees $\Delta_\Ge$.
\fi

\section{Implementation of Fast Approximate Progress}\label{sec:highlevel}
\iffull
Now 
\fi
we describe a method different from \textsc{Decay} 
\iffull
. Note that during an execution of the implementation additional messages from nodes that are $G_1$-neighbors but not $\Ge$-neighbors might occur in a probabilistic way. These do not affect our delay bound for approximate progress with respect to $\Ga$, as the analysis guarantees that messages from $\Ge$-neighbors arrive within time $\fapprog$. Note that Remark~\ref{rem:notes} applies to this Algorithm as well. This Algorithm~\ref{alg:AdHoc2} is described in this section and analyzed in Section~\ref{sec:anaI}.
\fi

\begin{theorem}\label{thm:approg}
In the SINR model using the assumptions of Section~\ref{sec:modelassumpt}, 
\iffull
Algorithm~\ref{alg:AdHoc2} implements 
\fi
 approximate progress of an absMAC with respect to graphs $\Ge$ and its approximation $\Ga$ with probability at least $1-\epsapprog$ in time 
 approximate progress of an absMAC with respect to graphs $\Ge$ and its approximation $\Ga$ with probability at least $1-\epsapprog$ in time 
$$\fapprog =\BO\left(\left(\log^\alpha(\Ratio) + \log^*\left(\frac{1}{\epsapprog}\right)\right)\log(\Ratio)\log\left(\frac{1}{\epsapprog}\right)\right).$$
\end{theorem}
\iffull
The algorithm presented by~\cite{daum2013broadcast} achieves global SMB in the strong connectivity graph $\Ge$. We modify this algorithm to fast (probabilistic) approximate progress with respect to $\Ga$.
In the algorithm of~\cite{daum2013broadcast}, after a node receives a bcast-message, it immediately forwards this (uniform) bcast-message. Inspired by this we implement this part in a similar style. However, we handle the possibility of multiple bcast-messages and need to guarantee that fast approximate progress can be proven. However, the modifications of their algorithm are substantial as described in this section to make it suitable for a localized analysis. In particular, in order to get an improved time-bound, we need to 1) introduce non-unique temporary labels instead of using unique IDs and handle this non-uniqueness, 2) acknowledge certain messages involved in coordination below the MAC layer, and 3) reduce the number $T$ of repeated transmissions such that $T$ is just large enough to guarantee low expected global interference from parts of the plane where computations went into a wrong direction based on communication-mistakes due to the reduced number $T$ of repetitions. The analysis in Section~\ref{sec:anaI} uses several Lemmas from~\cite{daum2013broadcast}. Whenever proofs of~\cite{daum2013broadcast} do not need to be changed significantly, we state versions of them adapted to our setting in Appendix~\ref{app:daum}. 
\fi

\iffull
\subsection{High-Level Description}
We start by presenting a high-level outline of the algorithm.
We follow the approach of~\cite{daum2013broadcast} and perform epochs, each consisting of $\fprog$ time steps. 
Each epoch corresponds to Lines~\ref{line:2}--\ref{line:3} of Algorithm~\ref{alg:AdHoc2}. 
During each epoch we compute approximations of a sequence of constant degree graphs $H_1, H_2,\dots, H_\Phi$, $\Phi=\Theta(\log\Ratio)$, used for communication. Graph $H_1$ is defined based on vertex set $S_1$, which is the set of nodes that have an ongoing broadcast at this time. As this set $S_1$ might change over time (depending on the algorithm using the absMAC and conditional wale-up, see Definition~\ref{def:condwake}), graphs $H_1, H_2,\dots, H_\Phi$ might be different in different epochs. Each $H_\phi$, $\phi>1$, is defined based on the nodes of a maximal independent set in $H_{\phi-1}$. For each $H_\phi$ it is guaranteed that, when each node in $H_\phi$ transmits with a certain (constant) probability $p\in (0,1/2]$, then for each edge $e$ of $H_\phi$ the transmission through $e$ is successful with a (constant) probability $\mu\in(0,p)$. 
Using geometric arguments we show in Lemma~\ref{lem:4} that when for $\Phi$ phases $\phi=1,\dots,\Phi$ during phase $\phi$ all nodes of graph $H_\phi$ transmit their message a certain amount of times, then approximate progress takes place in $\Ga$ within time $\fapprog=\BO\left(\left(\log^\alpha(\Ratio)+ \log^*\left(\frac{1}{\epsapprog}\right)\right)\log(\Ratio)\log\left(\frac{1}{\epsapprog}\right)\right)$. This happens with probability $1-\epsapprog$.

\paragraph{Intuition behind this algorithm:} Intuitively, this algorithm automatically adapts to regions of varying density. As the vertex set of graph $H_\phi$ is an MIS of $H_{\phi-1}$, it is typically a sparser version (with respect to density of nodes in the plane) of $H_{\phi-1}$. Finally we show that $H_\Phi$ is so sparse that nodes are too far away to communicate due to SINR constraints. Due to this sparsity, each node at this level is able to broadcast a message to its $\Ge$-neighbor with some probability. During this algorithm it will turn out that for each node $u \in N_\Ga(S_1)$, that has $\Ga$-neighbor with an ongoing broadcast, there is a $\Ge$-neighbor of $u$ in some $S_\phi$ from which $u$ receives a message in phase $\phi$. In particular, in phase $\phi$ the local density of nodes is reduced in a way that 1) there is still a node $u_\phi$ at distance at most $\Ree$, and 2) the density of nodes is so low that interference from these other nodes is low enough that $u_\phi$'s message reaches $u$ with some probability (due to random transmissions which further sparsify the set of transmitting nodes). We modify and extend the algorithm and analysis of~\cite{daum2013broadcast} and choose parameters of the algorithm to our benefit.

\paragraph{Suitability for localized analysis:} Thanks to the MAC layer that helps us to treat global and local parts of an algorithm separately in a structured way, we only need to provide an algorithm that ensures local approximate progress in order to implement this part of the MAC layer, while the algorithm of~\cite{daum2013broadcast} has to ensure global broadcast (and focuses on single-message broadcast, while we study multi-message broadcast). Note that therefore the authors of~\cite{daum2013broadcast} need to ensure that all their iterative computations of global approximations of communication graphs $H_\phi$ have the desired approximation-quality with high probability in $n$. Compared to this, we only need to make sure that for any point $i$ in space these graphs are local approximations with a certain probability. Therefore we require only a much lower probability and gain a speedup from this. In particular this probability only depends on the number of coordination-messages exchanged by those nodes that are locally involved in ensuring approximate progress of bcast-messages\footnote{We denote by \emph{bcast-message} any messages that contains information to be broadcast due to an bcast-event. By \emph{messages}, we refer to messages sent for coordination among the nodes.} that might reach $i$, as the sender is at distance at most $1-\eps$ to $i$.
\fi


\iffull
We make use of this locality aspect (in combination with the carefully chosen parameters) and perform a more careful and localized analysis extending the one of~\cite{daum2013broadcast}.

Naturally, some parts of our proof follow along the lines of the proof in~\cite{daum2013broadcast} or argue how their proofs can be adapted. However, they are significantly extended to derive the speedup from our modifications of their algorithm. In the end our detailed analysis of approximate progress yields faster global SMB than~\cite{daum2013broadcast}, see Section~\ref{sec:comb2}. 

\subsection{Graphs}\label{sec:H-graphs}
During the algorithm we consider a graph $H_p^\mu[S]$ that was defined in~\cite{daum2013broadcast}. This graph depends on a set of nodes $S$, a constant transmission probability $p\in(0,1/2]$ and a constant reliability parameter $\mu\in (0,p)$. The vertices of $H_p^\mu[S]$ are just the nodes in $S$. To define the edge set of $H_p^\mu[S]$, assume that each node in $S$ sends with probability $p$ and no node outside of $S$ (i.e.~in $V\setminus S$) is sending at the same time. Based on this assumption/experiment we define the edge set $E_p^\mu[S]$ to contain edge $(u,v)\in S\times S$ iff (i) $u$ receives a message from $v$ with probability at least $\mu$, and (ii) $v$ receives a message from $u$ with probability at least $\mu$.

As it is difficult to compute $H_p^\mu[S]$ in a distributed way (as pointed out in~\cite{daum2013broadcast}), the authors of~\cite{daum2013broadcast} compute a $(1-\gamma)$-approximation $\tilde{H}_p^\mu[S]=(S,\tilde{E}_p^\mu[S])$, where w.h.p. the following is true:
$$
E_p^\mu[S] \subseteq  \tilde{E}_p^\mu[S]
\subseteq E_p^{(1-\gamma)\mu}[S].
$$

To obtain a speedup, we do not compute graph $\tilde{H}_p^\mu[S]$, but define and compute a graph $\tilde{\tilde{H}}_p^\mu[S]=(S,\tilde{E}_p^\mu[S])$ that locally corresponds to $\tilde{H}_p^\mu[S]=(S,\tilde{E}_p^\mu[S])$ at each point $i$ with some probability much smaller than w.h.p (and demonstrate later that this is enough for our purposes). We postpone the precise formal definition of this locality to Definitions~\ref{def:6.5} and~\ref{def:suc}. There we define local correctness with respect to different sets $S_\phi$ together with other requirements for correct local computation during the algorithm. By postponing the definition, we avoid unnecessary general and therefore complicated notation. For now we only need to know that it should always be the case that for any node $v$ we desire that $N_{\tilde{\tilde{H}}_p^\mu[S]}(v)$ corresponds to neighbors of $v$ that would be present in a $(1-\gamma)$-approximation $\tilde{H}_p^\mu[S]$ of $H_p^\mu[S]$ as well. However, we typically consider a much larger neighborhood of $v$ and desire that the subgraph of $\tilde{\tilde{H}}_p^\mu[S]$ corresponding to this neighborhood matches the corresponding subgraph of a $(1-\gamma)$-approximation $\tilde{H}_p^\mu[S]$ of $H_p^\mu[S]$.

\subsection{Details of the Algorithm}
We propose the following algorithm that is executed by all nodes in $S_1$ and is inspired by~\cite{daum2013broadcast}, but has small modifications that yield substantial improvements when analyzed in detail. The algorithm consists of epochs that are continuously repeated and ensure approximate progress within each execution of a epoch. 
Like in~\cite{daum2013broadcast} we assume that all nodes get synchronized by other nodes when they wake up and join the algorithm at the beginning of the next epoch. A node $i$ wakes up either due to receiving a bcast-message from another node or due to the first $bcast$ event that occurred at node $i$. Whenever a $bcast(msg)_i$ event occurs, a variable $m$ stored in node $i$ is set to $msg$.

\begin{algorithm}[!htbp]
\begin{algorithmic}[1]
\Statex \textsc{Continuous execution of epochs:} 
\State $\Phi:=\Theta(\log \Ratio)$; $Q:=\Theta(\log^{\alpha}\Re)$;
\While{awake}
\State $S_1:=S$;
\If{$m\neq 0$} // ongoing broadcast of $m$
\State node $i$ marks itself as contained in set $S_1$;\label{line:16}
\For{$\phi=1,\dots,\Phi$}\label{line:2}
\If{$i\in S_\phi$}
\State Compute graph $\tilde{\tilde{H}}_p^\mu[S_\phi]$ and schedule $\tau_\phi$ as described in Section~\ref{sec:6.3.3};\label{line:compH}
\State Compute $S_{\phi+1}$ as described in Section~\ref{sec:6.3.4}; \label{line:compS}
\For{$\BO(Q\cdot\log (1/\epsapprog))$ rounds}\label{line:22}
\State \textbf{transmit} bcast-message $m$ with probability $p/Q$;\label{line:1}
\State // If not transmitting, listen for a bcast-message
\EndFor\label{line:23}
\EndIf
\EndFor\label{line:3}
\EndIf
\State $m'$:= (first) bcast-message received due to a transmission from another node in Line~\ref{line:1};
\State \textbf{output} $rcv(m')_i$~\label{line:28};
\EndWhile
\end{algorithmic}
\caption{Implementation of the part of absMAC that achieves fast approximate progress. As executed by a node $i$.}\label{alg:AdHoc2}
\end{algorithm}

Once a $bcast(msg)_i$ event occurs at node $i$ at time $t$, we say that node $i$ has an ongoing broadcast for $\fack/2$ time steps starting at time $t+1$. At the beginning of each epoch, each node $i$ marks itself as belonging to set $S_1$ if it has an ongoing broadcast (Line~\ref{line:16} of Algorithm~\ref{alg:AdHoc2}). Whenever a node $i$ receives a bcast-message $m'$ for the first time in an epoch, it delivers that bcast-message to its environment with a $rcv(m')_i$ output event. This behavior does not guarantee that no messages from nodes that are not $\Ge$-neighbors are received. See Remark~\ref{rem:exact} how exact local broadcast can be implemented.

Next in the epoch, a sequence of sets $S_1\supseteq S_2\supseteq\cdots\supseteq S_\Phi$ and corresponding graphs $\tilde{\tilde{H}}_p^\mu[S_1]\supseteq\dots \supseteq \tilde{\tilde{H}}_p^\mu[S_\Phi]$ are computed\footnote{These graphs were abbreviated by $H_\phi$ in the high-level description in Section~\ref{sec:highlevel}. We want to stress that the sets $S_1\supseteq S_2\supseteq \cdots\supseteq S_\Phi$ computed by our algorithm are likely to differ from the sets $S_1\supseteq S_2\supseteq \cdots\supseteq S_\Phi$ computed in~\cite{daum2013broadcast}, but might be the same with a very low probability.}. In this sequence $\tilde{\tilde{H}}_p^\mu[S_\phi]$ is a graph, that given any node $i\in S_1$, is likely to $(1-\gamma)$-approximate $H_p^\mu[S_\phi]$ in a certain neighborhood\footnote{\label{foot:4}The size of this neighborhood is specified later in the analysis and not relevant in the specification of the algorithm, see Definition~\ref{def:6.5}.} of $i$. Set $S_{\phi+1}$ is an independent set of $\tilde{\tilde{H}}_p^\mu[S_\phi]$ and is likely to be an MIS with respect to a certain neighborhood\footnoteref{foot:4} of $i$. Sections~\ref{sec:6.3.3} and~\ref{sec:6.3.4} describe in detail how graphs $\tilde{\tilde{H}}_p^\mu[S_\phi]$ and sets $S_{\phi+1}$ are computed.
While performing this computation, in each phase $\phi$ each node in $S_\phi$ transmits its respective bcast-message $m$ for $\BO(Q\cdot\log (1/\epsapprog))$ time steps, in each time step with probability $p/Q$ with $Q=\Theta(\log^\alpha \Ratio)$, see Lines~\ref{line:22}--\ref{line:23}. Denote by $m'$ the first bcast-message transmitted during Line~\ref{line:1} that node $i$ receives during an epoch. In Line~\ref{line:28} node $i$ outputs $rcv(m')_i$. 

\subsubsection{Computation of Graph $\tilde{\tilde{H}}_p^\mu[S_\phi]$ and Schedule $\tau_\phi$Based on $S_\phi$ in Line~\ref{line:compH}}\label{sec:6.3.3} 

We modify an algorithm described in~\cite{daum2013broadcast} to do this. In this algorithm we change the number of times $T$ that each message is sent. We define 

$$T:=\Theta\left(\frac{\log \left(\frac{f(h_1)}{\epsapprog}\right)}{\gamma^2\mu}\right)$$ 

where $h_1$ is defined in Definition~\ref{def:h} and $f$ is the function that bounds the growth of $G$, see Definition~\ref{def:growth}.

\begin{definition}\label{def:h}
For $\Phi=\Theta(\log\Ratio)$, we set $h_\Phi:=h'_\Phi:=1$ and define recursively $h'_\phi:=3h_{\phi+1}$ and $h_\phi:=h_{\phi}'+\SW+1$ for $1\leq \phi < \Phi$, where $c$ is chosen such that $\SW$ bounds the runtime of the MIS algorithm~\cite{DBLP:conf/podc/SchneiderW08} when applied on a network with node-IDs $\in[1,\frac{\poly \Ratio}{\epsapprog}]$.
\end{definition}

We restate the algorithm of~\cite{daum2013broadcast} with our modified parameter $T$ in order to perform our localized analysis later. All nodes in $S_\phi$ transmit their ID for $T$ rounds with probability $p$ in each round. Each node maintains a list of IDs that it received and counts how often each ID was received. Each ID that was received at least $(1-\gamma/2)\mu T$ times is a potential $\tilde{\tilde{H}}_p^\mu[S_\phi]$-neighbor. In another $T$ time slots, in each slot every node transmits all IDs of these $\BO(1)$ potential neighbors\footnote{\label{foot:3} Each node has at most $\frac{1}{(1-\gamma/2)\mu}=\BO(1)$ many potential neighbors (as remarked in~\cite{daum2013broadcast}).}, again with probability $p$ in each slot. A node $u\in S_\phi$ considers node $v\in S_\phi$ to be a $\tilde{\tilde{H}}_p^\mu[S_\phi]$-neighbor if $v$ is a potential neighbor of $u$ and $u$ appears in the list of potential neighbors of $v$ that $u$ received. 

Schedule $\tau_\phi$ keeps track of the nodes random choices to send depending on the time slot. That is $\tau_\phi$ maps time slot $t\in\{1,\dots,T\}$ to $\tau_\phi[t]\subseteq V$ of nodes that are sending in slot $t$.  

\subsubsection{Computation of Set $S_{\phi+1}$ Based on $\tilde{\tilde{H}}_p^\mu[S_\phi]$ and Schedule $\tau_\phi$ in Line~\ref{line:compS}}\label{sec:6.3.4} The authors of~\cite{daum2013broadcast} show how to simulate the MIS-algorithm of~\cite{DBLP:conf/podc/SchneiderW08} on $\tilde{\tilde{H}}_p^\mu[S_\phi]$ and define $S_{\phi+1}$ to be the computed MIS.
In order to perform a more localized analysis, we need to modify their approach. In particular the runtime of the deterministic MIS-algorithm~\cite{DBLP:conf/podc/SchneiderW08} depends on the range from which node IDs are chosen, not on the network size. While~\cite{daum2013broadcast} uses unique IDs $\in[1,\poly n]$, which results in a runtime of~$\BO(\log^* n)$, we desire a runtime that depends only on local parameters. 

In order to achieve such a runtime, we let each node $v\in S_\phi$ choose a temporary label $l_{i,\phi}\in\left[1,\frac{\poly \Ratio}{\epsapprog}\right]$ uniformly at random in each phase. Then we execute a modified version of the MIS-algorithm of~\cite{DBLP:conf/podc/SchneiderW08} for the CONGEST model using these temporary labels. As these labels might not be unique, we need to modify the algorithm of~\cite{DBLP:conf/podc/SchneiderW08}, as it might not terminate when non-unique labels are used. 

To state our modifications and being able to argue that this achieves the desired outcome, we review the algorithm of~\cite{DBLP:conf/podc/SchneiderW08}. After this, we present our modification of it in the CONGEST model. Subsequently we adapt the simulation of CONGEST algorithms in this probabilistic graph/SINR model given in~\cite{daum2013broadcast} to our modified parameters.

\paragraph{The MIS-algorithm of~\cite{DBLP:conf/podc/SchneiderW08}:} Each node starts in state $competitor$ and can change its state during the computation between states $\{competitor, ruler, ruled, dominator, dominated\}$. At the end of the algorithm, the set of all nodes in state $dominator$ is an MIS, and all other nodes are in state $dominated$ -- as shown in~\cite{DBLP:conf/podc/SchneiderW08}. To achieve this, the network executes a number of stages until all nodes are in state $dominator$ or $dominated$. At the beginning of each stage every node $v$ that is in state $competitor$ at that time sets a variable $r_v$ to its ID. After this, the stage performs $\log^*(N)+2$ phases, where $N$ indicates the range $[1,N]$ from which IDs are chosen. In each phase a node $v$ in state $competitor$ 1) exchanges $r_v$ with its neighbors, and 2) updates $r_v$ as well as its state depending on $r_v$ and the received $r_w, w\in N(v)$, from its neighbors. If in a phase $r_v<\min_{w\in N(v)\setminus\{v\}} r_w$ in that phase, node $v$ changes its state to $dominator$ and stays in that stage until the end of the algorithm. If $r_v=\min_{w\in N(v)\setminus\{v\}} r_w$, then $v$ changes its state to $ruler$. If $r_v>\min_{w\in N(v)\setminus\{v\}} r_w$, then $v$ updates $r_v$ depending on the bit where $r_v$ and $\min_{w\in N(v)\setminus\{v\}} r_w$ differ and might change its state to $dominated/ruled$ in case a neighbor changed its state to $dominator/ruler$. The proof of~\cite{DBLP:conf/podc/SchneiderW08} uses the fact that IDs in the network are unique to argue that after a constant number $c'$ of stages all nodes are in state $dominator$ or $dominated$ and nodes only terminate once they reached one of these states. 

\paragraph{Our modification of this algorithm in the CONGEST model:} We modify this algorithm to set $r_v:=l_{v,\phi}$ instead of using $v$'s ID at the beginning of each stage. As temporary labels $l_{v,\phi}$ are not unique, it can happen that after $c'$ many stages some nodes are neither in state $dominator$ nor $dominated$. Therefore we change the algorithm to terminate at a predetermined time (after $c'$ stages) instead of terminating at each node once it is in state $dominator$ or $dominated$. We still choose $S_{\phi+1}$ to consist only of nodes in state $dominator$ and ignore nodes not in state $dominator/dominated$.

\paragraph{Adapted simulation of CONGEST algorithms in our probabilistic graph/SINR model:}
Similar to~\cite{daum2013broadcast}, each round of communication in the CONGEST model is simulated by $T$ time steps in our model, where we use $T$ as defined above. In each time step $t\in \{1,\dots, T\}$ the messages (sent in a round of the algorithm for the CONGEST model) is sent by nodes $\tau_\phi[t]$, such that no messages are unsuccessful. 

In contrast to~\cite{daum2013broadcast}, our analysis requires that nodes know if their messages arrived at the destination. Such an acknowledgment can be implemented as node $i$ knows from which neighbors in $\tilde{\tilde{H}}_p^\mu[S_\phi]$ it should receive a message within time $T$ (as we just computed $\tilde{\tilde{H}}_p^\mu[S_\phi]$). We can acknowledge received messages by splitting each time slot into two slots, a transmission and an acknowledgment slot. This implies that the (reliability) probability of an acknowledged transmission is $\mu^2$. While w.h.p. communication is reliable in~\cite{daum2013broadcast}, it turns out that we cannot make these guarantees due to our choice of $T$. Therefore we say that communication at node $u\in S_\phi$ was unsuccessful (in phase $\phi$) when node $u$ did not receive messages (and acknowledgments for reception of its own messages) from all its $S_\phi$-neighbors within time $T$. Once communication was unsuccessful, a node $u\in S_\phi$ stops participating in this epoch and does not join $S_{\phi+1}$ in this epoch. A node $u\in S_\phi$ that stopped during the current epoch starts participating again in the next epoch as long as it has an ongoing broadcast. Messages received from nodes that are not $\tilde{\tilde{H}}_p^\mu[S_\phi]$-neighbors are ignored and not acknowledged.

\section{Analysis of our Implementation of Approximate Progress}\label{sec:anaI}

We start with an outline of the analysis in Section~\ref{subsec:ana} for the implementation of approximate progress of Section~\ref{sec:highlevel}. This is followed by Sections focusing on details of different issues mentioned in that outline.
\fi

\subsection{Outline of the Analysis}\label{subsec:ana}

We analyze the effect of the two main modifications of the algorithm of~\cite{daum2013broadcast} with respect to their analysis and put it into the context of approximate progress. 
\iffull
We outline the effects of these modifications here together with our approach before we dive into details in the next sections.

\subsubsection{First Modification: Non-Unique Labels in the MIS Computation}\label{sec:diff1}
\fi
\iffull
This difference is rooted in our modification of the MIS-algorithm of~\cite{DBLP:conf/podc/SchneiderW08} combined with using non-unique temporary labels $\in[1,\frac{\poly \Ratio}{\epsapprog}]$ instead of unique IDs $\in[1,\poly n]$ in~\cite{daum2013broadcast}. In Section~\ref{sec:locMIS} (Lemma~\ref{lem:locMIS}) we 
\fi
argue in the model of~\cite{DBLP:conf/podc/SchneiderW08} the sets $S_\phi$ computed by our modified MIS-algorithm are independent sets in $\tilde{H}_{\phi-1}$. Furthermore, for any given node $v$, with probability $1-\epsapprog/3$, this set is maximal in a neighborhood around $v$ ``large enough'' to ensure that this part of computations involved in approximate progress at node $v$ is correct.

\iffull
\subsubsection{Second Modification: Fewer Repetitions of Transmissions}\label{sec:diff2}
\fi

In the algorithm of~\cite{daum2013broadcast} each node sends every bcast-message $\BO(\log^{\alpha}(\Ratio)\log n)$ times, while we use only $\BO(\log^{\alpha}(\Ratio)\log(1/\epsapprog))$ repeated transmissions. This implies that~\cite{daum2013broadcast} can assume that all communication is successful at any point w.h.p.. For large $\epsapprog$ we only have weak probability guarantees for success of communication. One side-effect is that with very high probability the computed graphs $\tilde{\tilde{H}}_\phi$ are not the desired global approximations of graphs $H_\phi$. This in turn affects correctness of approximate progress and we need to analyze local and global implications caused by reducing the number of repeated transmissions.  
\begin{enumerate}[noitemsep]
\item \textbf{Global implications of unsuccessful transmissions:} Global interference might increase in the long term and we need to bound this. 
\iffull
Unsuccessful transmissions that are undetectable as the receiver does not know from which other nodes to expect messages can only appear 1) during the computation of $\tilde{\tilde{H}}_p^\mu[S_{\phi}]$, and 2) while transmitting the message in Line~\ref{line:1}. The latter will not cause increased global interference in the long term, as it does not influence the activity of nodes in future phases of the current epoch. Thus we only need to consider unsuccessful transmissions during the computation of $\tilde{\tilde{H}}_p^\mu[S_{\phi}]$. Consider a node $v$ and assume node $v$ has computed a wrong set of neighbors, that is a set of neighbors that does not correspond to a $(1-\gamma)$-approximation of $H_p^\mu[S_{\phi}]$. In such a case we just assume for the sake of worst case analysis of additional interference that $v$ joins the MIS $S_{\phi+1}$ of $\tilde{\tilde{H}}_p^\mu[S_{\phi}]$ -- regardless of where $v$ actually joins or not. Denote the set of all these nodes with wrong neighborhoods that unconsciously might cause additional interference during the current epoch by $W$. (Note that at the beginning of each epoch $W=\emptyset$, as no unsuccessful transmission happened yet.) We bound the additional (global) interference caused in case all nodes in $W$ erroneously decided to join $S_\phi$ in Lemma~\ref{lem:wronginterference} (regardless of which nodes in $W$ actually join $S_\phi$). Each time when we need to make an argument related to interference from nodes in $S_\phi$ in subsequent proofs, we also argue that the additional interference from nodes in $W$ is negligible compared to interference from a correctly computed $S_\phi$. Note that thus interference from nodes in $W$ might be counted twice (in particular nodes in $S_\phi \cap W$), but this does not hurt the analysis.
\fi
After $\tilde{\tilde{H}}_\phi$ is computed, all transmissions are successful. They use the same schedule used to compute $\tilde{\tilde{H}}_\phi$.

\item \textbf{Local implications of unsuccessful transmissions:}
\iffull Local communication of messages, which is based on the success of repeated transmissions, must be successful in 
a certain area\footnoteref{foot:4} around $v$ to ensure that 1) a node $v$ that has a broadcasting $\Ga$ neighbor receives a bcast-message from a broadcasting $\Ge$-neighbor in case all local computations are correct, and 2) we can transfer and extend tools from~\cite{daumfull} to our localized analysis.
This area in which this needs to be true contains all nodes possibly involved in the selection of a node from which $v$ might receive a bcast-message. 
\fi
These unsuccessful transmissions can only appear during the computation of $\tilde{\tilde{H}}_p^\mu[S_{\phi}]$ and while transmitting the 
\iffull
 bcast-message in Line~\ref{line:1}.
\fi
 Only if communication is locally successful, it is guaranteed that graph 
\iffull
$\tilde{\tilde{H}}_p^\mu[S_{\phi}]$
\fi
 is an 
\iffull
$(1-\gamma)$-
\fi
approximation of 
\iffull
$H_p^\mu[S_{\phi}]$ 
\fi
w.r.t. the above mentioned neighborhood of $v$, which is necessary in order to transfer the analysis of~\cite{daum2013broadcast}. We analyze the probability that 
\iffull
$\tilde{\tilde{H}}_p^\mu[S_{\phi}]$
\fi
is locally an approximation in 
\iffull
Lemma~\ref{lem:0.1}.  
\fi
Finally, approximate progress is made only if communication of bcast-messages succeeds locally.

\iffull
For all $\Ge$-neighbors $u'$ of $v$ (from which $v$ might receive a bcast-message), we lower bound the probability that all the above local computations/transmissions involved in the broadcast of $u'$ are successful in Lemme~\ref{lem:3}. 
\fi
\end{enumerate}

\iffull
\subsection{Local Effects of Non-Unique Labels}
\label{sec:locMIS}
\fi
We start by analyzing the effect of using (potentially) non-unique labels chosen uniformly at random $\in\left[1,\frac{\poly \Ratio}{\epsapprog}\right]$ in the modified MIS computation, which is the first difference to~\cite{daum2013broadcast}, as pointed out in Section~\ref{sec:diff1}.
\begin{lemma}\label{lem:locMIS}
Let $H=(V,E)$ be a constant degree growth-bounded graph and let $U\subseteq V$ be a set of nodes of size at most $\BO(\Ratio^2)$. Consider an execution of our modification of the MIS-algorithm of~\cite{DBLP:conf/podc/SchneiderW08} on $H$ in the CONGEST model using random labels $\in\left[1,\frac{\poly \Ratio}{\epsapprog}\right]$. Then the set of nodes in state $dominator$ is 1) an independent set, and 2) with probability at least $1-\frac{\epsapprog}{3\Phi}$ this set is maximal with respect to $N_{H,\hPhi}(U)$, the $\hPhi$-neighborhood of $U$ in $H$.
\end{lemma}
\iffull
\begin{proof}
From the description of the algorithm it follows that no neighboring nodes can be in state $dominator$. Therefore the set of $dominators$ remains an independent set despite our modification. 
 
Due to the analysis of~\cite{DBLP:conf/podc/SchneiderW08}, which uses that $H$ is growth bounded, we know that in case of unique IDs the algorithm computes an MIS within $c'<c$ stages. As the runtime of the algorithm is $\SW$, only the $\SW$-neighborhood of $N_{H,\hPhi+1}(U)$ is involved in deciding which nodes among $N_{H,\hPhi+1}(U)$ change their state to $dominator$. From this we can conclude that if nodes in $N_{H,\hPhi+1+\SW}(U)$ chose unique temporary labels, then the set of nodes in state $dominator$ is a maximal independent set with respect to $N_{H,\hPhi}(U)$. Note, that as we considered $dominators$ located in $N_{H,\hPhi+1}(U)$ for maximality in $N_{H,\hPhi}(U)$, it cannot happen that there is a node at the border of $N_{H,\hPhi}(U)$ that has no neighbor in state $dominator$. 

Now observe that it is $\hPhi+\SW=c\poly(\Ratio)\cdot\SWn$, as $\Phi=\Theta(\log\Ratio)$. As $H$ is growth bounded and has constant degree, this implies that there are at most $|U|\cdot\poly\Ratio=\poly\Ratio$ nodes involved in the state-changes of nodes in $N_{H,4^\Phi\cdot\SWn+1}(U)$. As we choose temporary labels from $\left[1,\frac{\poly \Ratio}{\epsapprog}\right]$, we can choose this range large enough such that with probability at least $1-\frac{\epsapprog}{3\Phi}$ the labels are unique among the $\poly\Ratio$ nodes in the $\hPhi+\SW$-neighborhood of $U$.
\end{proof}

\subsection{Global Effects of Unsuccessful Transmissions}
\fi
We analyze Case 1.a pointed out in Section~\ref{sec:diff2}, i.e.~we bound the global interference from nodes with undetectable unsuccessful transmissions.

\begin{definition}[Set $W$ of nodes with wrong neighborhoods (due to unsuccessful transmissions)]\label{def:W}
Denote by $W\subseteq S_1$ the set of all those nodes $v$ such that for at least one $\phi\in\{1,\cdots,\Phi\}$ it is not the case that $N_{H_p^\mu[S_{\phi}]}(v) \subseteq N_{\tilde{\tilde{H}}_p^\mu[S_{\phi}]}(v)\subseteq N_{H_p^{\mu(1-\gamma)}[S_{\phi}]}(v)$, i.e.~$v$'s direct neighborhood does not $(1-\gamma)$-approximate $N_{H_p^\mu[S_{\phi}]}(v)$.
\end{definition}

\begin{lemma}\label{lem:wronginterference} Given point $i$ in space, the expected total additional interference $I_W(i)$ that point $i$ receives from all nodes in $W$ at any given time is less than $\left(\frac{\epsapprog}{\Ratio}\right)^{\Theta(1)}$.
\end{lemma}
\iffull
\begin{proof}
We first bound the probability that $N_{\tilde{\tilde{H}}_p^\mu[S_{\phi}]}(v)$ does not correspond to a $(1-\gamma)$-approximation of $H_p^\mu[S_{\phi}]$ during a single phase. Let's consider a potential edge $(u,v)\in S_\phi \times S_\phi$. As each ID is transmitted $T$ times, a Chernoff bound implies that an edge $(u,v)$ is included in $\tilde{\tilde{H}}_p^\mu[S_\phi]$ if and only if $(u,v)$ belongs to a $(1-\gamma)$-approximation of $H_p^\mu[S_\phi]$ with probability at least 
\begin{eqnarray}
1-e^{-\Theta(T)}
=
1-e^{-\Theta\left(\log\frac{f(h_1)}{\epsapprog}\right)}
\geq
1-\left(\frac{\epsapprog}{f(h_1)}\right)^{\Theta(1)}
.\label{bd:1}
\end{eqnarray}
The constant hidden in the $\Theta$-notation depends on $\mu$, $\gamma$ and the Chernoff bound used. Recall\footnoteref{foot:3} that node $v$ has constant many neighbors in $H_p^\mu[S_{\phi}](v)$. Therefore the probability that $N_{\tilde{\tilde{H}}_p^\mu[S_{\phi}]}(v)$ is a $(1-\gamma)$-approximation of $N_{H_p^\mu[S_{\phi}]}(v)$ is at least $\left(1-\left(\frac{\epsapprog}{f(h_1)}\right)^{\Theta(1)}\right)^{\Theta(1)}\leq 1-\left(\frac{\epsapprog}{f(h_1)}\right)^{\Theta(1)}$.

From this we conclude, that in each square of size $\Ra$ times $\Ra$ the expected number of nodes that incorrectly have no edges in at least one of the $\Phi$ phases is at most 
$$\left(\frac{\epsapprog}{\Ratio}\right)^{\Theta(1)}\cdot\Ratio^2\cdot \Phi=\left(\frac{\epsapprog}{\Ratio}\right)^{\Theta(1)}.$$ 

Now assume that exactly the nodes in $W$ transmit at the same time. We use a standard argument from the SINR community to bound the expected interference that node $i$ receives from nodes in $W$ similar to the one in~\cite{DBLP:conf/dialm/GoussevskaiaMW08}. For the analysis we assume that the plane is partitioned into a $\Ra$-grid centered in $v$. 
 Denote by $A_d$ the set of grid-cells that contain nodes of $L_0$-distance at least $(d-1)\cdot\Ra$ and at most $d\cdot\Ra$ to $i$. Therefore $A_d$ contains $8d-4$ squares of size $\Ra$ times $\Ra$.

From this we conclude that the expected number of nodes in $A_d\cap W$ is upper bounded by $\BO\left(\left(\frac{\epsapprog}{\Ratio}\right)^{\Theta(1)}\cdot d\right)$. 

Each node in $A_d$ is at Euclidean distance at least $d-1$ to $i$, such that the interference caused at $i$ by a single node in $A_d$ sending with power $P$ is at most $P/(d-1)^{\alpha}$. Therefore the expected interference at point $i$ from nodes in $A_d\cap W$ is upper bounded by $\BO\left(\left(\frac{\epsapprog}{\Ratio}\right)^{\Theta(1)}/ d^{\alpha-1}\right)$, where we use that power $P$ is constant. Now we can upper bound the expected interference that point $i$ receives from $W$ by 
\begin{eqnarray*}
I_W(i)
&=&
\sum_{d=1}^\infty I_{A_d\cap W}(i) =  \sum_{d=1}^\infty \BO\left(\left(\frac{\epsapprog}{\Ratio}\right)^{\Theta(1)}/ d^{\alpha-1}\right)
\\
&=&
\BO\left(\left(\frac{\epsapprog}{\Ratio}\right)^{\Theta(1)}\right)=\left(\frac{\epsapprog}{\Ratio}\right)^{\Theta(1)},
\end{eqnarray*} 
where we use $\alpha>2$ in the second-last bound and the fact that $p$-series with $p>1$ converge to a constant.

Finally, note that each node actually sends only with probability $p$ (or $p/Q$) during the execution of each phase. Replacing the assumption that all node in $W$ transmit at the same time by these probabilities implies that the expected interference remains $\left(\frac{\epsapprog}{\Ratio}\right)^{\Theta(1)}$, as $p$ is constant and $Q=\BO(\log^{\alpha}(\Ratio))$. 
\end{proof}

\subsection{Local Effects of Unsuccessful Transmission}
\label{sec:correctness}

We analyze Case 2 pointed out in Section~\ref{sec:diff2}. We start with a bound on $h_1$ (see Definition~\ref{def:h}), define local success of an epoch (see Definition~\ref{def:suc}) and then analyze the probability of local success of an epoch. Lemma~\ref{lem:2} is the main Lemma of this section and states that for any set $S_1\subseteq V$ and node $i\in N_\Ga(S_1)$ two out of three properties of a successful epoch are satisfied with probability at least $1-\epsapprog/3$ at point $i$.

\begin{lemma}\label{lem:h-bound}
The following is true: $3^{\Phi-1}\leq h_1 \leq \hPhi$ for all parameters $\Phi, \Ratio, \epsapprog$ in the ranges considered in this paper. 
\end{lemma}
\begin{proof}
The first bound is immediate. To derive the second bound we show by induction on $\phi$ that $h_\phi\leq c\cdot 4^{\Phi-\phi}\cdot\SWn$. It is $h_\Phi=1\leq \SW$. For $\phi\leq \Phi$ assume that $h_\phi= c\cdot 4^{\Phi-\phi}\cdot\SWn$, then it is  
\begin{eqnarray*}
h_{\phi-1}
&=&
c\cdot 4^{\Phi-\phi}\cdot\SWn\cdot 3 + \SW
\\
&=&
c\cdot 4^{\Phi-\phi}\cdot\SWn\cdot 3 + c\cdot 4^{\Phi-\phi}\cdot\SWn
\\
&=&
c\cdot 4^{\Phi-(\phi-1)}\cdot\SWn
\end{eqnarray*}
\end{proof}

\subsubsection{Definition of Local Success of an Epoch}
Given point $i$, we define the sets of nodes involved in the local computation that selects a node from which $i$ might receive a bcast-messages in phase $\phi$. 
\begin{definition}[Sets $U_{\phi,i}, S_{\phi,i}$ and $S'_{\phi,i}$]\label{def:6.5}
Let $i\in N_\Ga(S_1)$ be a node (of which we can think as a point in space) that has a $\Ga$-neighbor with an ongoing broadcast. Let $U_{\phi,i}:=N_{\Ge}(i)\cap S_\phi$ be the subset of nodes at distance at most $\Ra$ from which $i$ might receive a bcast-message in phase $\phi$. We define sets $S_{\phi,i}$ and $S'_{\phi,i}$: 
\begin{itemize} 
\item $S_{\phi,i}:=N_{\tilde{\tilde{H}}_p^\mu[S_\phi],h_\phi}(U_{\phi,i})$, the $h_\phi$-hop $\tilde{\tilde{H}}_p^\mu[S_\phi]$-neighborhood of $U_{\phi}$. 
\item $S'_{\phi,i}:=N_{\tilde{\tilde{H}}_p^\mu[S_\phi],h'_\phi}(U_{\phi,i})\subseteq S_{\phi,i}$, the $h'_\phi$-hop $\tilde{\tilde{H}}_p^\mu[S_\phi]$-neighborhood of $U_{\phi,i}$.
\end{itemize}
\end{definition}

\begin{definition}\label{def:sucphase}(Local success of computing $\tilde{\tilde{H}}_p^\mu[S_\phi]$ and $S_{\phi+1}$).

 A computation of $\tilde{\tilde{H}}_p^\mu[S_\phi]$ is successful at point $i$ if $\tilde{\tilde{H}}_p^\mu[S_\phi]|_{S_{\phi,i}}$ corresponds to a $(1-\gamma)$-approximation of $H_p^\mu[S_\phi]|_{S_{\phi,i}}$. A computation of independent set $S_{\phi+1}$ on $\tilde{\tilde{H}}_p^\mu[S_\phi]$ is successful at node $i$ if $S_{\phi+1}$ is a $(\phi,i)$-locally maximal independent set in the sense that:
\begin{enumerate}
\item $S_{\phi+1}$ is independent in $\tilde{\tilde{H}}_p^\mu[S_\phi]$, and
\item there is no node $v \in S_\phi\setminus S_{\phi+1,i}$ such that $S_{\phi+1,i}\cup \{v\}$ is independent in $\tilde{\tilde{H}}_p^\mu[S_\phi]$ and $v$ is of distance at most $h_\phi$ to any $u\in U_{\phi+1,i}$ with respect to $\tilde{\tilde{H}}_p^\mu[S_{\phi+1}\cup \{v\}]$. 
\end{enumerate}
\end{definition}

 Adding a single node $v$ to the vertex-set $S_{\phi+1}$ might change the topology of $\tilde{\tilde{H}}_p^\mu[S_{\phi+1}]$, and thus distances in other parts of the graph due to SINR constraints. Therefore we need to show that this definition of $(\phi,i)$-local maximality is well-defined.

\begin{lemma}
The definition of $(\phi,i)$-local maximality is well-defined, i.e.~$(\phi,i)$-local maximality of set $S_{\phi+1}$ is invariant to adding a node $v$ that is independent to $S_{\phi+1}$ in $\tilde{\tilde{H}}_p^\mu[S_{\phi}]$.
\end{lemma}

\begin{proof}
 For any $u_1,u_2\in S_{\phi+1}$ the distance between $u_1$ and $u_2$ in $\tilde{\tilde{H}}_p^\mu[S_{\phi+1}\cup \{v\}]|_{S_{\phi+1,i}}$ compared to the distance in $\tilde{\tilde{H}}_p^\mu[S_{\phi+1}]|_{S_{\phi+1,i}}$ might potentially
\begin{enumerate}
\item decrease, as there might now be a shorter $u_1,u_2$-path via $v$, or  
\item increase, as $v$ adds interference, which might reduce connectivity among the nodes in $S_{\phi+1}$. Even though there might now be a short-cut for part of the paths via $v$, added interference might still cause the overall path to be longer.
\end{enumerate}
Therefore, in case that $v$ is not at distance at most $h_{\phi+1}$ to any node in a set $S_{\phi+1,i}$, set $S_{\phi+1,i}$ stays $(\phi,i)$-locally maximal independent of adding $v$, as the only way a node can be closer to $i$ is via a path through $v$---and $v$ in turn is at distance at least $h_{\phi+1}$. 
\end{proof}

\begin{definition}[Local success of an epoch]\label{def:suc}
An epoch is successful at point $i$ if  
\begin{enumerate}
\item the computations of each graph $\tilde{\tilde{H}}_p^\mu[S_1],\dots,\tilde{\tilde{H}}_p^\mu[S_\Phi]$ are successful at point $i$, and
\item the computations of each set $S_2,\dots,S_\Phi$  are successful at point $i$, and
\item there is a $\phi\in\{1,\dots,\Phi\}$, such that $i$ receives the bcast-message $m$ transmitted by some node $u_\phi\in U_{\phi,i}$ in Line~\ref{line:1} of phase $\phi$.
\end{enumerate}
\end{definition}

Note that in the proofs of this Section we never assume that we know the location of $i$ nor that we know $u_\phi$ or $u_\phi$'s location/distance to $i$.

\subsubsection{Probability of Local Success of Computing Graph $\tilde{\tilde{H}}_p^\mu[S_{\phi}]$ Based on $S_\phi$} 

This is an important step towards analyzing Property 1 of local success of an epoch. Using this property we later iteratively guarantee local success of computing graphs $\tilde{\tilde{H}}_p^\mu[S_{\phi}]$ in each phase. As the nodes involved in decisions of other nodes cannot be too far away, this helps to bound the probability that locally correct computations take place. If all involved computations at nodes in transmission range are successful, approximate progress takes place.

\begin{remark}\label{rem:unique} In the remaining part of Section~\ref{sec:correctness} we focus only on unsuccessful transmission to keep the analysis clean. Therefore we assume for now that $S_{\phi,i}$ is assigned unique temporary labels such that in absence of unsuccessful transmissions the modified MIS-algorithm always computes a set that is $(\phi+1,i)$-locally maximal. In Section~\ref{sec:kuhnproof} we argue that this assumption can be dropped at the cost of the probability derived in Lemma~\ref{lem:locMIS}.
\end{remark}

\begin{lemma}\label{lem:0.1}
Consider a node $i\in S_1$ and phase $\phi$ of Algorithm~\ref{alg:AdHoc2}. 
Line~\ref{line:compH} described in Section~\ref{sec:6.3.3} computes a graph $\tilde{\tilde{H}}_p^\mu[S_\phi]$ in time 
$\BO\left(\Phi + \log(1/\epsapprog)\right),$
such that with probability at least $1-\left(\frac{\epsapprog}{f(h_1)}\right)^{\Theta(1)}$ the computation of graph $\tilde{\tilde{H}}_p^\mu[S_\phi]$ is successful at node $i$. Given set $S_\phi$, the decision whether an edge $(u,v)$ is in graph $\tilde{\tilde{H}}_p^\mu[S_\phi]$ does not involve communication between nodes other than $u$ and $v$.
\end{lemma}
\begin{proof}

From Bound~\ref{bd:1} in Lemma~\ref{lem:wronginterference} we know that the probability that an edge $(u,v)\in S_\phi \times S_\phi$ that belongs to a $(1-\gamma)$-approximation of $H_p^\mu[S_\phi]$ is included in $\tilde{\tilde{H}}_p^\mu[S_\phi],h_\phi$ is at least $1-\left(\frac{\epsapprog}{f(h_1)}\right)^{\Theta(1)}$. As 
\begin{itemize}
\item $S_{\phi,i}$ is defined using the $h_\phi$-hop $\tilde{\tilde{H}}_p^\mu[S_\phi]$-neighborhood of nodes $U_{\phi,i}$ (see Definition~\ref{def:6.5}), and
\item $U_{\phi,i}$ contains at most $\Ratio^2$ many edges (see Definition~\ref{def:6.5}), and
\item $\tilde{\tilde{H}}_p^\mu[S_\phi]$ has degree\footnoteref{foot:3} $\BO(1)$, and $\tilde{\tilde{H}}_p^\mu[S_\phi]$ is growth bounded by $f$,
\end{itemize}
there are at most $\BO(f(h_\phi)\cdot \Ratio^2)$ edges in $\tilde{\tilde{H}}_p^\mu[S_\phi]|_{S_{\phi,i}}$ among which the algorithm needs to choose $\tilde{\tilde{E}}_p^\mu[S_\phi]\cap (S_{\phi,i}\times S_{\phi,i})$ correctly.  
Therefore the probability that $\tilde{\tilde{E}}_p^\mu[S_\phi]\cap (S_{\phi,i}\times S_{\phi,i})$ of edges among nodes $S_{\phi,i}$ is chosen in a way that $(1-\gamma)$-approximates edges in $H_p^\mu[S_\phi]|_{S_{\phi,i}}$ is at least 
\begin{eqnarray*}
\left(1-\left(\frac{\epsapprog}{f(h_1)}\right)^{\Theta(1)}\right)^{\BO\left(f(h_\phi)\Ratio^2\right)}
&\geq&
1-\BO(f(h_\phi)\Ratio^2)\cdot \left(\frac{\epsapprog}{f(h_1)}\right)^{\Theta(1)}
=
1-\left(\frac{\epsapprog}{f(h_1)}\right)^{\Theta(1)},
\end{eqnarray*} 
where we use 
\begin{itemize}
\item $h_\phi\leq \poly\Ratio$, as $h_\phi\leq h_1$, $h_1 \geq 3^{\Phi-1}\geq \Phi$ (see Lemma~\ref{lem:h-bound}), and $\Phi:=\Theta(\log\Ratio)$; and
\item that the growth bound $f$ is a monotonic increasing function (as the number of neighbors can only grow with the distance), and
\item that we can choose the constant hidden in the $\Theta$-notation arbitrarily high.
\end{itemize}

Furthermore, as $\mu$ and $\gamma$ are constants, and as $h_1\leq \hPhi$ (see Lemma~\ref{lem:h-bound}), and as $f$ is a polynomial function, we can bound the runtime $T$ by

\begin{eqnarray*}
T
&=&
\Theta\left(\frac{\log \frac{f(h_1)}{\epsapprog}}{\gamma^2\mu}\right)
=
\Theta\left(\log \frac
{f\left(\hPhi\right)}{\epsapprog}\right)
\\
&=&
\Theta\left(\Phi + \log(\SWn) +  \log(1/\epsapprog)\right)
=
\Theta\left(\Phi + \log(1/\epsapprog)\right),
\end{eqnarray*}

where we choose the constant hidden in the $\Theta$-notation sufficiently high. Finally, note that in this process the decision whether an edge $(u,v)$ is in the graph $\tilde{\tilde{H}}_p^\mu[S_\phi]$ does not involve communication between nodes other than $u$ and $v$.
\end{proof}

\subsubsection{Probability of Local Success of Computing Set $S_{\phi+1}$ Based on $\tilde{\tilde{H}}_p^\mu[S_{\phi}]$ } 
This is an important step towards the analysis of Property 2 of local success of an epoch. Using this property we later iteratively guarantee local success of computing sets $S_\phi$ in each phase. As the nodes involved in decisions of other nodes cannot be too far away, this helps to bound the probability that locally correct computations take place. If all involved computations at nodes in transmission range are successful, approximate progress takes place.

\begin{lemma}\label{lem:0.2}
Given graph $\tilde{\tilde{H}}_p^\mu[S_\phi]$, consider phase $\phi$ in Algorithm~\ref{alg:AdHoc2}. Line~\ref{line:compS} described in Section~\ref{sec:6.3.4} computes in time 
$
\BO\left((\Phi + \log(1/\epsapprog)) \SWn\right)
$
a set $S_{\phi+1}$ that is an independent set in $\tilde{\tilde{H}}_p^\mu[S_\phi]$.  
The computation of set $S_{\phi+1,i}$ is successful at point $i$. Furthermore, determining the $S_{\phi+1,i}$ part of $S_{\phi+1}$ involves only nodes in $S_{\phi,i}$. 
\end{lemma}

\begin{proof}
\textbf{Runtime analysis:} The algorithm described in Section~\ref{sec:6.3.4} consists of simulating (in the SINR model) an algorithm to compute an MIS in the CONGEST model taking $\SW$ rounds in the CONGEST model. The provided simulation of each round of the CONGEST model in the SINR model takes $\BO(T)$ time slots. Therefore the total runtime is
\begin{eqnarray*}
\BO(T\cdot \SWn)
&=&
\BO\left(\log\left(\frac{f(h_1)}{\epsapprog}\right)\cdot \SWn\right)
\\
&=&
\BO\left(\log\left(\frac{f\left(\hPhi\right)}{\epsapprog}\right)\cdot \SWn\right)
\\
&=&
\BO\left((\Phi + \log(\SWn) + \log(1/\epsapprog))\cdot \SWn\right).
\\
&=&
\BO\left((\Phi + \log(1/\epsapprog))\cdot \SWn\right),
\end{eqnarray*}
where we use similar arguments as in the runtime analysis of Lemma~\ref{lem:0.1} as well as the fact that $\log \log^*(\Ratio)\leq \Theta(\log(\Ratio))=\Phi$. 
\\
\\
\noindent 
\textbf{An independent set is computed:} This algorithm simulates the MIS algorithm of~\cite{DBLP:conf/podc/SchneiderW08}, which computes an MIS in growth-bounded graphs, and attempts to compute a subset $S_{\phi+1}$ of an MIS on $\tilde{\tilde{H}}_p^\mu[S_\phi]$. The algorithm might not achieve maximality as nodes might stop participating in this epoch after their communication was unsuccessful\footnote{Recall that communication at node $u\in S_\phi$ is unsuccessful if $u$ did not receive a message (and acknowledgments) for own messages form each neighbor in $\tilde{\tilde{H}}_p^\mu[S_\phi]$.}. As these nodes do not join $S_{\phi +1}$, set $S_{\phi+1}$ is still an independent set in $\tilde{\tilde{H}}_p^\mu[S_\phi]$. 
\\
\\
\textbf{Given graph $\tilde{\tilde{H}}_p^\mu[S_\phi]$, the computation of set $S_{\phi+1}$ is successful at point $i$
:} As all communication in $\tilde{\tilde{H}}_p^\mu[S_\phi]|_{S_{\phi,i}}$ is successful due to using the same schedule $\tau_\phi$ as when computing $\tilde{\tilde{H}}_p^\mu[S_\phi]$, the set $S_{\phi+1}\cap S'_{\phi,i}$ is $(\phi,i)$-locally maximal. Recall that set $S'_{\phi,i}$ depends on $h'_\phi:=3h_{\phi+1}$, a choice taking into account that each hop in $\tilde{\tilde{H}}_p^\mu[S_{\phi+1}]|_{S_{\phi+1,i}}$ corresponds to at most $3$ hops in $\tilde{\tilde{H}}_p^\mu[S_{\phi}]|_{N_{\tilde{\tilde{H}}_p^\mu[S_{\phi}]}(S_{\phi+1,i})}$, as otherwise $(\phi+1,i)$-local maximality of $S_{\phi+1}$ in $\tilde{\tilde{H}}_p^\mu[S_{\phi}]|_{N_{\tilde{\tilde{H}}_p^\mu[S_{\phi}]}(S_{\phi+1,i})}$ was violated. Therefore we conclude that $N_{\tilde{\tilde{H}}_p^\mu[S_{\phi}]}(S_{\phi+1,i})\subseteq S'_{\phi,i}$. Furthermore any node $v$ that could be added to $S_{\phi+1,i}$ without violating independence in $\tilde{\tilde{H}}_p^\mu[S_{\phi}]$ is at least $h'_{\phi}+1$ hops away from $u_{\phi+1}$ in $\tilde{\tilde{H}}_p^\mu[S_{\phi}]$ and thus at least $h_{\phi+1}+2$ hops away from $u_\phi$ in $\tilde{\tilde{H}}_p^\mu[S_{\phi+1}]$. Therefore Definition~\ref{def:suc} of local successful computation of $S_{\phi+1,i}$ is satisfied.

Also note that only nodes in $N_{\tilde{\tilde{H}}_p^\mu[S_{\phi}],\SW}(S'_{\phi,i})=S_{\phi,i}$ are involved in the computation, as the runtime of the MIS algorithm of~\cite{DBLP:conf/podc/SchneiderW08} is $\SW$. 
\end{proof}

\subsubsection{Probability of Satisfying Properties 1 and 2 of a Local Successful Epoch}

\begin{lemma}\label{lem:2}
For any set $S_1\subseteq V$ and node $i\in N_\Ga(S_1)$, both Properties 1 and 2 of Definition~\ref{def:suc} of a successful epoch at point $i$ are satisfied with probability at least $1-\epsapprog/3$. 
\end{lemma}
\begin{proof}
Due to Lemma~\ref{lem:0.1} only nodes in $S_{\phi,i}$ are involved in computing $\tilde{\tilde{H}}_p^\mu[S_\phi]|_{\phi,i}$ and Lemma~\ref{lem:0.2} states that only this part of the graph is involved in computing $S_{\phi+1,i}$. By induction we only need to bound the probability that all graphs $\tilde{\tilde{H}}_p^\mu[S_1],\dots,\tilde{\tilde{H}}_p^\mu[S_{\Phi-1}]$ and all sets $S_2,\dots,S_{\Phi}$ are computed successfully at point $i$ to prove the statement. 

Due to Lemma~\ref{lem:0.1} the probability that any of the graphs $\tilde{\tilde{H}}_p^\mu[S_1],\dots,\tilde{\tilde{H}}_p^\mu[S_{\Phi-1}]$ is computed successfully at point $i$ is at least $1-\left(\frac{\epsapprog}{f(h_1)}\right)^{\Theta(1)}$. The probability that all of the sets $S_2,\dots,S_{\Phi}$ are computed successfully at point $i$ is 
$1$ due to Lemma~\ref{lem:0.2}. Notice that $1-\left(\frac{\epsapprog}{f(h_1)}\right)^{\Theta(1)}
$ can be lower bounded by $(1-\frac{\epsapprog}{3^{\Phi-1}})^{\Theta(1)}$, as $h_1\geq 3^{\Phi-1}$  (Lemma~\ref{lem:h-bound}) and $f$ is a monotonic increasing polynomial. While we obtain this generous bound as a side effect of other parts of the analysis, it is sufficient for our purposes to use $1-\frac{\epsapprog}{6\Phi}$ as a lower bound for this probability. Here, we assume $\Phi=\Theta(\log\Ratio)\geq 4$ for simplicity of the presentation.
As there are $\Phi$ phases, in total
$\Phi$ graphs need to be computed. Thus the probability that all these computations are successful at point $i$ is at least $(1-\frac{\epsapprog}{6\Phi})^{
\Phi}\geq 1-\epsapprog/3$.
\end{proof}

\subsection{Probability of Approximate Progress Conditioned on Satisfaction of Property 3 of a Local Successful Epoch}
After proving initial lemmas in the previous subsections, we first give an outline how these connect to the remaining parts of the proof of Theorem~\ref{thm:approg} via this section. In Lemma~\ref{lem:2} we argued that with probability at least $1-\epsapprog/3$ we can assume that Properties 1 and 2 of Definition~\ref{def:suc} of a successful epoch at point $i$ are satisfied. Therefore we assume in Lemma~\ref{lem:3} that Properties 1 and 2 of Definition~\ref{def:suc} of a successful epoch at point $i$ are satisfied, and show that in this case there is a phase $\phi'\in\{1,\dots,\Phi\}$ such that $i$ could be able to receive a bcast-message $m$ from a $\Ge$-neighbor. Node $i$ will receive such a message if Property 3 of Definition~\ref{def:suc} of a successful epoch at point $i$ is satisfied. As we cannot yet analyze the probability of satisfaction of Property 3, we condition our probabilities in this section on satisfaction of Property 3. To analyze probability of satisfaction of Property 3, we first need to bound the runtime of an epoch, which is done in Lemma~\ref{lem:4}. In Section~\ref{sec:kuhnproof} we analyze the probability that $i$ indeed receives $m$ in phase $\phi'$ from a $\Ge$-neighbor by combining results from this and previous sections with a bound on the probability for satisfaction of Property 3 of Definition~\ref{def:suc} of a successful epoch at point $i$. Section~\ref{sec:kuhnproof} also concludes the proof of the bound on $\fapprog$ stated in Theorem~\ref{thm:approg}. 

The main Lemma that we prove in this Section is
\fi
\begin{lemma}\label{lem:3}
Given set $S_1$ and a node $i$ and let there be a $\Ga$-neighbor of $i$ with an ongoing broadcast event $bcast(m)_j$. Assume Properties 1 and 2 of 
\iffull
Definition~\ref{def:suc} 
\fi
of a successful epoch at point $i$ are satisfied. Then there is a phase $\phi'\in\{1,\dots,\Phi\}$ such that in phase $\phi$ node $i$ receives a bcast-message from a $\Ge$-neighbor of $i$ with probability $1-\epsapprog/3$.
\end{lemma}
\iffull
However, before we can proof Lemma~\ref{lem:3}, we need to derive a few more lemmas which are extended from~\cite{daumfull} to our localized setting.

\begin{lemma}[Extended version of Lemma~4.3 of~\cite{daumfull}]\label{lem:daum4.3}
For all $p\in(0,1/2]$, there is a $\mu\in (0,p)$ such that: Let $d_{\min}\leq \Ra$ be the shortest distance between two nodes in a set $S\subseteq S_1$. Then the graph $H_p^\mu[S]$ contains all edges between pairs $u,v\in S$ for which $d(u,v)\leq \min\{2d_{\min},\Ra\}$.
\end{lemma}
The part of the proof of Lemma 4.3 of~\cite{daumfull} that changes relies on bounding the interference $I_S(u)$ that $u$ receives from a set $S$. Compared to~\cite{daumfull} we not only need to bound interference from a set $S$, but from $S\cup W$, as nodes in $W$ might still participate in the computation and send messages due to unsuccessful transmissions in a previous phase that made them compute wrong neighborhoods. We show that one can choose slightly modified parameters in the algorithm/analysis such that the interference $I_{S\cup W}(u)$ is as small as in the original proof by~\cite{daumfull}. Therefore other parts of their proof are not affected and can be immediately transferred. 

For completeness we restate the full proof of~\cite{daumfull} adapted to our modifications and extentions. 

\begin{proof}
We restrict our attention to the case $d_{\min} < r_s/2$. If the minimum distance is between $r_s/2$ and $r_s$, the claim can be shown by a similar, simpler argument.

Consider some node $u \in S$. Due to the underlying metric space in our model, there are at most $\BO(k^\delta)$ nodes in $S$ within distance $kd_{\min}$ of node $u$. Let $v$ be a node at distance at most $2d_{\min}$ from $u$. For any constant $k_0$, with probability $p(1-p)^{\BO(k_0^\delta)} = \Omega (p)$, node $v$ is the only node transmitting among all the nodes within distance $k_0d_{\min}$ from node $u$. Further, assuming that all nodes at distance greater than $k_0d_{\min}$ transmit, the interference $I_{S\cup W}(u)$ at $u$ can be bounded from above by

\begin{align*}
I_{S\cup W}(u) 
&\leq I_{W}(u) + \sum\limits_{w\in S \ s.t.\ d(u,w)\geq k_0d_{\min}} \frac{P}{d(u,w)^\alpha} 
\leq I_{W}(u) + \sum\limits_{k=k_0}^\infty \sum\limits_{w\in S \ s.t. \ 1\leq\frac{d(u,w)}{kd_{\min}}<1+\frac{1}{k}} \frac{P}{d(u,w)^\alpha} 
\\
&\overset{(1)}{=} I_{W}(u) + \sum\limits_{k=k_0}^\infty \frac{P}{k^\alpha d_{\min}^\alpha} O\left(\delta k^{\delta-1}\right)
= I_{W}(u) + \frac{P}{d_{\min}^\alpha}O \left(\delta\int\limits_{k_0}^\infty k^{-(1+\alpha_{\min}-\delta)}dk\right) 
\\
&= I_{W}(u) + \frac{P}{d_{\min}^\alpha}\BO \left(\delta\frac{k_0^{\delta-\alpha_{\min}}}{\alpha_{\min}-\delta}\right)
\leq \left(\frac{\epsapprog}{\Ratio}\right)^{\Theta(1)} +  \kappa(k_0)\frac{P}{d_{\min}^\alpha}
\\
&\overset{(2)}{\leq}  \left(\frac{\epsapprog}{\Ratio}\right)^{\Theta(1)}
+ \kappa(k_0)\frac{P\cdot\Ratio}{\Ra}\\
&\overset{(3)}{\leq} \kappa(k_0/2)\frac{P\cdot\Ratio}{\Ra}.
\end{align*}

Step $(1)$ stems from bounding $|\{w\in S : kd_{\min} \leq d(u,w) < (k+1)d_{\min}\}|$, the maximum number of nodes within a ring of diameter $d_{\min}$ at distance $kd_{\min}$. If we define the function $\kappa$ so as to replace the $\BO$-term with $\kappa(k_0) = \kappa(k_0, \alpha_{\min}, \delta)>0$, which decreases polynomially in $k_0$. 
Step $(2)$ stems from bounding $I_W(u)$ using Lemma~\ref{lem:wronginterference} and restating $\kappa(k_0)\frac{P}{d_{\min}}$ by $\kappa(k_0)\frac{P\cdot\Ratio}{\Ra}$.
Step $(3)$ is true, as $\Ratio\geq 1$, $\Ra$ is constant, and the exponent hidden in the $\Theta$-notation can be chosen arbitrary large in order to match the function $\kappa$.

Due to the choice of $k_0$ and $\kappa$, we get for $SINR(u, v, I)$, where I is the set of all nodes with distance greater than $k_0d_{\min}$:
\begin{equation*}
\frac{\frac{P}{d(u,v)^\alpha}}{N + \kappa (k_0) \frac{P}{d_{\min}^\alpha}} \geq \frac{\frac{P}{(2d_{\min})^\alpha}}{N + \kappa (k_0) \frac{P}{d_{\min}^\alpha}} \geq \frac{\frac{P}{r_s^\alpha}}{\frac{P}{\beta r_w^\alpha} + \kappa(k_0) \frac{2^\alpha P}{r_s^\alpha}} = \frac{\beta}{\frac{1}{(1+\rho)^\alpha} + \kappa(k_0)\beta2^\alpha} \geq \beta
\end{equation*}
The second inequality follows from $N = \frac{p}{\beta r_w^\alpha}$ and from $d_{\min} \leq r_s/2$. The last inequality holds for sufficiently large $k_0$. If we choose $\mu$ to be the probability that no more than one node in a ball of radius $k_0d_{\min}$ transmits, then node $v$ can transmit to $u$ with probability $\mu$.

In the above proof, $\mu$ depends on the unknown parameter $\beta$, so we use $\beta_{max}$ as the base for
computing $\mu$. Note also that since $H_p^\mu[S] \subseteq \tilde{H}_p^\mu[S]$, the lemma induces the same properties on $\tilde{H}_p^\mu[S]$ with high probability.
\end{proof}

\begin{lemma}[Version of Lemma 4.4 of~\cite{daumfull}]\label{lem:daum5}
Given node $i\in N_\Ga(S_1)$ and assume Properties 1 and 2 of Definition~\ref{def:suc} of a successful epoch at point $i$ are satisfied. Then for any $\phi\in\{1,\dots,\Phi\}$, the minimum distance between any two nodes in $S_{\phi,i}$ is at least $d_\phi\geq 2^{\phi-1}\cdot d_{\min}$.
\end{lemma}

\begin{proof}
The proof appears in 
the full version of this paper~\cite{halldorsson2015local-arxiv}, as it requires only a minimal modification of a proof provided in~\cite{daumfull}.  
\end{proof}

\begin{lemma}[Extended version of Lemma 4.5,~\cite{daumfull}]\label{lem:daum4.5}
For all $p\in (0,1/2]$, there is a $\hat{Q}, \gamma=\Theta(1)$, such that for all $Q\geq \hat{Q}$ the following holds. Consider a round $r$ in phase $\phi$ where each node in $S_\phi$ transmits a bcast-message with probability $p/Q$ (Line~\ref{line:1}). Let $i\in N_{\Ga}(S_1)$ and let $u_\phi\in S_\phi\setminus \{v\}$ be the closest node to $v$ in $S_\phi$. Assume Property 1 of Definition~\ref{def:suc} of a successful epoch at point $i$ are satisfied. Let $d_{u_\phi}$ be the distance between $u_\phi$ and its farthest neighbor in $\tilde{\tilde{H}}_p^\mu[S_\phi]$. If $d(u_\phi, v)\leq (1+\eps)\Ra$ and $d_{u_\phi}\geq \gamma Q^{-1/\alpha}\cdot d(u_\phi,v)$, node $v$ receives a bcast-message from $u_\phi$ in round $r$ with probability $\Theta(1/Q)$.
\end{lemma}

Note that the proof presented in~\cite{daumfull} reveals that the bcast-message is actually received from node $u_\phi$ such that we adapted the statement to this fact (instead of stating that $v$ receives a message from some node). We restate the full proof of~\cite{daumfull} with our extensions that yield Lemma~\ref{lem:daum4.5}.
There are two main issues we need to take care of: 
\begin{enumerate}
\item The proof presented in~\cite{daumfull} relies on $\tilde{H}_p^\mu[S_\phi]$ being a $\gamma$-close approximation of $H_p^\mu[S_\phi]$. When looking at this proof in more detail, it turns out that this approximation is only required for all nodes located at distance at most $2d_{u_\phi}$ around $u_\phi$. We show that this area is covered by a $\BO(1)$-neighborhood of $u_\phi$ in $\tilde{\tilde{H}}_p^\mu[S_\phi]$ such that the statement of~\cite{daumfull} on $\tilde{H}_p^\mu[S_\phi]$ can be transferred to our graph $\tilde{\tilde{H}}_p^\mu[S_\phi]$.

\item The proof given in~\cite{daumfull} deals with interference from nodes in $S_\phi$ at distance further than $2d_{\phi}$ from $u_\phi$. We show that additional interference from nodes $W$ that arises due to our modification of their algorithm is negligible compared to interference from nodes in $S_\phi$. We conclude that nodes in $W$ do not affect the remainder of the proof of~\cite{daumfull}. 
\end{enumerate}

\begin{proof}[of Lemma~\ref{lem:daum4.5}]
The full proof by~\cite{daumfull} with the described extensions is deferred to the Appendix, Lemma~\ref{applem:daumfull4.5}
\end{proof} 

\begin{lemma}[Version of Lemma 4.6. of~\cite{daumfull}]\label{lem:daumfull4.6} Assume Property 2 of the Definition~\ref{def:suc} of a successful epoch is satisfied. With probability $1-\epsapprog/3$, either $u$'s bcast-message reaches $i$ in phase $\phi$, or $d(u_{\phi+1}, i) \leq \Ra\left(1+\phi\frac{\eps}{\log \Ratio}\right)$.
\end{lemma}

\begin{proof}
We defer the proof to the Appendix (Lemma~\ref{applem:daumfull4.6}), as it requires only a minimal modification of a proof provided in~\cite{daumfull}.  
\end{proof}

\subsubsection{Proof of Lemma~\ref{lem:3}}
\begin{proof}
We extend and adapt parts of the proof presented in Section 4.3 of~\cite{daumfull} to our setting. We show that $\Phi=\Theta(\log\Ratio)$ phases are sufficient such that each node $v\in N_{\Ga}(S_1)$ receives a bcast-message from a $\Ge$-neighbor. 
First, we know due to Lemma~\ref{lem:daum5} that the minimum distance between nodes in $S_{\phi,i}$ grows exponentially with $\phi$. Therefore in some phase $\phi\leq \Phi=\Theta(\log\Ratio)$ (assuming $\Ratio\geq \Ra/d_{\min}$, which is satisfied in any non-trivial instance) the minimum distance between nodes in $S_\phi$ exceeds $\Ra\cdot\left(1+\phi\frac{\eps}{\log \Ratio}\right)$. 
Second, by applying Lemma~\ref{lem:daumfull4.6}, there must be a phase $\phi$ in which $i$ receives with probability $1-\epsapprog/3$ a bcast-message from a node $u_\phi\in S_\phi$ at distance $d(u_{\phi+1}, i) \leq \Ra\left(1+\phi\frac{\eps}{\log \Ratio}\right)$ to $i$. As $\phi\leq\Phi=\Theta(\log\Ratio)$, this can be bounded to be less than $\R_{1-2\eps-2\eps^2}< \Ra$. We conclude that the there must be a bcast-message $m'$ sent by a $\Ge$-neighbor which arrives at node $i$ during the epoch. This is the bcast-message $m'$ for which node $i$ outputs an $rcv$ event (Line~\ref{line:1}).
\end{proof}

\subsection{Runtime of an Epoch}

\begin{lemma}\label{lem:4}
The runtime of an epoch is
$$\BO\left(\log^{\alpha+1}(\Ratio)\cdot \log\left(\frac{1}{\epsapprog}\right) \ \ +\ \  \log(\Ratio)\cdot \log\left(\frac{1}{\epsapprog}\right)\cdot \log^*\left(\frac{1}{\epsapprog}\right)\right).$$
\end{lemma}

\begin{proof} Due to Lemma~\ref{lem:0.1} each execution of Line~\ref{line:compH} that constructs a graph $\tilde{\tilde{H}}_p^\mu[S_\phi]$ takes time $\BO\left(\Phi + \log(1/\epsapprog)\right)$. Due to Lemma~\ref{lem:0.2} each execution of Line~\ref{line:compS} that constructs a set $S_{\phi+1}$ takes time $\BO\left((\Phi + \log(1/\epsapprog)) \SWn\right)
$. 
In Lines~\ref{line:22}--\ref{line:23} a bcast-message is sent $\BO(Q \log(1/\epsapprog))$ times. Due to the choice of $Q$ this is $\BO(\log^{\alpha}(\Ratio)\cdot\log(1/\epsapprog))$. All this is executed for each of the $\Phi$ phases of an epoch (Lines~\ref{line:2}--~\ref{line:3}). Thus the total runtime of an epoch is 
\begin{eqnarray*}
&&\BO\left(\Phi\left(\left(\Phi \ \ +  \ \ \log\left(\frac{1}{\epsapprog}\right)\right)\SWn  \ \ + \ \ \log^{\alpha}(\Ratio)\cdot\log\left(\frac{1}{\epsapprog}\right)\right)\right)
\\
&&= \ \ \BO\left(\log^2(\Ratio)\SWn  \ \ + \ \ \log(\Ratio)\SWn\log\left(\frac{1}{\epsapprog}\right)\right.
\\
&& \left. \ \ \ \ \  \ \  \ \  \ \  \ \  \ \  \ \ \ \  \ \  \ \  \ \  \ \  \ \  \ \ \ \  \ \  \ \  \ \  \ \  \ \  \ \ \ \  \ \  \ \  \ \  \ \ \ + \ \ \log^{\alpha+1}(\Ratio)\cdot\log\left(\frac{1}{\epsapprog}\right)\right)
\\
&&= \ \
\BO\left(\log^{\alpha+1}(\Ratio)\cdot \log\left(\frac{1}{\epsapprog}\right) \ \ +\ \  \log(\Ratio)\cdot \log\left(\frac{1}{\epsapprog}\right)\cdot \log^*\left(\frac{1}{\epsapprog}\right)\right),
\end{eqnarray*}
where we use the definition of $\Phi=\Theta(\log \Ratio)$ and $\alpha> 2$. 
\end{proof}

\subsection{Proof of Theorem~\ref{thm:approg} (Approximate Progress Bound)}\label{sec:kuhnproof}
\begin{proof}
\textbf{Probability of approximate progress conditioned on locally unique labels (Remark~\ref{rem:unique}).}
Consider any node $i$ that has an $\Ga$-neighbor $j$ with an ongoing broadcast event (i.e.~$j\in S_1$). Under the assumption of 1) locally unique labels (Remark~\ref{rem:unique}), and 2) satisfaction of Properties 1 and 2 of Definition~\ref{def:suc} of a successful epoch at point $i$, Lemma~\ref{lem:3} states that node $i$ receives a bcast-message from a $\Ge$-neighbor of $i$ with probability $1-\epsapprog/3$ within one epoch. Lemma~\ref{lem:2} provides that Properties 1 and 2 of Definition~\ref{def:suc} of a successful epoch at point $i$ are satisfied with probability $1-\epsapprog/3$. Therefore the total probability that $i$ receives a bcast-message from a $\Ge$-neighbor within an epoch is at least $1-2\epsapprog/3$, which on the one hand implies satisfaction of Property 3 of Definition~\ref{def:suc} of a successful epoch at point $i$, and on the other hand implies that approximate progress is made at point $i$.
\\
\\
\textbf{Probability of locally unique labels (Remark~\ref{rem:unique}).}
We apply Lemma~\ref{lem:locMIS} with $H:=\tilde{\tilde{H}}_p^\mu[S_\phi]$ and $U:=U_{\phi,i}$ using random labels $\in[1,\frac{\poly \Ratio}{\epsapprog}]$. Lemma~\ref{lem:locMIS} can be applied, as $|U_{\phi,i}|=\BO(\Ratio^2)$. This implies that the modified MIS algorithm computes an independent set that is maximal with respect to $N_{\tilde{\tilde{H}}_p^\mu[S_\phi],\hPhi}(U_{\phi,i})$ with probability $1-\frac{\epsapprog}{3\Phi}$. As $S_{\phi,i}\subseteq N_{\tilde{\tilde{H}}_p^\mu[S_\phi],\hPhi}(U_{\phi,i})$, set $S_{\phi+1}$ is $(\phi,i)$-locally maximal in $\tilde{\tilde{H}}_p^\mu[S_\phi]$ with probability $1-\frac{\epsapprog}{3\Phi}$. The probability that we can assume locally unique labels (Remark~\ref{rem:unique}) at each of the $\Phi$ phases is at least $\left(1-\frac{\epsapprog}{3\Phi}\right)^\Phi\geq 1-\epsapprog/3$.
\\
\\
\textbf{Final conclusion.} When the two arguments the are combined, we conclude that approximate progress is made within one epoch with probability at least $(1-2\epsapprog/3)\cdot(1-\epsapprog/3)\geq 1-\epsapprog$.  Therefore Lemma~\ref{lem:4} bounds not only the runtime of an epoch, but also $\fapprog$ by 

$$\BO\left(\left(\log^\alpha(\Ratio)+ \log^*\left(\frac{1}{\epsapprog}\right)\right)\log(\Ratio)\log\left(\frac{1}{\epsapprog}\right)\right).$$
\end{proof}

\begin{remark}
It might be the case that in Algorithm~\ref{alg:AdHoc2} a node receives the same bcast-message over and over again for $\fack/2$ time slots (until the sender stops broadcasting), which is an extreme case that still satisfies the definition of progress and approximate progress. We want to stress that this is not a problem, as Algorithm~\ref{alg:AdHoc2} is only required to implement fast approximate progress and not acknowledgments. Acknowledgments are obtained in Algorithm~\ref{alg1}.
\end{remark}

\section{A Probabilistic AbsMAC Implementation with Fast Acknowledgments and Approximate Progress in the SINR-Model}\label{sec:kuhn2}

\begin{theorem}\label{thm:kuhn1}
Algorithm~\ref{alg:AdHoc0} implements the probabilistic absMAC of~\cite{DBLP:journals/adhoc/KhabbazianKKL14} for $G:=\Ge$. Approximate progress is measured with respect to $\tilde{G}:=\Ga$. The algorithm ensures local broadcast in $G$ s.t. 
$$\fack=\BO\left(\Delta_{\Ge}\cdot \log \left(\frac{\Ratio}{\epsack}\right) \ \ +\ \   \log(\Ratio)\log\left(\frac{\Ratio}{\epsack}\right)\right)$$
with probability at least $1-\epsack$ with respect to acknowledgments in $G:=\Ge$, and 
$$\fapprog=\BO\left(\left(\log^\alpha(\Ratio)+ \log^*\left(\frac{1}{\epsapprog}\right)\right)\log(\Ratio)\log\left(\frac{1}{\epsapprog}\right)\right).$$
with probability at least $1-\epsapprog$ with respect to approximate progress in $\tilde{G}:=\Ga$.
\end{theorem}

\begin{remark}
The bound on $\fapprog$ is significantly better than the best possible bound on $\fprog$ due to the lower bound in Theorem~\ref{thm:fprogLB}. E.g. for graphs where $\Delta_{\Ge}=\Omega(\Ratio^\gamma)$, for $\gamma>0$, the bound on $\fapprog$ of Theorem~\ref{thm:kuhn1} is polylogarithmic in the degree $\Delta_{\Ge}$, while the lower bound on $\fprog$ in Section~\ref{sec:fprogLB} is linear in the degree $\Delta_{\Ge}$.
\end{remark}
To achieve the bounds stated in Theorem~\ref{thm:kuhn1}, we use two algorithms that run in parallel. 
\begin{itemize}
\item Algorithm of Theorem~\ref{thm:ack} is executed in every even time step with respect to $\Ge$ and ensures an almost optimal bound on $\fack$. 
\item Algorithm~\ref{alg:AdHoc2} is executed in every odd time step and ensures fast approximate progress with respect to $\tilde{G}$.
\end{itemize}

Combining these two algorithms provides good bounds on both, $\fack$ and $\fapprog$. Such a combination is necessary, as the Algorithm of Theorem~\ref{thm:ack} might not yield a good bound on approximate progress and Algorithm~\ref{alg:AdHoc2} might not lead to an acknowledgment at all. Therefore they complement each other.

\subsection{Details of the Algorithm}
Once a $bcast(m)_i$ event occurs at node $i$, node $i$ starts to execute Algorithm~\ref{alg:AdHoc0} for $\fack$ time steps. After this node $i$ performs $ack(m)_i$. If node $i$ has an ongoing broadcast and receives an $abort(m)_i$ input from the environment before it performs $ack(m)_i$, then it (i) continues to participate until the current epoch of Algorithm~\ref{alg:AdHoc2} is finished, (ii) after this epoch performs no further transmission on behalf of bcast-message $m$, and (iii) does not perform $ack(m)_i$. Whenever a node $i$ receives a message $m'$ for the first time in an epoch, it delivers that message to its environment with a $rcv(m')_i$ output event.

\begin{algorithm}[!htbp]
\begin{algorithmic}[1]
\Statex
\State $m:=0$; // $m$ stores bcast-message input of an ongoing $bcast$-event at $i$
\State whenever a bcast-message is received or a $bcast(m')_i$ event occurs, wake up if not awake;
\State whenever a $bcast(m')_i$ event occurs, set $m:=m'$ and reset $m:=0$ after $\fack$ rounds;
\State Execute in parallel in even/odd time steps:
\Statex \ \ \ The algorithm of Theorem~\ref{thm:ack}, and
\Statex \ \ \ Algorithm~\ref{alg:AdHoc2};
\end{algorithmic}
\caption{Implementation of local broadcast as executed by a node $i$.}\label{alg:AdHoc0}
\end{algorithm}

\begin{proof}[of Theorem~\ref{thm:kuhn1}]
Details on the Algorithm and proof corresponding to the bound on $\fack$ in Theorem~\ref{thm:kuhn1} are stated in Section~\ref{sec:ackcons}.
Details on the Algorithm and proof corresponding to the bound on $\fapprog$ in Theorem~\ref{thm:kuhn1} are stated in Sections~\ref{sec:highlevel} and~\ref{sec:anaI}. We also point the reader to Remark~\ref{rem:exact} used in these proofs.
\end{proof}

\fi

\section{Application: Improved Network-Wide Broadcast}\label{sec:comb2}
We combine this absMAC implementation with algorithms from~\cite{DBLP:journals/adhoc/KhabbazianKKL14} for global broadcast in this absMAC. 
\iffull
We recall the relevant Theorems of~\cite{DBLP:journals/adhoc/KhabbazianKKL14}.
\fi
In Theorem~\ref{thm:replace} we argue that we can replace $\fprog$ and $\epsprog$ in the relevant Theorems of~\cite{DBLP:journals/adhoc/KhabbazianKKL14} by $\fapprog$ and $\epsapprog$ under certain conditions and state the effect that this replacement has on other parameters of the runtime. 
\iffull
We use

$$ c_2:=\frac{2}{1-\epsprog} \ \ \ \ \ \ and \ \ \ \ \ \ c_3:=\frac{3}{1-\epsprog}.$$

For the convenience of the reader we restate Theorem 7.7 and 8.20 of~\cite{DBLP:journals/adhoc/KhabbazianKKL14} with respect to broadcast in graph $G$ and diameter $D_G$, and recall notation used in these Theorems.
\begin{theorem}[Version of Theorem 7.7 of~\cite{DBLP:journals/adhoc/KhabbazianKKL14}]\label{thm:prob-BSMB-ld}
Let $G$ be a graph in which local broadcast is implemented that can be used via the probabilistic absMAC. Let $\gamma'$ be a real number, $0<\gamma'\leq 1$. The BSMB protocol [of~\cite{DBLP:journals/adhoc/KhabbazianKKL14}] guarantees that, with probability at least 
$$1-\gamma'-n\cdot\epsack,$$
$rcv$ events and hence, deliver events, occur at all nodes $\neq i_0$ by time 
$$(c_3 D_G+ c_2 \ln(n/\gamma'))\fprog$$
\end{theorem}

\begin{definition}[Nice broadcast events and nice executions, Definition 4.1 of~\cite{DBLP:journals/adhoc/KhabbazianKKL14}]
Suppose a $bcast(m)_i$ event $\pi$ occurs at time $t_0$ in execution $\alpha$. Then we say that
$\pi$ is nice if $ack(m)_i$ occurs by time $t_0+\fack$ and is preceded by a $rcv(m)_j$ for every neighbor $j$ of $i$. We say that execution $\alpha$ is nice if all $bcast$ events in $\alpha$ are nice. Let $N$ be the set of all nice executions. 
\end{definition}

\begin{definition}[Clear events, Definition 8.1 of~\cite{DBLP:journals/adhoc/KhabbazianKKL14}]
Let $\alpha$ be an execution in $N$ (the set of nice executions), and let $m\in M$
be a message such that an $arrive(m)$ event occurs in $\alpha$. We define the event
$clear(m)$ to be the final $ack(m)$ event in $\alpha$.
\end{definition}

\begin{definition}[The Set
$K(m)$, Definition 8.2 of~\cite{DBLP:journals/adhoc/KhabbazianKKL14}]
Let $\alpha$ be an execution in $N$ and let $m\in M$ be a message such that
$arrive(m)$ occurs in $\alpha$. We define $K(m)$ to be the set of messages
 $m'\in M$ such that an $arrive(m')$ event precedes the $clear(m)$ event and the
$clear(m')$ event follows the $arrive(m)$ event. That is, $K(m)$ is the set
of messages whose processing overlaps the interval between the $arrive(m)$ and
$clear(m)$ events.
\end{definition}
\begin{theorem}[Version of Theorem 8.20 of~\cite{DBLP:journals/adhoc/KhabbazianKKL14}]\label{thm:prob-BMMB-old}
Let $G$ be a graph in which local broadcast is implemented that can be used via the probabilistic absMAC. Let $m\in M$ and let $\gamma'$ be a real number, $0<\gamma'<1$. The BMMB protocol [of~\cite{DBLP:journals/adhoc/KhabbazianKKL14}] guarantees that, with probability at least
$$1-\gamma'-nk\epsack,$$
the following property holds for the generated execution $\alpha$:
\\
Suppose an $arrive(m)_i$ event $\pi$ occurs in $\alpha$, and let $t_0$ be the time of occurrence of $\pi$. Let $k'$ be a positive integer such that $|K(m)|\leq k'$. Then $get(m)$ events, and hence, deliver events occur at all nodes in $\alpha$ by time 
\begin{eqnarray*}
& t_0 + \left(
(c_3+c_2)D_G + ((c_3+2c_2)\left\lceil\ln\left(\frac{2n^3 k}{\gamma'}\right)\right\rceil + c_3+c_2)k'\right)\fprog
+(k'-1)\fack
\end{eqnarray*}
\end{theorem}
\fi
\begin{theorem}\label{thm:replace}
Let $G$ be a graph in which local broadcast is available via the probabilistic absMAC of~\cite{DBLP:journals/adhoc/KhabbazianKKL14}. Let $\tilde{G}$ be the graph in which approximate progress is measured and let the vertex sets of the connected components of $\tilde{G}$ and $G$ be the same. Then one can replace $\fprog, \epsprog$ and $D_G$ in Theorems
\iffull
~\ref{thm:prob-BSMB-ld} and~\ref{thm:prob-BMMB-old}
\fi concerning their global SMB and MMB algorithms by $\fapprog,\epsapprog$ and $D_{\tilde{G}}$.
\end{theorem}

\iffull
\begin{proof}
We start by recalling the BMMB and BSMB protocols of~\cite{DBLP:journals/adhoc/KhabbazianKKL14} for global MMB and SMB, for which we present our argument. 
\\
\\
\textbf{Basic Multi-Message Broadcast (BMMB) Protocol:} Every process $i$ maintains a FIFO
queue named $bcastq$ and a set named $rcvd$. Both are initially empty. If process $i$ is not currently sending a message on the MAC layer and its $bcastq$ is not empty, it sends the message at the head of the queue on the MAC layer (disambiguated with identifier $i$ and sequence number) using a $bcast$ output. If $i$ receives a message from the environment via an $arrive(m)_i$ input, it immediately delivers the message $m$ to the environment using a $deliver(m)_i$ output, and adds $m$ to the back of $bcastq$ and to the $rcvd$ set. If $i$ receives a message $m$ from the MAC layer via a $rcv(m)_i$ input, it first checks $rcvd$. If $m\in  rcvd$ it discards it. Else, $i$ immediately performs a $deliver(m)_i$ output and adds $m$ to $bcastq$ and $rcvd$. \\
\\
\textbf{Basic Single-Message Broadcast (BSMB) Protocol:} This is just BMMB specialized to
one message, and modified so that the message starts in the state of a designated initial node $i_0$.
\\
\\
In the above algorithms, once a node $i$ receives a message, node $i$ broadcasts the message if it did not broadcast it before. The result of global broadcast is independent of whether a message was received due to transmission from a $\tilde{G}$-neighbor or a $G$-neighbor as long as the components of $\tilde{G}$ and $G$ are the same. Only the runtime changes.

In time $\fprog$ it is guaranteed that with probability $1- \epsprog$ a message is received by a node $v$ when a $G$-neighbor of $v$ is sending. Therefore the runtime presented in~\cite{DBLP:journals/adhoc/KhabbazianKKL14} depends on $D_{\tilde{G}}$. Compared to this it is guaranteed with probability $1- \epsapprog$ that in time $\fapprog$ a message arrives when a $\tilde{G}$-neighbor is sending. A message that causes approximate progress in $G$ with respect to $\tilde{G}$ essentially causes progress in $\tilde{G}$ if we restrict local broadcast to $\tilde{G}$. Here, if required by the specification of the abstract MAC layer (or algorithms using it) we output $rcv$-events for messages that arrive from $G$-neighbors, but not from other nodes outside of $G$. Therefore $D_G$ needs to be replaced by $D_{\tilde{G}}$.

Now note that a node $i$ that receives a message $m$ from the MAC layer via a $rcv(m)_i$ discards $m$ if $m\in  rcvd$. Therefore messages from $G$ are only placed into $bcastq$ once and cannot cause delays more than once. Based on this we can now replace $\fprog$ and $\epsprog$ in Theorems~\ref{thm:prob-BSMB-ld} and~\ref{thm:prob-BMMB-old} by $\fapprog$ and $\epsapprog$ if we also take into account the change of the diameter of the graph in which we consider broadcast. Therefore the diameter $D_G$ is replaced by $D_{\tilde{G}}$.

Although one might now only need $\fack$ with respect to broadcast in $\tilde{G}$, we still need to use the bound of $\fack$ for $G$, as broadcast is implemented in $G$. We conclude the statement, as $G:=\Ge$ and $\tilde{G}:=\Ga$.
\end{proof}

By combining Theorems~\ref{thm:prob-BSMB-ld} and~\ref{thm:prob-BMMB-old} with our results, we obtain:
\fi

\begin{theorem}\label{thm:combKuhn}
Consider the SINR model using the model assumptions stated in Section~\ref{sec:modelassumpt}. 
We present an algorithm that performs global SMB in graph $\Ge$ with probability at least $1-\epsSMB$ in time 
\iffull
$$\BO\left(\left(D_\Ga+\log \left(\frac{n}{\epsSMB}\right)\right)\cdot \log^{\alpha+1}(\Ratio)\right).$$
\fi
\iffull
$$\BO\left(D_\Ga \log^{\alpha+1}(\Ratio) \ \ + \ \ k'\left(\Delta_{\Ger}+\polylogf{\frac{nk\Ratio}{\epsMMB}}\right)\log \left(\frac{nk}{\epsMMB}\right)\right).$$
\fi
\end{theorem}


\iffull 
\begin{proof}
Theorem~\ref{thm:kuhn1} states that $\fack=\BO\left(\Delta_{\Ger}\cdot \log \left(\frac{\Ratio}{\epsack}\right) \ \ +\ \   \log(\Ratio)\log\left(\frac{\Ratio}{\epsack}\right)\right)$ and $\fapprog=\BO\left(D_\Ga \log^{\alpha+1}(\Ratio) \ \ + \ \ 
k(\Delta_{\Ger}+\polylogf{nk\Ratio})\log \left(nk\right)\right)$.

\paragraph{Global SMB:}
Theorem~\ref{thm:prob-BSMB-ld} combined with Theorem~\ref{thm:replace} guarantees that for $0<\gamma'\leq 1$ with probability $1-\gamma'-n\cdot\epsack$, global SMB can be performed in time $(c_3 D_\Ga+ c_2 \ln(n/\gamma'))\fapprog$. We choose $\gamma'=\epsSMB/2$ and $\epsack:=\epsSMB/(2n)$. Therefore we obtain that global SMB is performed with probability $1-\gamma'-n\cdot\epsack=1-\epsSMB$. Choosing $\epsapprog:=1/8$ yields the following total runtime: 

$$(c_3 D_\Ga+ c_2 \ln(n/\gamma'))\fapprog = \BO\left((D_\Ga+\log (n/\epsSMB))\cdot \log^{\alpha+1}(\Ratio)\right).$$

\paragraph{Global MMB:}
Theorem~\ref{thm:prob-BMMB-old} guarantees that for $0<\gamma'\leq 1$ with probability $1-\gamma'-nk\epsack$, global MMB is completed at time 

$$t_0 + \left(
(c_3+c_2)D_\Ga + ((c_3+2c_2)\left\lceil\ln\left(\frac{2n^3 k}{\gamma'}\right)\right\rceil + c_3+c_2)k'\right)\fprog +(k'-1)\fack.$$

We choose $\gamma'=\epsMMB/(2k)$ and $\epsack:=\epsMMB/(2kn)$. Therefore we obtain that global MMB is performed with probability $1-\gamma'-nk\cdot\epsack=1-\epsMMB$. 
This yields the following total runtime: 
\begin{eqnarray*}
&& \left((c_3+c_2)D_\Ga + ((c_3+2c_2)\left\lceil\ln\left(\frac{2n^3 k}{\gamma'}\right)\right\rceil + c_3+c_2)k'\right)\fprog +(k'-1)\fack
\\
&&=
\BO\left(D_\Ga\fprog +k'(\fack + \log\left(nk/\epsMMB\right)\fprog)\right)
\\
&&=
\BO\left(D_\Ga\fapprog +k'(\fack + \log\left(nk/\epsMMB\right)\fapprog)\right)
\\
&&=
\BO\left(D_\Ga \log^{\alpha+1}(\Ratio) \ \ + \ \ 
k'\left(\Delta_{\Ger}+\polylogf{\frac{nk\Ratio}{\epsMMB}}\right)\log \left(\frac{nk}{\epsMMB}\right)\right)
\end{eqnarray*}

In our setting we can simply replace $k'$ by $k$, as we consider the one-shot version of $k$-message broadcast. This results in the claimed runtime. Furthermore note that we do not need the model assumption (see Section~\ref{sec:modelassumpt}) that nodes know their $\Ge$-neighbors in case $\Ge$ is connected (also see the discussion in Remark~\ref{rem:exact}). When looking at the proof of Theorem~\ref{thm:replace} and the BMMB protocol stated therein, we conclude that even if messages are received from nodes in transmission range that are not $\Ge$-neighbors, messages are added to $bcastq$ only once and cannot cause delays several times.
\end{proof}
\fi

\paragraph{Acknowledgments:} We thank Sebastian Daum, Mohsen Ghaffari, Fabian Kuhn and Calvin Newport for answering questions concerning their earlier work and helpful discussions. In particular we thank Erez Kantor for many helpful discussions---especially in the early stages of this work.

\addcontentsline{toc}{section}{References} 
\bibliographystyle{abbrv}
\bibliography{references}

\newpage
\begin{center}
\textbf{Appendix}
\end{center}
\appendix
\section{Basic Lemma on Growth Bounded Graphs}\label{app:growth}

\begin{lemma} Let $G$ be growth bounded  with polynomial bounding function $f(r)$. Then it is  $|N_{G,r}(v)|\leq \Delta f(r)$.
\end{lemma}
\begin{proof}
This statement is well known in the unit-disc graph community. As we did not find a reference to this version of the statement we include a proof for completeness. The number of nodes in $|N_{G,r}(v)|$ that are in an independent set of $G$ is bounded by $f(r)$. Consider the subgraph $H$ of $G$ that consist of nodes $N_{G,r}(v)$ and edges of $G$ between them. Any independent set in $H$ can be extended to an independent set in $G$ and thus is of size at most $f(r)$. On the other hand an independent set on $H$ dominates all nodes in $H$ such that the size of $H$ is at most $\Delta f(r)$, as each node has degree at most $\Delta$.
\end{proof}

\section{Proof of~\cite{halldorsson2012towards} Adapted to our Theorem~\ref{thm:ack}}\label{app:MMHproof}
We only restate Algorithm and Analysis from~\cite{halldorsson2012towards} adapted to our needs for completeness and convenience of the reader with the goal of making it simpler to verify our claim in Theorem~\ref{thm:ack}. The proof is minimal modified and variables are replaced by more general parameters to demonstrate correctness. In particular we restate Theorem 3 of~\cite{halldorsson2012towards} with respect to an upper bound $\tilde{N_x}$ of the local contention $N_x$. Here $N_x$ is defined to be the number of $G_{\MMHphi}$-neighbors of $x\in V$ that have ongoing broadcasts at the time of execution. Compared to this Theorem 3 of~\cite{halldorsson2012towards} assumes $n$ as an upper bound such that the runtime depends on $n$. However, we are interested in local parameters. Furthermore we only wish to claim successful local broadcast within the stated time with probability $1-\epsack$, while~\cite{halldorsson2012towards} claims w.h.p., which affects the runtime as well. For simplicity we use their Notation of regions $T_x$ and $B_x$, see Definition~\ref{def:TB}, which at the same time serve as sets of nodes with ongoing broadcast located in these regions. Our notion of $N_{G_\MMHphi}(x)$ describes a similar set of all nodes in the ball of radius $R_\MMHphi$, but also those nodes with no ongoing broadcasts. Furthermore this set cannot be treated as an area and $B_x$ is more useful for this.
\begin{definition}~\label{def:TB}
The \emph{transmission region} $T_x$ is the ball of radius $\R_1$ around a node $x$ which $x$ can reach without any other node transmitting. The \emph{broadcasting region} $B_x$ is a ball of radius $\R_{\MMHphi}$ around any node $x$, containing all nodes to which $x$ would like to transmit.
\end{definition}
\begin{remark} The analysis below is transferred from \cite{halldorsson2012towards} and requires $\eps$ large enough such that $\MMHphi=:\phi\leq\frac{1}{6}$. However, this is only for simplicity of the presentation and can be adapted to arbitrary small constant $\eps$.
\end{remark}

\begin{theorem}[Version of Theorem 3 of~\cite{halldorsson2012towards}]\label{appthm:ack-app}
Let $\tilde{N_x}$ be an upper bound on the local contention $N_x$ and let $\epsack>0$. When executing Algorithm~\ref{alg1} each node $x$ successfully performs a local broadcast within 
$$O(N_x\log (\tilde{N_x}/\epsack) +  \log(\tilde{N_x})\log(\tilde{N_x}/\epsack))$$
rounds with probability at least $1-\epsack$.
\label{mainth1}
\end{theorem}

The  symbols $\gamma', \lambda$ used in Algorithm~\ref{alg1} are appropriate constants.

\begin{algorithm}[htb]
\caption{LocalBroadcast (For any node $y$)}   
\label{alg1}                           
\begin{algorithmic}[1]                    
     \State $tp_y \leftarrow 0$
      \State $p_y \leftarrow \frac{1}{4\tilde{N_x}}$
    \Loop
      \State $p_y \leftarrow \max\{\frac{1}{128 \tilde{N_x}}, \frac{p_y}{32}\}$ \label{resetline}
      \State $rc_y \leftarrow 0$
     \Loop
      \State $p_y \leftarrow \min\{\frac{1}{16}, 2 p_y\}$ \label{probincrease}
     \For{$j \leftarrow{} 1, 2, \ldots, \delta \log (\tilde{N_x}/\epsack)$} \label{innerloop}
       \State $s \leftarrow 1$ with probability $p_y$ \label{choosetransmit}
       \If{s = 1}
        \State transmit
       \EndIf
       \State $tp_y \leftarrow tp_y + p_y$
       \If{$tp_y > \gamma' \log(\tilde{N_x}/\epsack)$}  \label{haltcondition}
          \State halt;
       \EndIf
      \If{message received}
        \State $rc_y \leftarrow rc_y + 1$
        \If{$rc_y > 8\log (2\tilde{N_x}/\epsack)$}
          \State goto line \ref{resetline} \label{resettrigger}
        \EndIf
       \EndIf
      \EndFor
     \EndLoop
    \EndLoop
\end{algorithmic}
\label{alg1fig}
\end{algorithm}

The intuition behind the algorithm is as follows. The ``right'' probability for $x$ to transmit at is about $\frac{1}{\tilde{N_x}}$ (too high, and collisions are inevitable; too low, nothing happens). The algorithm starts from a low probability, continuously increasing it, but once it starts receiving messages from others, it uses that as an indication that the ``right''  transmission probability has been reached.

To prove Thm.~\ref{mainth1}, we will first need the following definition.
\begin{definition}
For any node $x$, the event {\lp} occurs at a time slot if the received power at $x$ from other nodes, $P_x \leq \frac{1}{(4 (\beta + 4) \R_{\MMHphi})^{\alpha}}$.
\end{definition}

\noindent The following technical Lemma follows from geometric arguments.
\begin{lemma}\label
If $x$ transmits and {\lp} occurs at $x$, all nodes in $2 B_x$ receive the message from $x$ (thus a successful local broadcast occurs for $x$).
\label{lowpmeanssuccess}
\end{lemma}

\begin{proof}{[of Lemma \ref{lowpmeanssuccess}]}
Consider any $y \in 2 B_x$. By definition of $2 B_x$, $d(x, y) \leq 2 \R_{\MMHphi}$. Now consider any other transmitting node $z$. We will show that,

\begin{claim}
$d(z, x) \leq 3 (\beta + 2) d(z, y)$
\end{claim}
\begin{proof}
By the signal propagation model, $\frac1{d(z, x)^{\alpha}}$ is the power received at $x$ from $z$. Since {\lp} occurred,
\begin{align*}
\frac1{d(z, x)^{\alpha}} & \leq \frac{1}{((4 \beta + 4) \R_{\MMHphi})^{\alpha}}   \\
\Rightarrow\,\, d(z, x) & \geq 4 (\beta + 4) \R_{\MMHphi}
\end{align*}

By the triangle inequality, $d(z, y) \geq  d(z, x) - d(x, y)  > 
4 (\beta + 4) \R_{\MMHphi} - 2 \R_{\MMHphi} \geq 3 (\beta + 4) \R_{\MMHphi}$, proving the claim.
\end{proof}

This implies, by basic computation and summing over all transmitting $z$, that
\begin{equation}
P_y \leq \left(\frac43\right)^{\alpha} P_x
\label{relatepxpy}
\end{equation}
Now, the SINR at node $y$ (in relation to the message sent by $x$) is

\begin{align*}
  \frac{\frac1{2^{\alpha} \R_{\MMHphi}^{\alpha}}}{P_y + N} \overset{1}{\geq}
    \frac{\frac1{2^{\alpha} \R_{\MMHphi}^{\alpha}}}{\left(\frac43\right)^{\alpha}  P_x + N} 
    \overset{2}{\geq}
 \frac{\frac1{2^{\alpha} \R_{\MMHphi}^{\alpha}}}{\left(\frac43\right)^{\alpha}  
\frac{1}{((4 (\beta + 4)) \R_{\MMHphi})^{\alpha}}  + \frac{{\MMHphi}^{\alpha}}{\R_{\MMHphi}^{\alpha}\beta}} \overset{3}{\geq} \beta
\end{align*}

Explanation of numbered (in)equalities:

\begin{enumerate}
\item By Eqn.~\ref{relatepxpy}.
\item Plugging in the bound of $P_x$ (since {\lp} occurs at $x$) and noting that
$N = \frac1{\beta \R_1^{\alpha}} = \frac{{\MMHphi}^{\alpha}}{\beta \R_{\MMHphi}^{\alpha}}$, from the definitions of $\R_1$ and $\R_{\MMHphi}$.
\item Follows from simple computation once ${\MMHphi}$ is set to a small enough constant ($\MMHphi = \frac16$ suffices).
\end{enumerate}

Thus the SINR condition is fulfilled, and $y$ receives the message from $x$.
\end{proof}

\noindent We will also need the following definition:
\begin{definition}
A {\fb} event is said to occur for node $y$
if line \ref{resettrigger} is executed for $y$.
\end{definition}

We will refer to the transmission probability $p_y$ for a node $y$ at 
given time slots. This will always refer to the value of $p_y$ in line \ref{choosetransmit}. We first provide a few basic lemmas needed for the proof of Lemma~\ref{lem:9}, that bounds the transmission probability in any broadcast region at a given time.

\begin{lemma}
Consider any slot $t$ and any node $z$.
Assume that in that slot, for all broadcast regions $B_x$, $\sum_{y  \in B_x} \leq \frac{1}{2}$. Then, {\lp} occurs for $z$ with probability 
at least $\frac{1}{2} \left(\frac{1}{4}\right)^{\frac12 O\left(\frac1{\MMHphi^2}\right)}$.
\label{lphappens}
\end{lemma}
\begin{proof}

Let $B = B_x \setminus \{x\}$.  We first prove that there is a substantial probability that no node in $B$ transmits. Assuming this probability is $\Pro_{N_x}$

\begin{align*}
&  \Pro_{N_x} \geq \prod\limits_{w \in B} (1 - p_w)  \geq \prod\limits_{w \in B_x} (1 - p_w)  \geq \left(\frac14\right)^{\sum_w p_w} \geq \left(\frac14\right)^{\frac12}
\end{align*}
The third inequality is from Fact 3.1 \cite{DBLP:conf/dialm/GoussevskaiaMW08}, and the last from the bound $\sum_w p_w \leq \frac12$.

Let $\Pro_T$ be the probability that no other node transmits in $T_x$. Since $\R_{\MMHphi} = \MMHphi \R_1$, $T_x$ can be covered by $O(\frac1{\MMHphi^2})$ broadcast regions (this can be shown using basic geometric arguments). 
Thus, 

\begin{equation}
\Pro_T \geq \Pro_{N_x}^{O\left(\frac1{\MMHphi^2}\right)} \geq \left(\frac14\right)^{\frac12 O\left(\frac1{\MMHphi^2}\right)}
\end{equation}

Since no other node in $T_x$ is transmitting, we only need to bound the signal received from outside $T_x$. 

To this end, we need the following Claim (which is a restatement of Lemma 4.1 of \cite{DBLP:conf/dialm/GoussevskaiaMW08} and can be proven by standard techniques):

\begin{claim}
Assume that for all broadcast regions $B_x$, \\$\sum_{y  \in B_x} p_y \leq \frac12$.
Consider a node $x$. Then the expected power received at node $x$ from nodes not in $T_x$ can be upper bounded by 

$$\frac18 \frac{\alpha -1}{\alpha -2} 3^3 2^{\alpha -2} \frac{\MMHphi^2}{\R_{\MMHphi}^{\alpha}} \leq \frac{1}{2 (4 (\beta + 4) \R_{\MMHphi})^{\alpha}}$$ for appropriately small $\MMHphi$.\end{claim}

Then by Markov's inequality, with probability at least $\frac12$, the power
received from nodes outside of $T_x$ is at most $\frac{1}{ (4 (\beta + 4) \R_{\MMHphi})^{\alpha}}$.

Thus, with probability $\frac12 \Pro_T$, {\lp} occurs at $x$, proving the Lemma.
\end{proof}

\begin{lemma}\label{lem:9}
Consider any node $x$. Then during any time slot $t \leq 10 N_x^2$, 
\begin{equation}
\sum_{y \in B_x} p_y \leq \frac{1}2
\label{boundedprob}
\end{equation}
with probability at least  $1/2$.
\end{lemma}
\begin{proof}
For contradiction, we will upper bound the probability that Eqn.~\ref{boundedprob} is violated for the first time at any given time $t$, after which we will union bound over all $t \leq 10 N_x^2 \leq 10 \tilde{N_x}^2$.

Let $\calT$ be the interval (time period) $\{t - \delta \log (\tilde{N_x}/\epsack) + 1 \ldots t -1\}$. Then we claim,

\begin{claim}
In each time slot in the period $\calT$,
\begin{equation}
  \frac{1}2 \geq \sum_{y \in B_x} p_y \geq \frac{1}4 \label{probbounds}
\end{equation}
\end{claim}
\begin{proof}
The first inequality is by the assumption that $t$ is the first slot when Eqn.~\ref{boundedprob} is violated. The second is because probabilities (at most) double once every $\delta \log (\tilde{N_x}/\epsack)$ slots (by the description of the algorithm).
\end{proof}

We now show that Eqn.~\ref{probbounds} is not possible. To that end, we show that in the $\delta \log (\tilde{N_x}/\epsack)$ interval preceding $t$, a {\fb} will occur with high probability:

\begin{claim}
With probability $1 - \frac{1}{N_x^8}$, each node $z \in B_x$ will {\fb} once in the period $\calT$.
\label{fallbackoccurs}
\end{claim}

\begin{proof}
Fix any $z \in B_x$.
By the algorithm 
\begin{equation}
p_z \leq \frac{1}{16} \label{yupperbound}  
\end{equation}

Thus, at any time slot,
\begin{equation}
\Pro(z  \text{ does not transmit}) \geq \frac{15}{16}
\label{ydoesnttransmit}
\end{equation}

Now, combining Eqn.~\ref{yupperbound}  and Eqn.~\ref{probbounds}, and defining
$B = B_x \setminus \{z\}$,
\begin{equation}
  \sum_{y \in B} p_y \geq \frac3{16}
  \label{highprobothers}
\end{equation}

For $y \in B_x$ define $\succ_y$ to be the event that $y$ transmits and {\lp} occurs for $y$. By Lemma \ref{lowpmeanssuccess}, $\succ_y$ implies that $z$ will receive the message from $y$.
Thus, the probability of $z$ receiving a message from some node in $B$ in a given round is at least $\frac{15}{16}\Pro(\bigcup\limits_{y\in B}\succ_y)$.

We claim that for any $y \neq w$ (both in $B)$, the events $\succ_y$ and $\succ_w$ are disjoint. This is implicit in Lemma \ref{lowpmeanssuccess}, since $\succ_y$ means that $w$ cannot be transmitting and vice-versa. Thus, the probability of $z$ receiving a message from some node in $B$ is at least:

\begin{align*}
\lefteqn{\frac{15}{16} \Pro(\bigcup\limits_{y\in B}\succ_y)  = \frac{15}{16} \sum\limits_{y\in B} \Pro(\succ_y)} \\
& \geq \frac{15}{16} \sum\limits_{y\in B} p_y \frac{1}2 \left(\frac{1}4\right)^{\frac{1}2 O\left(\frac{1}{\MMHphi^2}\right)} 
\geq \frac{15}{32} \left(\frac{1}4\right)^{\frac{1}2 O\left(\frac{1}{\MMHphi^2}\right)} \frac3{16}\ , 
\end{align*}
where we use Lemma \ref{lphappens} for the first inequality and Eqn.~\ref{highprobothers} for the last.

Setting $\delta \geq \frac{10}{\frac{15}{32} \left(\frac{1}4\right)^{\frac{1}2 O(\frac{1}{\MMHphi^2})} \frac3{16}}$ and using the Chernoff bound, we can show that $z$ will receive $> 8\log (2\tilde{N_x}/\epsack)$ messages in $\calT$ with probability $1 - \frac{1}{\tilde{N_x}}$, thus triggering the {\fb}.
\end{proof}

Now we show that the above claim implies that Eqn.~\ref{probbounds} is not possible.
\begin{claim}
There exists a time slot in $\calT$ such that \\$\sum_{y \in B_x} p_y < \frac{1}4$.
\end{claim}
\begin{proof}
For any $y \in B_x$, let $p_y^1$ be the value of $p_y$ in the first slot of $\calT$. Let $p_y^f$ be the value of $p_y$ in the slot when {\fb} happened for $y$. Since probabilities can at most double during $\calT$,
\begin{equation}
\sum_{y \in B_x} p_y^f \leq 2 \sum_{y \in B_x} p_y^1 \leq 1\ ,
\label{pfubound}
\end{equation}
the last inequality using the fact that $\sum_{y \in B_x} p_y^1 \leq \frac{1}2$ (Eqn.~\ref{probbounds}).

Now by lines \ref{resetline} and \ref{probincrease} of the algorithm,
in the slot after {\fb}, $p_y = \max\{\frac{1}{128\tilde{N_x}},
\frac{p_y^f}{32}\} \leq \frac{1}{128\tilde{N_x}} + \frac{p_y^f}{32}$. Since
probabilities at most double during $\calT$, the value of $p_y$ at the
final slot of $\calT$ is at most $\frac{1}{64\tilde{N_x}} +
\frac{p_y^f}{16}$. Summing over all $y$, during the final slot of
$\calT$,

\begin{align*}
\sum_{y \in B_x} p_y  \leq \frac{N_x}{32\tilde{N_x}} + \sum_{y \in B_x} \frac{p_y^f}{8}
\leq \frac{1}{32} + \frac{1}8 < \frac{1}4
\end{align*}
contradicting Eqn.~\ref{probbounds}. We used Eqn.~\ref{pfubound} in the second
inequality.
\end{proof}

The proof of the Claim is completed by union bounding over time slots 
$t \leq 10 N_x^2 \leq 10 \tilde{N_x}^2$.
\end{proof}

\noindent Now we prove that nodes stop running the algorithm by a certain time.
\begin{lemma}
Each node $x$ stops executing within $O(N_x \log (\tilde{N_x}) + \log^2 (\tilde{N_x}) + \log(\tilde{N_x}/\epsapprog))$ slots, with probability at least $1-\epsack/2$.
\label{perf1}
\end{lemma}
\begin{proof}
Fix $x$. We derive four claims that together imply the lemma.

\begin{claim}
The number of slots for which $p_x \geq \frac{1}{32}$ is $O(\log (\tilde{N_x}/\epsack))$.
\label{pxusuallysmall}
\end{claim}
\begin{proof}
This is ensured by the halting condition in line~\ref{haltcondition}.
\end{proof}

Assume that $x$ experienced
$k$ {\fb s}.
Consider the times $t_x(1), t_x(2) \ldots t_x(k)$ when
a {\fb}  happened for $x$. Now,
\begin{claim}
$t_x(1) = O(\log (\tilde{N_x})\log(\tilde{N_x}/\epsack))$. Also, there are $O(\log(\tilde{N_x})\log(\tilde{N_x}/\epsack) )$ slots after $t_x(k)$.
\end{claim}
\begin{proof}
The two claims are very similar. Let us prove the latter one.
Since {\fb} does not occur after $t_x(k)$, the probability $p_x$ of each node doubles every $\delta \log (\tilde{N_x}/\epsack)$ slots. Since the minimum probability is $\Omega(\tilde{N_x})$, by $O(\log (\tilde{N_x})\log(\tilde{N_x}/\epsack))$ slots, the probability will reach $\frac{1}{32}$. Once this happens, the algorithm terminates in $O(\log (\tilde{N_x}/\epsack))$ additional slots, by Claim \ref{pxusuallysmall}.
\end{proof}

Given the above claim it suffices 
to bound $t_x(k) - t_x(1)$. By Claim \ref{pxusuallysmall} we can also restrict ourselves to slots for which $p_x < \frac{1}{32}$. For these slots,
line $\ref{probincrease}$ does not need the $\min$ clause, i.e., 
$p_y \leftarrow 2 p_y$ each time line $\ref{probincrease}$ is executed. 

Define $b_i$ such that $p_x = \frac{1}{2^{b_i}}$ at time $t_x(i)$. 
Note that if $\tilde{N_x}$ is a power of $2$, $b_i$ is always an integer (the case of other values of $\tilde{N_x}$ can be easily managed).

We can characterize the running time between two {\fb s} as follows.
\begin{claim}
$t_x(i+1) - t_x(i) \leq (b_{i} - b_{i+1} + 5) \delta \log (\tilde{N_x}/\epsack)$, for all $i = 1, 2 \ldots k-1$.
\label{tintermsofb}
\end{claim}
\begin{proof}
During slots in $[t_x(i), t_x(i+1))$, $p_x$ doubles every $\delta \log (\tilde{N_x}/\epsack)$ slots (by the description of the algorithm and the fact that 
$p_x < \frac{1}{32}$). 
Let $b$ be  such that $p_x = \frac{1}{2^{b}}$ at time $t_x(i+1) - 1$. Then, 
\begin{align*}
& \frac{1}{  2^{b}} = \frac{2^{\left\lfloor{\frac{t_x(i+1) - t_x(i)}{\delta \log (\tilde{N_x}/\epsack)}}\right\rfloor}}{2^{b_i}} \\
\Rightarrow\,\,  & b_i - b = \left\lfloor{\frac{t_x(i+1) - t_x(i)}{\delta \log (\tilde{N_x}/\epsack)}}\right\rfloor
\end{align*}
By lines \ref{probincrease} and \ref{resetline} of the algorithm, $b_{i+1} \leq b + 4$, and thus, 
\begin{align*}
b_i - b_{i+1} + 4 & \geq  \left\lfloor{\frac{t_x(i+1) - t_x(i)}{\delta \log (\tilde{N_x}/\epsack)}}\right\rfloor\\ \Rightarrow
b_i - b_{i+1} + 5 & \geq  \frac{t_x(i+1) - t_x(i)}{\delta \log (\tilde{N_x}/\epsack)} \ ,
\end{align*}
completing the proof of the Lemma.
\end{proof}

Thus, the running time $t_x(k) - t_x(1)$ can be bounded by:
\begin{align}
\lefteqn{t_x(k) - t_x(1)}  \nonumber \\ 
& = (t_x(k) - t_x(k-1)) + (t_x(k-1) - t_x(k-2)) \nonumber \\ 
& \qquad \ldots  + (t_x(2) - t_x(1)) \nonumber \\
& \leq  ((b_{k-1} - b_{k} + 5) + (b_{k-2} - b_{k-1} + 5) \nonumber \\
  & \qquad \ldots + (b_{1} - b_{2} + 5)) \delta \log (\tilde{N_x}/\epsack) \nonumber \\
& = (b_1 - b_k + 5 k ) \delta \log (\tilde{N_x}/\epsack) \nonumber \\
& = O(\log(\tilde{N_x})\log(\tilde{N_x}/\epsack) +  k \log (\tilde{N_x}/\epsack)) \ ,\label{mainrunbound1}
\end{align}
where we use Claim \ref{tintermsofb}, the non-negativity of $b_k$ and the fact that $b_i = O(\log \tilde{N_x})$ (as $p_x = \Omega(\frac{1}{\tilde{N_x}})$).

To complete the proof of the Lemma, we need a bound on $k$:

\begin{claim}
With probability $1-\epsack/2$, each node transmits at least $4 \gamma' \log (\tilde{N_x}/\epsack)$ times, and at most $16 \gamma' \log (\tilde{N_x}/\epsack)$ times. 
\label{transmissiontimes}
 \end{claim}
 \begin{proof}
By the description of the algorithm, when the node stops, its total transmission probability is $\gamma' \log (\tilde{N_x}/\epsack)$.
By the standard Chernoff bound, the actual number of transmissions is very close to this number, with probability at least $1-\epsack/(2N_x^8)$, which is at least $1-\epsack/2$.
\end{proof}

\begin{claim}
$k = O(N_x)$ with probability at least $1-\epsack/2$. \label{kbounded}
\end{claim}
\begin{proof}
The total number of possible transmissions that $x$ could possibly hear is upper bounded by $O(N_x \log (\tilde{N_x}/\epsack))$, with probability at least $1-\epsack/(\poly N_x)$ (due to a Chernoff bound). (However, we only need probability at least $1-\epsack/4$ for our purposes.) 
This is because each node transmits $O(\log (\tilde{N_x}/\epsack))$ times, with probability at least $1-\epsack/4)$  (by Claim \ref{transmissiontimes}) and a node can only hear messages from nodes in $T_x$ (by the definition of $T_x$).
But nodes only {\fb} once for every $8\log (\tilde{N_x}/\epsack)$ messages received (by the condition immediately preceding line \ref{resettrigger}). The claim is proven with probability guarantee $(1-\epsack/4)(1-\epsack/4)\geq (1-\epsack/2)$. 
\end{proof}

Applying the above claim to Eqn.~\ref{mainrunbound1}, 

\begin{eqnarray*}
t_x(k) - t_x(1) 
&\leq&
O(\log(\tilde{N_x})\log (\tilde{N_x}/\epsack) + k \log (\tilde{N_x}/\epsack)) 
\\
&=& O(N_x \log (\tilde{N_x}/\epsack) +\log (\tilde{N_x})\log (\tilde{N_x}/\epsack)),
\end{eqnarray*}

completing the argument.
\end{proof}

The final piece of the puzzle is to show that for each node, a successful local broadcast happens with probability at least $1-\epsack/2$ during one of its $\Theta(\gamma' \log (\tilde{N_x}/\epsack))$ transmissions. 

\begin{lemma}
By the time a node halts, it has successfully locally broadcast a message, with probability at least $1-\epsack/2$.
\label{lem:succ-lb}
\end{lemma}
\begin{proof}
The expected number of transmission made by a node is $\gamma' \log (\tilde{N_x}/\epsack)$ (by the algorithm). By Lemma \ref{lphappens} (which can be applied, as Lemma \ref{lphappens}'s prerequisites are met each time with probability $1/2$ due to Lemma~\ref{lem:9}) and Lemma \ref{lowpmeanssuccess}, during each such transmission, local broadcast succeeds with probability $\frac{1}{2} \left(\frac{1}4\right)^{\frac{1}2 O(\frac{1}{\MMHphi^2})}$, at least. Thus, the expected number of successful local broadcasts is $(1-1/2)\cdot\frac{1}2 \left(\frac{1}4\right)^{\frac{1}{2} O(\frac{1}{\MMHphi^2})} \gamma' \log (\tilde{N_x}/\epsack)$. Setting $\gamma'$ to a high enough constant, and using Chernoff bounds, with probability at least $1-\epsack/2$, a successful local broadcast happens at least once.
\end{proof}

Lemmas \ref{perf1} and \ref{lem:succ-lb} together imply Thm.~\ref{mainth1} with probability guarantee $(1-\epsack/2)(1-\epsack/2)\geq(1-\epsack)$.

\section{Useful Lemmas and Proofs from~\cite{daumfull} Adapted to our Needs}\label{app:daum}
We restate two lemmas and proofs from~\cite{halldorsson2012towards} adapted to our needs for completeness. This is done only for the convenience of the reader with the goal of making it simpler to verify our claim. Compared to the adapted proofs in the main-body of the paper, the proofs presented here have only minor modifications and are adapted to our notation.

\begin{lemma}[Version of Lemma 4.4 of~\cite{daumfull}]\label{applem:daum5}
Given node $i\in N_\Gar(S_1)$ and assume Properties 1 and 2 of Definition~\ref{def:suc} of a successful epoch at point $i$ are satisfied. Then for any $\phi\in\{1,\dots,\Phi\}$, the minimum distance between any two nodes in $S_{\phi,i}$ is at least $d_\phi\geq 2^{\phi-1}\cdot d_{\min}$.
\end{lemma}

\begin{proof}
For completeness and clarity we restate the full proof of~\cite{daumfull} and extend it to our setting. We prove the lemma by induction on $\phi$. By definition of $d_{\min}$, it is $d_1\geq 2^0\cdot d_{\min}=d_{\min}$. By the definition of $\gamma'$-close approximation of $H_p^\mu[S]$ and as we assume that Properties 1 and 2 of Definition~\ref{def:suc} of a successful epoch at point $i$ are satisfied, we can apply Lemma~\ref{lem:daum4.3} and conclude that $H_p^\mu[S_\phi]|_{S_{\phi,i}}$ contains edges between all pairs of nodes $S_{\phi,i}$ at distance $d(u,v)\leq 2\cdot d_\phi$. As $S_{\phi+1,i}$ is $(\phi+1,i)$-locally maximal in $\tilde{\tilde{H}}_p^\mu[S_\phi]|_{S_{\phi,i}}$, nodes in $S_{\phi+1,i}$ are at distance more than $2\cdot d_\phi$. Using the induction hypothesis, it is $d_{\phi+1}>2\cdot d_\phi\geq 2^\phi \cdot d_{\min}$. 
\end{proof}

\begin{lemma}[Version of Lemma 4.5. of~\cite{daumfull}]\label{applem:daumfull4.5} Assume Property 2 of the For all $p\in (0,1/2]$, there is a $\hat{Q}, \gamma=\Theta(1)$, such that for all $Q\geq \hat{Q}$ the following holds. Consider a round $r$ in phase $\phi$ where each node in $S_\phi$ transmits a bcast-message with probability $p/Q$ (Line~\ref{line:1}). Let $i\in N_{\Ga}(S_1)$ and let $u_\phi\in S_\phi\setminus \{v\}$ be the closest node to $v$ in $S_\phi$. Assume Property 1 of Definition~\ref{def:suc} of a successful epoch at point $i$ are satisfied. Let $d_{u_\phi}$ be the distance between $u_\phi$ and its farthest neighbor in $\tilde{\tilde{H}}_p^\mu[S_\phi]$. If $d(u_\phi, v)\leq (1+\eps)\Ra$ and $d_{u_\phi}\geq \gamma Q^{-1/\alpha}\cdot d(u_\phi,v)$, node $v$ receives a bcast-message from $u_\phi$ in round $r$ with probability $\Theta(1/Q)$.
\end{lemma}

\begin{proof}
For completeness and we restate the full proof of~\cite{daumfull} and extend it based on the ideas summarized in the main-body of our paper. 
The lemma states under what conditions in round $r$ of block 2 in phase $\phi$ a node $v \in
N(S)\backslash S$ can receive the message. The roadmap for this proof is to show that if $u$ is able to
communicate with probability $(1-\varepsilon)\mu$ with its farthest neighbor $u'$ in some round $r'$ of block 1 in phase $\phi$, using the broadcast probability $p$, then $u$ must also be able to reach $v$ with probability $\Theta(1/Q)$ in round $r$ of block 2, in which it transmits with probability $p/Q$. We start with some notations and continue with a connection between the interference at $u$ and at $v$. We then analyze the interference at $u$ created in a ball of radius $2d_u$ around $u$, as well as the remaining interference coming from outside that ball. Finally, we transfer all the knowledge we gained for round $r'$ to round $r$ to conclude the proof.

For a node $w \in V$ , let $I_{S_\phi\cup W}(w) = \sum_{x\in S_\phi\cup W} \frac{P}{d(x,w)^\alpha}$, i.e., the amount of interference at node $w$ if all nodes of $S_\phi \cup W$ transmit. For round $r'$, the random variable $X_x^p(w)$ denotes the actual interference at node $w$ coming from a node $x\in S$ (the superscript $^p$ indicates the broadcasting probability of nodes in round $r'$). The total interference at node $w$ is thus $X^p(w) := \sum_{x\in S_\phi\cup W} X_x^p(w)$. If we only want to look at the interference stemming from nodes within a subset $A \subseteq S_\phi$, we use $I_A(w)$ and $X_A^p(w)$ respectively. For round $r$, in which nodes use the broadcasting probability $p/Q$, we use the superscript $^{p/Q}$. Finally, for a set $A\subseteq S_\phi$, we define $\bar{A}:= S_\phi \backslash A$.

For any $w \in S_\phi$, the triangle inequality implies that $d(u,w) \leq d(u,v) + d(v,w) \leq 2d(v,w)$.
By comparing $I_{S'}(u)$ and $I_{S'}(v)$ for an arbitrary set $S' \subseteq S_\phi$ we obtain the following observation:
\begin{equation}
\label{eq:lemma45_1}
I_{S'}(u) \geq 2^{-\alpha}I_{S'}(v).
\end{equation}

Let $u'$ be the farthest neighbor of node $u$ in $\tilde{H}_p^\mu[S_\phi]$. Because $\tilde{H}_p^\mu[S_\phi]$ is an $\varepsilon$-close approximation of $H_p^\mu[S_\phi]$, we know that $\tilde{H}_p^\mu[S_\phi]$ and that this is a subgraph of $H_p^{(1-\varepsilon)\mu}[S_\phi]$. Therefore~\cite{daumfull} now argues that in round $r'$, $u$ receives a message from $u_\phi$ with probability at least $(1-\varepsilon)\mu$. 
In our case, we can only claim that $\tilde{\tilde{H}}_p^\mu[S_\phi]$ is a subgraph of $H_p^{(1-\varepsilon)\mu}[S_\phi]$ in a certain area around $u$. It turns out that this is sufficient, as we argue below. 

Let $A \subseteq S_\phi$ be the set of nodes at distance at most $2d_u$ from $u$. Note that $d(u,u_\phi) = d_u$ and therefore both $u$ and $u_\phi$ are in $A$. In round $r'$, if more than $2^\alpha/\beta = \BO(1)$ nodes $u' \in A$ transmit, then node $u$ cannot receive a message from $u_\phi$. Since node $u$ receives a message from $u_\phi$ with probability at least $(1-\varepsilon)\mu$ in round $r'$, we can conclude that fewer than $2^\alpha/\beta$ nodes transmit with at least the same probability.

We show that this disc of radius $2d_u$ around $u$ is covered by a $\BO(1)$-neighborhood of $u_\phi$ in $\tilde{\tilde{H}}_p^\mu[S_\phi]$. For sake of contradiction assume Lemma~\ref{lem:daum4.5} is not true while  Property 1 of Definition~\ref{def:suc} of a successful epoch at point $i$ is satisfied. Then the communication link between $u_\phi$ and its furthest neighbor in $H_p^\mu[S_\phi]$ could not be $\mu$-reliable, as there are $\omega(1)$ nodes within distance $2d_{u_\phi}$ that are sending with probability $p$ each.
We now bound the interference from nodes outside of $A$. The authors of~\cite{daumfull} prove that $I_{\bar{A}}(u)\leq c\cdot\frac{P}{p\beta d_v^{\alpha}}$, where $c$ is a constant. However, compared to~\cite{daumfull} we need to take interference from nodes $W$ into account, as already pointed out in the proof of Lemma~\ref{lem:daum4.3}, and we modify their proof to derive $I_{\bar{A}\cup W}(u)\leq c'\cdot\frac{P}{p\beta d_v^{\alpha}}$ for some constant $c'$. Using the fact that node $u$ receives a message from node $u_\phi$ with constant probability at least $(1-\varepsilon)\mu$ allows us to upper bound $I_{\bar{A}\cup W}(u)$ and by~(\ref{eq:lemma45_1}) also $I_{\bar{A}\cup W}(v)$. For node $u$ to be able to receive a message from $u_\phi$, two things must hold:\\
(a) $\frac{P}{d_u^\alpha(N + X_{\bar{A}\cup W}^p(u))} \geq \frac{P}{d_u^\alpha(N + X^p(u))} \geq \beta$ and\\
(b) $u_\phi$ transmits and $u$ listens (event $R^{u,u_\phi})$.\\

Due to Lemma~\ref{lem:wronginterference} we know that $I_{W}(u)\leq \left(\frac{\epsapprog}{\Ratio}\right)^{\Theta(1)}$. This implies, that we can transform (a) to $\frac{P}{d_u^\alpha(N + X_{\bar{A}}^p(u))} \leq c'\cdot\frac{P}{p\beta d_v^{\alpha}}$ for some $c'$, when we choose the exponent hidden in $\Theta$-notation to match the choice of constant $c'$. Thus we have for $X_{\bar{A}}^p(u)$ that
\begin{equation}
\label{eq:lemma45_2}
(1-\varepsilon)\mu \leq \mathbb{P}R^{u,u_\phi} \cdot \mathbb{P} \left(X_{\bar{A}}^p(u) \leq \frac{P}{\beta d_u^\alpha}-N\right) \leq p(1-p)\cdot \mathbb{P} \left(X_{\bar{A}}^p(u) \leq \frac{P}{\beta d_u^\alpha}\right).
\end{equation}
Using Lemma B.1, we can therefore bound $X_{\bar{A}}^p(u)$ as
\begin{equation}
\label{eq:lemma45_3}
\mathbb{P}\left(X_{\bar{A}}^p(u) \leq \frac{\mathbb{E}[X_{\bar{A}}^p(u)]}{2}\right) = \mathbb{P}\left(X_{\bar{A}}^p(u) \leq \frac{pI_{\bar{A}}(u)}{2}\right) \leq e^{-\frac{p2^\alpha d_u^\alpha}{8p}\cdot I_{\bar{A}}(u)}.
\end{equation}
For the sake of contradiction, assume that $I_{\bar{A}}(u) > c \cdot \frac{P}{p\beta d_u^\alpha}$ for $c = \text{max}\{2, \frac{16\beta}{2^\alpha} \cdot \ln\frac{p(1-p)}{(1-\varepsilon)\mu}\}$. Combining~(\ref{eq:lemma45_2}) and~(\ref{eq:lemma45_3}), we obtain
\begin{equation*}
\frac{(1-\varepsilon)\mu}{p(1-p)}\overset{(2)}{\leq} \mathbb{P} \left(X_{\bar{A}}^p(u) \leq \frac{P}{\beta d_u^\alpha}\right) \leq \mathbb{P} \left(X_{\bar{A}}^p(u) \leq \frac{cP}{2\beta d_u^\alpha}\right)< \mathbb{P}\left(X_{\bar{A}}^p(u) \leq \frac{pI_{\bar{A}}(u)}{2}\right) \leq e^{-\frac{2^\alpha c}{16\beta}},
\end{equation*}
which is a contradiction to the definition of $c$. We therefore have $I_{\bar{A}}(u) \leq c\cdot \frac{P}{p\beta d_u^\alpha}$ and $I_{\bar{A}\cup W}(u) \leq c'\cdot \frac{P}{p\beta d_u^\alpha}$.

We now have all tools to show that $v$ receives a message from $u$ in round $r$, with broadcasting probabilities $p/Q$. From the fact that the link $\{u,u_\phi\}\in E[\tilde{H}_p^\mu[S_\phi]]$ is reliable, we have seen that with probability at least $(1-\varepsilon)\mu$ fewer than $\frac{2^\alpha}{\beta}$ nodes in $A$ send in round $r'$. But then in round $r$ with the same probability no more than $\frac{2^\alpha}{\beta Q}$ send within $A$. Markov's inequality shows that $\mathbb{P}\left(X_{\bar{A}\cup W}^{p/Q}(v) < 2\frac{p}{Q}I_{\bar{A}\cup W}(v)\right) \geq 1/2$. Finally, $u$ sends with probability $p/Q$. All those events are independent, thus all of them happen with probability at least $\frac{(1-\varepsilon)\mu p}{2Q} = \Theta(1/Q)$. Let us assume that this is the case. To see that $v$ indeed gets $u$'s message under those conditions, we check whether $SINR(u,v,I) = \frac{Pd(u,v)^{-\alpha}}{N + X_{\bar{A}\cup W}^{p/Q}(v) +  X_{A}^{p/Q}} \geq \beta$:

$$\beta d(u,v)^\alpha(N + X_{\bar{A}\cup W}^{p/Q}(v) +  X_{A}^{p/Q})$$
\begin{eqnarray*}
&\overset{(*)}{\leq} & 
\beta d(u,v)^\alpha N + 2^{\alpha+1}c_\beta P\frac{d(u,v)^\alpha}{d_u^\alpha} + \beta \sum\limits_{w\in A,w\text{ sends}} P\frac{d(u,v)^\alpha}{Qd(w,v)^\alpha}
\\
&\overset{d(u,v)^\alpha \leq Q\frac{d_u^\alpha}{\gamma^\alpha}}{\leq} &
\left(1+ \frac{\rho}{2}\right)^\alpha r_s^\alpha N \beta + \frac{2^{\alpha+1}c_\beta}{\gamma^\alpha}P + \frac{2^\alpha}{Q}P
\\
&\overset{(1+\rho)^\alpha \geq \left(1+\frac{\rho}{2}\right)^\alpha + \alpha\frac{\rho}{2}}{\leq} &
\left(1-\frac{\alpha\rho}{2(1+\rho)^\alpha}\right)(1+\rho)^\alpha r_s^\alpha N\beta + \frac{2^{\alpha+1}c_\beta}{\gamma^\alpha}P + \frac{2^\alpha}{\hat{Q}}P
\\
&\overset{P = N\beta(1+\rho)^\alpha r_s^\alpha}{\leq} &
P + P \left(\frac{2^{\alpha+1}c_\beta}{\gamma^\alpha} + \frac{2^\alpha}{\hat{Q}} - \frac{\alpha\rho}{2(1+\rho)^\alpha}\right) 
\\
&\overset{(**)}{\leq}& 
P
\end{eqnarray*}

Inequality $(*)$ holds due to the assumption that $X_{\bar{A}\cup W}^{p/Q}(v) < 2I_{\bar{A}\cup W}(v)p/Q$ and \eqref{eq:lemma45_2}. Inequality $(**)$ holds for properly chosen $\gamma = \Theta(1)$ and $\hat{Q} = \Theta(2^\alpha) = \BO(\log^\alpha_{\max}R_s)$.
\end{proof} 

\begin{lemma}[Version of Lemma 4.6. of~\cite{daumfull}]\label{applem:daumfull4.6} Assume Property 2 of the Definition~\ref{def:suc} of a successful epoch is satisfied. With probability $1-\epsapprog/3$, either $u_\phi$'s bcast-message reaches $i$ in phase $\phi$, or $d(u_{\phi+1}, i) \leq \Ree\left(1+\phi\frac{\eps}{\log \Ratio}\right)$.
\end{lemma}

\begin{proof}
For completeness and clarity we restate the full proof of~\cite{daumfull} and extend it to our setting. 
Clearly, $d(u_1, i) \leq \Ra$.
Let $\phi$ be any phase. If $d_{u_\phi}\geq Q^{-1/\alpha}d(u_\phi, i)$, then we can apply Lemma~\ref{lem:daum4.5} and we are done, because $u_\phi$ sends for $\BO(Q \log (1/\epsapprog))$ rounds in Line~\ref{line:1}. To see this, we choose the constant hidden in the $\BO$-notation large enough, and derive that with probability
$1-(1-(1/\Theta(Q))^{\BO(Q\cdot \log(1/\epsapprog))}\geq 1-e^{-\log(3/\epsapprog)} =1-\epsapprog/3$ the bcast-message sent by $u_\phi$ reaches $i$ during the execution of Lines~\ref{line:22}--\ref{line:23}. Now let this not be the case and let $u_{\phi+1}$ be the closest neighbor to $i$ in $S_{\phi+1}$. Due to the assumption that Property 2 of Definition~\ref{def:suc} of a successful epoch at point $i$ is satisfied, $S_{\phi+1}$ is $(\phi,i)$-locally maximal in $\tilde{\tilde{H}}_p^\mu[S_{\phi}]$. Using this maximality property of our construction, it is $d(u_{\phi+1}, i)\leq  d(u_\phi, i) + d_{u_\phi}$, and therefore

\begin{eqnarray*}
d(u_{\phi+1}, i)
&\leq&
\left(1 + \frac{\gamma'}{Q^{1/\alpha}}\right)d(u_\phi,i)\leq \Ra\left(1+\phi\frac{\eps}{\log \Ratio}+\frac{2\gamma'}{Q^{1/\alpha}}\right)
\\
&\leq&
\Ra\left(1+\frac{(\phi+1)\eps}{\log \Ratio}\right).
\end{eqnarray*}

The last inequality holds for properly chosen $Q = \Theta\left(\log^{\alpha} \Ra\right)$, $Q\geq\hat{Q}$, proving the Lemma.
\end{proof}

\end{document}